\documentclass[journal,twocolumn,romanappendices,9pt]{IEEEtran}

 \usepackage{amsmath,epsfig,amssymb,algorithm,amsthm,cite,url,mathtools}

\usepackage{algpseudocode,tabularx}
\usepackage{hyperref}
\DeclarePairedDelimiter{\ceil}{\lceil}{\rceil}
\usepackage{mathtools}
\usepackage{lipsum}
\let\norm\undefined 
\DeclarePairedDelimiter\norm{\lVert}{\rVert}
\usepackage[font=scriptsize]{caption}
\newtheorem{theorem}{Theorem}
\newtheorem{lemma}[theorem]{Lemma}
\newtheorem{assumption}{Assumption~A-\kern-0pt}

\newtheorem{corollary}[theorem]{Corollary}
\newtheorem{remark}{Remark}
\newtheorem*{remark*}{Remark}

\usepackage{tikz,pgfplots}
\usetikzlibrary{patterns}
\usepgfplotslibrary{fillbetween}
\pgfplotsset{compat=1.7}
\usetikzlibrary{shapes.geometric, intersections,plotmarks}
\usetikzlibrary{calc}
\usepgfplotslibrary{groupplots}
\DeclareMathOperator*{\argmin}{\arg\!\min}

\newcommand{\asto}{\overset{\rm a.s.}{\longrightarrow}}
 
\usepackage{subfig}
  \usepackage{listings}
\usepackage{color} 
\definecolor{mygreen}{RGB}{28,172,0} 
\definecolor{mylilas}{RGB}{170,55,241}

\title{On the Precise Error Analysis of Support Vector Machines}
\author{Abla Kammoun and Mohamed-Slim Alouini
\thanks{A. Kammoun and M.S. Alouini are with the Computer, Electrical, and Mathematical Sciences and Engineering (CEMSE) Division, KAUST, Thuwal, Makkah Province, Saudi Arabia (e-mail: abla.kammoun@kaust.edu.sa, slim.alouini@kaust.edu.sa)}
}
\begin{document}

\maketitle
\begin{abstract}
This paper investigates the asymptotic behavior of the soft-margin and hard-margin support vector machine (SVM) classifiers for simultaneously high-dimensional and numerous data (large $n$ and large $p$ with $n/p\to\delta$) drawn from a Gaussian mixture distribution. Sharp predictions of the classification error rate  of the hard-margin and soft-margin SVM are provided, as well as asymptotic limits of as such important parameters as the margin and the bias. As a further outcome, the analysis allow for the  identification of the maximum number of training samples that the hard-margin SVM is able to separate. The precise nature of our results allow for an accurate performance comparison of the hard-margin and soft-margin SVM as well as a better understanding of the involved parameters (such as the number of measurements and the margin parameter) on the classification performance. Our analysis, confirmed by a set of numerical experiments, builds upon  the convex Gaussian min-max Theorem, and extends its scope to  new  problems never studied before by this framework.   

\end{abstract}
\section{Introduction}
With the advent of the era of big data, attention is now turned to modern classification problems that  require to solve  non-linear problems involving large and numerous data sets. Large margin classifiers constitute a typical example of these novel classification methods and include as particular cases support vector machines \cite{vapnik}, logistic regression \cite{McCullaghNelder89} and Adaboost \cite{FreundJCSS97}.  The performance of these methods is known to be very sensitive to some design parameters, the setting of which is considered as a critical step, as an inappropriate setting can lead to severe degradation in the performance of the underlying classification technique. To properly set these design parameters, cross validation is the standard approach that has been adopted in the machine learning research.  However, such an approach becomes rather computationally expensive in high dimensional settings, since it involves to design the classifier for each candidate  value of the design parameters. Recently, a new technique based on  large dimensional statistical analyses has been emerged to assist in the design of a set of machine learning algorithms including kernel clustering techniques \cite{couillet-first}, classification \cite{khalil,khalil-MLSP}, and regression. It is based on determining sharp performance characterizations  that can be assessed based on the foreknowledge of the data statistics or be approximated using training data. The advantages of this new technique are two-fold. First, it allows easy prediction of the performances for any set of design parameters, avoiding the  prohibitively high computational complexity of the cross-validation approach and paving the way towards optimal setting of the design parameters. Second, it is more instrumental to gain a deep understanding of the performances with respect to the data statistics and the different underlying parameters. However, the application of this approach has  been mainly concentrated on methods and algorithms in which the output possesses a closed-form expression,  as algorithms involving implicit formulation are  much less tractable. 

Recently, a line of research works has emerged that studies the performance of high-dimensional regression problems involving non-smooth convex optimization methods. 
The approaches that have thus far used can be classified into three main categories: a leave-one out approach proposed by El Karoui in \cite{elkaroui}, an approximate message passing based approach developed in \cite{Donoho2016} and finally the convex Gaussian min-max theorem (CGMT) based approach initiated by Stojnic \cite{stojnic} and further developed by Thrampoulidis {\it et al} in \cite{thrampoulidis-IT}. Out of these approaches, the CGMT has three main advantages: 1) it is  the  most direct approach in that it requires very little preliminary work; 2) it requires minimal  assumptions as compared to the other approaches ; 3) it allows for a unified approach to handle  generic problems, with less requirements on the structure of the objective function in the underlying optimization problem.

The present work focuses on the use of the CGMT for the asymptotic analysis of the popular support vector machines (SVM) \cite{Bishop}. Previous works considering the analysis of the SVM have been based on non-rigorous calculations using either the replica method \cite{Huang17} or a leave-one out based approach \cite{Mai}. It should be noted that the optimization problem involved in SVM could not be written as an instance of the general high-dimensional regression problem considered in \cite{thrampoulidis-IT}. Moreover, it raises several new challenges towards the direct application of the CGMT.  Although the considered setting assumes isotropically distributed Gaussian data which is less general than that of previous works in  \cite{Huang17} and   \cite{Mai}, the present work is to the best of our knowledge the first one that provides rigorous proofs for the analysis of SVM.  More specifically, our contributions lie on two levels. At the practical level, we establish a phase transition for the behavior of the hard-margin SVM, which shows that asymptotically the number of samples should be below a certain threshold for the hard-margin SVM to be feasible. If such a condition is satisfied, we provide asymptotic limits for the margin and the classification performance. Similarly, we provide sharp characterizations of the performance of the soft-margin SVM. Our analysis is confirmed by a set of numerical results which shows a good match even for finite dimensions. On the theoretical level, the present work makes a significant progress in contributing to the development of the CGMT framework. The consideration of SVM exemplifies a difficult situation in which the use of the CGMT poses several technical challenges, a list of which is presented in section \ref{sec:diff}. Our work develops new tools to handle these technicalities, which we believe will be key to extending the scope of CGMT to the asymptotic behavior of optimization-based classifiers in general.

The rest of the paper is organized as follows. Section \ref{sec:model} introduces the hard-margin and soft-margin SVM as well as the considered statistical model. Section \ref{sec:main_results} presents our main results along with  some important  implications. Numerical  illustrations are provided in section \ref{sec:numerical}. Finally, section \ref{sec:technical} is devoted to the development of the technical proofs. 


\section{Problem formulation}
\label{sec:model}
Assume we are at our disposition a set of training observations $\left\{({\bf x}_i, y_i)\right\}_{i=1}^n$ where for each ${\bf x}_i\in\mathbb{R}^{p}$ a given input vector, $y_i=1$ if ${\bf x}_i$ belongs to class $\mathcal{C}_1$ or ${y_i}=-1$ is ${\bf x}_i$ belongs to class $\mathcal{C}_{0}$.  We assume that there are $n_0$ observations in class $\mathcal{C}_0$ and $n_1$ observations in class $\mathcal{C}_1$, both of them are drawn from Gaussian distribution with different means and common covariance matrix equal to $\sigma^{2}{\bf  I}_p$. More specifically:
$$
i\in \mathcal{C}_k \ \ \Leftrightarrow \ \  {\bf x}_i\sim \mathcal{N}(\boldsymbol{\mu}_k,\sigma^2{\bf I}_p) \ \ . 
$$ 
As suggested by several previous studies \cite{houssem,liao,Elkhalil}, the performance of a classifier shall depend on the difference between the mean vectors $\boldsymbol{\mu}\triangleq\boldsymbol{\mu}_1-\boldsymbol{\mu}_0$ and the covariance matrix associated with each class, which is in our case equal to $\sigma^2{\bf I}_p$. 
Since the classification problem would not change upon a translation of all observations with the same vector, we will assume for technical reasons that $\boldsymbol{\mu}_0=-\boldsymbol{\mu}$ and $\boldsymbol{\mu}_1=\boldsymbol{\mu}$ without any loss of generality. Such an assumption has been made in the asymptotic analysis of SVM \cite{Huang17}
\subsection{Hard Margin SVM}
Given a set of training data  $\left\{({\bf x}_i, y_i)\right\}_{i=1}^n$ that is linearly separable, hard-margin SVM seeks for the affine plane that separates both classes with the maximum margin \cite{vapnik}. This amounts to solving the following optimization problem:
\begin{equation}\begin{array}{cc}
	\Phi^{(n)}\triangleq\displaystyle\min_{{\bf w},b} \ \ \norm{{\bf w}}_2^2\\
s.t. \ \  \forall i \in\left\{1,\dots,n\right\}, \ \ y_i\left({\bf w}^{T}{\bf x}_i+b\right)\geq 1 \ \ .
\end{array}
\label{P_hard_margin}
\end{equation}
Let $\hat{\bf w}_{H}$ and $\hat{b}_{H}$ solve the above problem, then the hard-margin classifier applied to an unseen observation ${\bf x}$  is given by ${L}_H({\bf x})={\rm sign}\left(\hat{\bf w}_{H}^{T}{\bf x}+\hat{b}_H\right)$. 
\subsection{Soft Margin SVM}
If the data are not linearly separable, the hard-margin optimization problem does not have a finite solution. Under such settings, one alternative is to use the soft-margin SVM which by construction tolerates that some training data are mis-classified but pays the cost of each misclassified observation by adding an upper bound on the number of the misclassified training observations. More formally, the soft-margin SVM is equivalent to solving the following optimization problem:
  \begin{equation}\begin{array}{cc}
\tilde{\Phi}\triangleq\displaystyle\min_{{\bf w},b,\left\{\xi_i\right\}_{i=1}^n} \norm{{\bf w}}_2^2 +\frac{\tilde{\tau}}{p}\sum_{i=1}^n \xi_i\\
s.t. \ \  \forall i \in\left\{1,\dots,n\right\}, \ \ y_i\left({\bf w}^{T}{\bf x}_i+b\right)\geq 1-\xi_i, \ \  \xi_i\geq 0\ \ ,
\end{array}
\label{eq:soft_margin}
\end{equation}
where $\tilde{\tau}$ is a strictly positive scalar, set beforehand by the user, and aims to make a trade-off between maximizing the margin and minimizing the training error. In this respect, a small $\tilde{\tau}$ tends to put more emphasize on the margin while a larger $\tau$ penalize the training error.   Let $\hat{\bf w}_S$ and $\hat{b}_{S}$ solve the above problem, then the soft-margin SVM classifier applied to an unseen observation ${\bf x}$ is given by $L_{S}({\bf x})={\rm sign}\left(\hat{\bf w}_{S}^{T}{\bf x}+\hat{b}_{S}\right)$. 
\section{Main results}
\label{sec:main_results}
The study of the statistical behavior of the hard-margin and soft-margin SVM is carried out under the following asymptotic regime:
\begin{assumption}
We shall assume the following
\begin{itemize}
\item $n$, $n_0$, $n_1$ and $p$ grow to infinity with $\frac{n}{p}\to \delta$, $\frac{n_0}{n}\to \pi_0$ and $\frac{n_1}{n}\to \pi_1$. 
\item  $\sigma^2$ is a  fixed strictly positive scalar, while $\norm{\boldsymbol{\mu}}_2\to \mu$.
\item The training samples ${\bf x}_1,\dots,{\bf x}_n$ are independent. Moreover, for $k\in\{0,1\}$, ${\bf x}_i\in\mathcal{C}_k$, if and only if ${\bf x}_i=y_k\boldsymbol{\mu}+{\bf z}_i$ with $y_k=1$ if $k=1$ and $y_k=-1$ if $k=0$. 
\end{itemize}
\label{ass:regime}
\end{assumption}

\subsection{Hard Margin SVM}
\label{sec:hard_margin}
In this section, we analyze the behavior of the hard-margin SVM under Assumption \ref{ass:regime}. 
\begin{theorem}
Let $\eta^\star(\rho)$ be the unique solution in $\eta$ to the following equation:
\begin{equation}
\eta=\frac{\pi_1\int_{\frac{\rho\mu}{\sigma}+\eta}^\infty (x-\frac{\rho \mu}{\sigma})Dx+\pi_0\int_{\frac{\rho\mu}{\sigma}-\eta}^\infty (\frac{\rho\mu}{\sigma}-x)}{\pi_1\int_{\frac{\rho\mu}{\sigma}+\eta}^\infty Dx+\pi_0\int_{\frac{\rho\mu}{\sigma}-\eta}^\infty Dx} \ \ .
\label{eta_star}
\end{equation}
	where $Dx=\frac{dx}{\sqrt{2\pi}}\exp(-\frac{x^2}{2})$. 
Assume that:
\begin{align}
	&\min_{ 0\leq \rho\leq 1}\frac{1}{1-\rho^2}\left(\pi_1\int_{\frac{\rho\mu}{\sigma}+\eta^\star(\rho)}^{\infty}(x-\frac{\rho\mu}{\sigma}-\eta^\star(\rho))^2Dx \right.\nonumber\\
	&\left.+\pi_0\int_{\frac{\rho\mu}{\sigma}-\eta^\star(\rho)}^\infty (x-\frac{\rho\mu}{\sigma}+\eta^\star(\rho))^2Dx\right)>\frac{1}{\delta}
\label{condition}
\end{align}
Then, under Assumption \ref{ass:regime}
$$
\mathbb{P}\left[\Phi=\infty, n \ \  \textnormal{large enough}\right]=1.
$$
\label{th:hardmargin_condition}
\end{theorem}
\begin{proof}
The proof is postponed to Section \ref{tech_hard_margin_cond}. 
\end{proof}
Theorem \ref{th:hardmargin_condition} establishes a phase transition phenomenon for the hard-margin SVM, according to which, the ratio between the number of samples and that of features should be less than a certain threshold for the hard-margin SVM to be capable of linearly separating the data. Equivalently, it can be used to have an idea of the minimum number of training samples that cannot be linearly separated without errors. Assuming that $n$ and $p$ are sufficiently large, if the  number of training samples is greater  than:
\begin{align*}
n>&p\left(\min_{1\leq \rho \leq 1}\frac{1}{1-\rho^2}\left(\pi_1\int_{\frac{\rho\mu}{\sigma}+\eta^\star(\rho)}^\infty(x-\frac{\rho\mu}{\sigma}-\eta^\star(\rho))^2Dx\right.\right.\\
&+\left.\left.\pi_0\int_{\frac{\rho\mu}{\sigma}-\eta^\star(\rho)}^\infty(x-\frac{\rho\mu}{\sigma}+\eta^\star(\rho))^2Dx\right) \right)^{-1}.
\end{align*}
then the hard margin SVM fails to linearly separate the training samples. 
To the best of our knowledge, a similar condition has never been established before, except from some works limited to the treatment of the one-dimensional case \cite{funaya}. The above result does not tell, however, as to when the hard-margin SVM guarantees perfect separation of the training samples.  This constitutes the objective of the following Theorem, which in addition to providing this condition, determines  almost sure limits of the margin, the bias, and the angle between the solution vector $\hat{\bf w}_H$ and vector $\boldsymbol{\mu}$.     
\begin{theorem}
Assume that:
\begin{align}
	&\min_{ 0\leq \rho\leq 1}\frac{1}{1-\rho^2}\left(\pi_0\int_{\frac{\rho\mu}{\sigma}+\eta^\star(\rho)}^{\infty}(x-\frac{\rho\mu}{\sigma}-\eta^\star(\rho))^2Dx\right.\nonumber\\
	&\left.+\pi_1\int_{\frac{\rho\mu}{\sigma}-\eta^\star(\rho)}^\infty (x-\rho\mu+\eta^\star(\rho))^2Dx\right)<\frac{1}{\delta}
\label{eq:cond_hard_margin}
\end{align}
	with $\eta^\star(\rho)$ given in \eqref{eta_star}. Define function ${D}_H$ in \eqref{eq:DH}:
	\begin{figure*}
\begin{align}
D_H(q_0,\rho,\eta)&\triangleq \sqrt{{\delta\pi_1}\int_{\frac{\rho \mu}{\sigma} +\eta-\frac{1}{q_0\sigma}}^\infty (x+\frac{1}{q_0\sigma}-\frac{\rho\mu}{\sigma}-\eta)^2Dx+ {\delta\pi_0}\int_{\frac{\rho\mu}{\sigma}-\frac{1}{q_0\sigma}-\eta}^\infty (x+\frac{1}{q_0\sigma}-\frac{\rho\mu}{\sigma}+\eta)Dx}-\sqrt{1-\rho^2} .
	\label{eq:DH}
\end{align}
\hrule
	\end{figure*}
Let $\beta:\mathbb{R}_{+}\to\mathbb{R}, q_0\mapsto \displaystyle{\min_{\substack{0\leq\rho\leq 1 \\ \eta\in\mathbb{R}}}D_{H}(q_0,\rho,\eta)}$. Then, function $\beta$ has a unique zero $q_0^\star$. Moreover, with probability $1$, for $n$ and $p$ large enough,
$$
\left\|\widehat{\bf w}_H\right\|\to q_0^\star. 
$$
Let $\rho^\star$ and $\eta^\star$ be such that $(\rho^\star,\eta^\star)=\displaystyle{\argmin_{\substack{0\leq \rho\leq 1\\ \eta\in\mathbb{R}}}D_H(q_0^\star,\rho,\eta)}$, then, with probability $1$, 
$$
 \frac{\hat{\bf w}_{H}^{T}\boldsymbol{\mu}}{\norm{\hat{\bf w}_H}\norm{\boldsymbol{\mu}}}\to \rho^\star, \ \ \textnormal{and} \ \ \hat{b}_H\to \eta^\star q_0^\star \sigma \ \ . 
$$
\label{th:convergence_hard}
\end{theorem}
\begin{proof}
The proof is postponed to Section \ref{technical_proof_hard_margin_asym}.
\end{proof}
The combination of the results of Theorem \ref{th:hardmargin_condition} and Theorem \ref{th:convergence_hard} provides a complete picture of the behavior of the hard-margin-SVM. Particularly, it entails from these results that for the hard-margin SVM to lead to perfect linear separation of the training samples, the number of samples should be strictly less than:
\begin{align*}
n< & p\left(\min_{1\leq \rho \leq 1}\frac{1}{1-\rho^2}\left(\pi_1\int_{\frac{\rho\mu}{\sigma}+\eta^\star(\rho)}(x-\frac{\rho\mu}{\sigma}-\eta^\star(\rho))^2Dx\right.\right.\\
&\left.\left.+\pi_0\int_{\frac{\rho\mu}{\sigma}-\eta^\star(\rho)}(x-\frac{\rho\mu}{\sigma}+\eta^\star(\rho))^2Dx\right) \right)^{-1}\ \ .
\end{align*}
In case $n$ is strictly greater than the right-hand side term  then the hard-margin SVM would asymptotically fail, but in case of equality, no conclusion can be drawn. This phase transition phenomenon is illustrated in  Figure \ref{fig:prediction}  $\pi_0=\pi_1=0.5$, which displays the failure and success regions with varying $\frac{\mu}{\sigma}$ and $\delta$. Interestingly, it is noteworthy to mention that as the factor $\frac{\mu}{\sigma}$ increases, the capabilities of the hard-margin SVM get significantly improved, the number of samples that can be linearly separated increases in an exponentially manner. 
\begin{figure}[t]
\centering{
	\begin{tikzpicture}[scale=0.85]
\begin{semilogyaxis}[xlabel=$\frac{\mu}{\sigma}$, ylabel=$\delta$,xmin=0.1,xmax=3,grid=both,ymax=320]
\addplot[name path=f,color=red] coordinates{
(0.100000,2.012767)(0.200000,2.051490)(0.300000,2.117461)(0.400000,2.212917)(0.500000,2.341184)(0.600000,2.506868)(0.700000,2.716147)(0.800000,2.977167)(0.900000,3.300529)(1.000000,3.700023)(1.100000,4.193603)(1.200000,4.804553)(1.300000,5.563407)(1.400000,6.510113)(1.500000,7.697465)(1.600000,9.195418)(1.700000,11.097548)(1.800000,13.529333)(1.900000,16.660553)(2.000000,20.722920)(2.100000,26.033251)(2.200000,33.030532)(2.300000,42.324067)(2.400000,54.769454)(2.500000,71.574029)(2.600000,94.456163)(2.700000,125.879762)(2.800000,169.404105)(2.900000,230.218260)(3.000000,315.921855)};
\path[name path=axis] (axis cs:0.1,320) -- (axis cs:3,320);
 \addplot[pattern=north east lines,gray!50] fill between[of=f and axis];
\node[] at (axis cs: 1,50){\parbox[c]{0.15\textwidth}{\small Failure of hard-margin SVM}};
\node[] at (axis cs: 2.3,6){\parbox[c]{0.15\textwidth}{\small Success of hard-margin SVM}};
\end{semilogyaxis}
\end{tikzpicture}}
\caption{Theoretical predictions of the regions of failure and success of hard-margin SVM when $\pi_0=\pi_1=0.5$. }
\label{fig:prediction}
\end{figure}
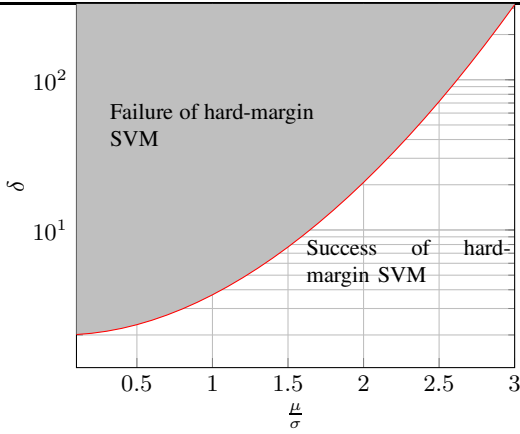
\begin{corollary}
	Let $L_H({\bf x})=\hat{\bf w}_H^{T}{\bf x}+\hat{b}_H$ be the hard-margin classifier, where $\widehat{\bf w}_H$ and $\hat{b}_H$ are solutions to \eqref{P_hard_margin}. 
	Under the assymptotic regime defined in Assumption \ref{ass:regime}, the  classification error rate associated with class $\mathcal{C}_0$ and $\mathcal{C}_1$ converges to:
	\begin{align*}
		\mathbb{P}\left[L_H({\bf x})>0| {\bf x}\in\mathcal{C}_0\right]\to Q(\frac{\rho^\star\mu}{\sigma}+\eta^\star)\\
		\mathbb{P}\left[L_H({\bf x})<0| {\bf x}\in\mathcal{C}_1\right]\to Q(\frac{\rho^\star\mu}{\sigma}-\eta^\star),
	\end{align*}
	where $Q(x)=\frac{1}{\sqrt{2\pi}}\int_{x}^\infty \exp(-\frac{t^2}{2})dt.$
\label{cor:hard_margin}
Let $\varepsilon$ denote the total classification error rate of the hard-margin SVM. It thus converges  to $\varepsilon^\star=\pi_0Q(\frac{\rho^\star\mu}{\sigma}+\eta^\star)+\pi_1Q(\frac{\rho^\star\mu}{\sigma}-\eta^\star)$. 
\end{corollary}
Some important remarks that can be drawn from Corollary \ref{cor:hard_margin} are in order. It is important to note that the  classification error rate depends on the bias and the alignment between $\boldsymbol{\mu}$ and $\hat{\bf w}_{H}$, capitalized by the quantity $\rho^\star$. Obviously, both quantities depend on the margin, but this dependence is not explicit in the asymptotic limits. In our case, the optimal Bayes separating hyperplane has direction aligned with $\mu$, hence $\rho^\star$ also represents the angle between the direction of SVM separating hyperplane and the Bayes optimal separating hyperplane.  
Finally,  it is important to note that the classification error rate is not the same for both classes, unless $\pi_0=\pi_1$ in which case it is easy to see that $\eta^\star=0$. Moreover, if $\pi_1>\pi_0$, it is easy to prove that $\eta^\star>0$. Hence, it is the class  with a higher number of training data  that presents the lowest misclassification error rate. 
\subsection{Soft Margin SVM}
The following Theorem characterizes the asymptotic behavior of the solution to the soft-margin SVM under the asymptotic regime defined in Assumption~\ref{ass:regime}. 
\begin{theorem}
\label{th:soft_margin}
 Let the map ${R}_{\tilde{\tau}}(x,\rho,\eta,\xi):\mathbb{R}_{>0}\times[0,1]\times \mathbb{R}\times \mathbb{R}_{>0}\to\mathbb{R}$ be defined  as:
\begin{align}
&{R}_{\tilde{\tau}}(x,\rho,\eta,\xi):={\tilde{\tau}\pi_1\delta}\int_{\frac{\tilde{\tau}}{\xi}+\eta+\frac{\rho\mu}{\sigma}-\frac{x}{\sigma}}^{\infty}\left(t+\frac{x}{\sigma}-\frac{\rho\mu}{\sigma}-\eta-\frac{\tilde{\tau}}{2\xi}\right) Dt\nonumber\\
&+{\tilde{\tau}\pi_0\delta}\int_{\frac{\tilde{\tau}}{\xi}-\eta+\frac{\rho\mu}{\sigma}-\frac{x}{\sigma}}^{\infty}\left(t+\frac{x}{\sigma}-\frac{\rho\mu}{\sigma}+\eta-\frac{\tilde{\tau}}{2\xi}\right)Dt\nonumber\\
&+\frac{\xi\pi_1\delta}{2}\int_{-\frac{x}{\sigma}+\frac{\rho\mu}{\sigma}+\eta}^{-\frac{x}{\sigma}+\frac{\rho\mu}{\sigma}+\eta+\frac{\tilde{\tau}}{\xi}}\left(t+\frac{x}{\sigma}-\frac{\rho\mu}{\sigma}-\eta\right)^2Dt\nonumber\\
	&+\frac{\xi\pi_0\delta}{2}\int_{-\frac{x}{\sigma}+\frac{\rho\mu}{\sigma}-\eta}^{-\frac{x}{\sigma}+\frac{\rho\mu}{\sigma}-\eta+\frac{\tilde{\tau}}{\xi}}\left(t+\frac{x}{\sigma}-\frac{\rho\mu}{\sigma}+\eta\right)^2Dt-\frac{\xi}{2}(1-\rho^2)\ \  ,
\end{align}
Define $\mathcal{D}_{S,\tilde{\tau}}(q_0,\rho,\eta,\xi)$:$\mathbb{R}_{>0}\times[0,1]\times \mathbb{R}\times \mathbb{R}_{>0}\to\mathbb{R}$ as:
$$
\mathcal{D}_{S,\tilde{\tau}}(q_0,\rho,\eta,\xi):=q_0^2+q_0{R}_{\tilde{\tau}}(\frac{1}{q_0},\rho,\eta,\xi).
$$
Then, the following convex-concave minimax optimization problem
\begin{equation}
\inf_{q_0> 0}\inf_{\eta\in\mathbb{R}} \min_{0\leq \rho\leq 1}\sup_{\xi> 0} \mathcal{D}_{S,\tilde{\tau}}(q_0,\rho,\eta,\xi)
\label{min_max_soft_margin}
\end{equation}
admits a unique solution $(q_0^\star,\rho^\star,\eta^\star)$. Moreover, with probability 1, the following convergences hold true:
\begin{align*}
&\norm{\hat{\bf w}_S}_2\asto q_0^\star, \ \ \frac{\hat{\bf w}_{S}^{T}\boldsymbol{\mu}}{\norm{\hat{\bf w}_{S}}_2\norm{\boldsymbol{\mu}}_2}\asto \rho^\star\ \ \textnormal{and} \ \ \hat{b}_S\asto \eta^\star q_0^\star \sigma .
\end{align*}
\end{theorem}
\begin{corollary}[Misclassification error rate]
Let ${L}_{S}({\bf x})=\widehat{\bf w}_S^{T}{\bf x}+\hat{b}_S$ be the soft margin SVM classifier, where $\widehat{\bf w}_{S}$ and $\hat{b}_S$ are solutions to \eqref{eq:soft_margin}. Under the asymptotic regime defined in Assumption \ref{ass:regime}, the classification error rate of the soft-margin SVM classifier associated with class $\mathcal{C}_0$ and class $\mathcal{C}_1$ converge to:
\begin{align*}
&\mathbb{P}\left[\hat{L}_S({\bf x})>0| {\bf x}\in\mathcal{C}_0\right]\to Q\left(\frac{\rho^\star\mu}{\sigma}+\eta^\star\right)\\
&\mathbb{P}\left[\hat{L}_S({\bf x})<0| {\bf x}\in\mathcal{C}_1\right]\to Q\left(\frac{\rho^\star\mu}{\sigma}-\eta^\star\right).
\end{align*}
where $\rho^\star$, $\eta^\star$ and $q_0^\star$ are the unique solutions to \eqref{min_max_soft_margin}.
Let $\varepsilon$ denote the total classification error rate of the soft-margin SVM. It thus converges to $\varepsilon^\star=\pi_0Q(\frac{\rho^\star\mu}{\sigma}+\eta^\star)+\pi_1Q(\frac{\rho^\star\mu}{\sigma}-\eta^\star)$.
\label{cor:soft_margin} 
\end{corollary}
\begin{remark}
Similar to the hard margin SVM, in case of balanced classes ($\pi_0=\pi_1=0.5$), it is easy to see that $\eta^\star=0$. This confirms the intuition according to which, for the symmetric case $(\boldsymbol{\mu}_1=-\boldsymbol{\mu}_2$), it is best to separate the data with a hyperplane crossing the origin. Again it is easy to see that if $\pi_1>\pi_0$, $\eta^\star>0$, showing that the class with more training data is the one that presents the best misclassification performance. 
\end{remark}

\section{Numerical results}
\label{sec:numerical}
\subsection{Hard Margin SVM }
Figure \ref{fig:hard_margin_1} illustrates the impact of $\mu$ on the angle between the optimal Bayes separating hyperplane, (in our case aligned with $\boldsymbol{\mu}$) and $\hat{\bf w}_H$ the separating hyperplane of hard-margin SVM, as well as on the margin and the classification error rate. As can be seen, the alignment with   $\boldsymbol{\mu}$ improves rapidly in the small region of $\mu$, for which the inverse of the margin reaches very high values, being in the limit of feasibility of the hard-margin SVM. Moreover, the classification error rate decreases considerably as $\mu$ increases.  
Figure \ref{fig:hard_margin_2} describes the impact of $\delta$. We note that the use of more training samples tend to improve the alignment and at the same time decrease the margin. This does not imply a reduction in the classification error performance. On the contrary, the better alignment results in a higher classification performance, despite the decrease  of the margin value.  
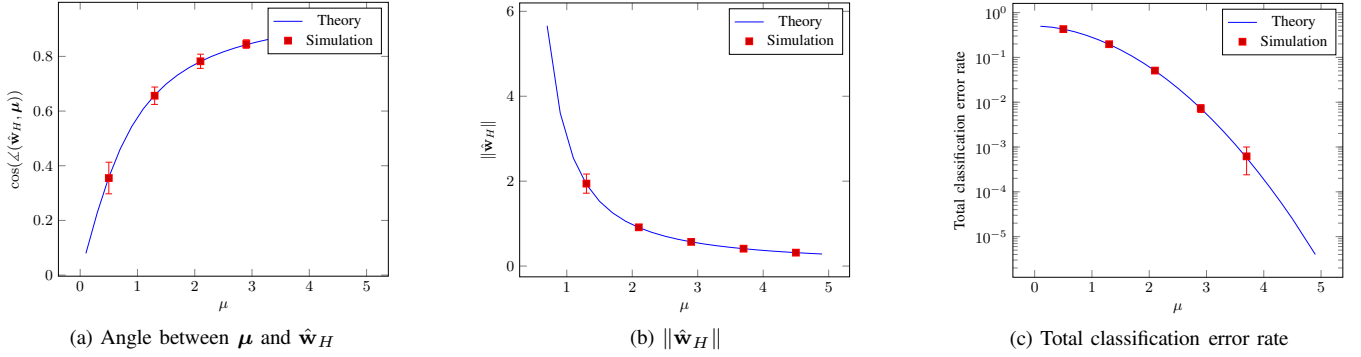
\begin{figure*}[ht]
\captionsetup[subfigure]{justification=centering}
\subfloat[Angle between $\boldsymbol{\mu}$ and $\hat{\bf w}_H$]{
\begin{minipage}[c]{0.3\textwidth}
\begin{tikzpicture}[scale=0.64] 
\begin{axis}[xlabel=$\mu$, ylabel=${\cos(\measuredangle(\hat{\bf w}_H,\boldsymbol{\mu}))}$]
\addplot[blue,no marks] coordinates{
(0.100000,0.079393)(0.300000,0.229386)(0.500000,0.358020)(0.700000,0.461860)(0.900000,0.543620)(1.100000,0.607703)(1.300000,0.658363)(1.500000,0.698891)(1.700000,0.731778)(1.900000,0.758847)(2.100000,0.781418)(2.300000,0.800448)(2.500000,0.816695)(2.700000,0.830669)(2.900000,0.842819)(3.100000,0.853459)(3.300000,0.862846)(3.500000,0.871180)(3.700000,0.878624)(3.900000,0.885312)(4.100000,0.891317)(4.300000,0.896787)(4.500000,0.901757)(4.700000,0.906318)(4.900000,0.910479)
};
\addlegendentry{Theory};
\addplot+[only marks,error bars/.cd,y dir=both,y explicit] coordinates{
(0.500000,0.355081)+-(-0.057828,0.057828)(1.300000,0.656146)+-(-0.031774,0.031774)(2.100000,0.782139)+-(-0.025958,0.025958)(2.900000,0.845220)+-(-0.015475,0.015475)(3.700000,0.877715)+-(-0.014778,0.014778)(4.500000,0.903221)+-(-0.011098,0.011098)
};
\addlegendentry{Simulation};
\end{axis} 
\end{tikzpicture}
\end{minipage}}\hfill
\subfloat[${\norm{\hat{\bf w}_H}}$]{
\begin{minipage}[c]{0.3\textwidth}
\begin{tikzpicture}[scale=0.64] 
\begin{axis}[xlabel=$\mu$, ylabel=${\norm{\hat{\bf w}_H}}$]
\addplot[blue,no marks] coordinates{
(0.700000,5.659959)(0.900000,3.602824)(1.100000,2.544868)(1.300000,1.924015)(1.500000,1.525285)(1.700000,1.251891)(1.900000,1.054920)(2.100000,0.907416)(2.300000,0.793481)(2.500000,0.703223)(2.700000,0.630202)(2.900000,0.570070)(3.100000,0.519800)(3.300000,0.477220)(3.500000,0.440745)(3.700000,0.409186)(3.900000,0.381638)(4.100000,0.357404)(4.300000,0.335935)(4.500000,0.316793)(4.700000,0.299630)(4.900000,0.284161)
};
\addlegendentry{Theory};
\addplot+[only marks,error bars/.cd,y dir=both,y explicit] coordinates{
(1.300000,1.942461)+-(-0.225534,0.225534)(2.100000,0.913011)+-(-0.057097,0.057097)(2.900000,0.569401)+-(-0.027160,0.027160)(3.700000,0.409894)+-(-0.014909,0.014909)(4.500000,0.316385)+-(-0.008516,0.008516)
};
\addlegendentry{Simulation};
\end{axis}
\end{tikzpicture}
\end{minipage}}\hfill
\subfloat[Total classification error rate]{
\begin{minipage}[c]{0.3\textwidth}
\begin{tikzpicture}[scale=0.64]
\begin{axis}[xlabel=$\mu$, ylabel=Total classification error rate, ymode=log]
\addplot[blue,no marks] coordinates{
(0.100000,0.496833)(0.300000,0.472568)(0.500000,0.428965)(0.700000,0.373233)(0.900000,0.312330)(1.100000,0.251916)(1.300000,0.196034)(1.500000,0.147242)(1.700000,0.106746)(1.900000,0.074678)(2.100000,0.050401)(2.300000,0.032808)(2.500000,0.020589)(2.700000,0.012455)(2.900000,0.007259)(3.100000,0.004076)(3.300000,0.002204)(3.500000,0.001148)(3.700000,0.000575)(3.900000,0.000277)(4.100000,0.000129)(4.300000,0.000058)(4.500000,0.000025)(4.700000,0.000010)(4.900000,0.000004)
};
\addlegendentry{Theory};
\addplot+[only marks,error bars/.cd,y dir=both,y explicit] coordinates{
(0.500000,0.429058)+-(-0.012297,0.012297)(1.300000,0.197248)+-(-0.013070,0.013070)(2.100000,0.050890)+-(-0.006015,0.006015)(2.900000,0.007336)+-(-0.001420,0.001420)(3.700000,0.000620)+-(-0.000380,0.000380)
};
\addlegendentry{Simulation};
\end{axis}
\end{tikzpicture}
\end{minipage}}
\caption{Effect of $\mu$ on hard-margin SVM performances, when $\delta=2$, $\pi_0=\pi_1=0.5$, $p=200$ and $\sigma=1$. The solid blue line corresponds to $\rho^\star$, $q_0^\star$ and $\varepsilon^\star$ as defined in Theorem \ref{th:convergence_hard} and Corollary \ref{cor:hard_margin}, while the squares and bars represent the mean and standard deviation of $\cos(\measuredangle(\boldsymbol{\mu},\hat{\bf w}_H)$, $\norm{\hat{\bf w}_H}_2$ and $\varepsilon$ based on $100$ simulated data sets. }
\label{fig:hard_margin_1}
\end{figure*}
\begin{figure*}[ht]
\captionsetup[subfigure]{justification=centering}
\subfloat[Angle between $\boldsymbol{\mu}$ and $\hat{\bf w}_H$]{
\begin{minipage}[c]{0.3\textwidth}
\begin{tikzpicture}[scale=0.64] 
\begin{axis}[xlabel=$\delta$, ylabel=${\cos(\measuredangle(\hat{\bf w}_H,\boldsymbol{\mu}))}$]
\addplot[blue,no marks] coordinates{
(0.100000,0.286106)(0.200000,0.367818)(0.300000,0.414463)(0.400000,0.445502)(0.500000,0.468072)(0.600000,0.485498)(0.700000,0.499455)(0.800000,0.511007)(0.900000,0.520750)(1.000000,0.529173)(1.100000,0.536536)(1.200000,0.543052)(1.300000,0.548844)(1.400000,0.554096)(1.500000,0.558866)(1.600000,0.563228)(1.700000,0.567238)(1.800000,0.570940)(1.900000,0.574377)(2.000000,0.577580)(2.100000,0.580575)(2.200000,0.583385)(2.300000,0.586029)(2.400000,0.588524)(2.500000,0.590884)(2.600000,0.593122)(2.700000,0.595248)(2.800000,0.597271)(2.900000,0.599200)(3.000000,0.601043)
};
\addlegendentry{Theory};
\addplot+[only marks,error bars/.cd,y dir=both,y explicit] coordinates{
(0.100000,0.292169)+-(-0.068311,0.068311)(0.500000,0.471335)+-(-0.046313,0.046313)(0.900000,0.513079)+-(-0.046942,0.046942)(1.300000,0.553376)+-(-0.041451,0.041451)(1.700000,0.569304)+-(-0.037837,0.037837)(2.100000,0.586522)+-(-0.044094,0.044094)(2.500000,0.588257)+-(-0.037694,0.037694)(2.900000,0.599839)+-(-0.040312,0.040312)
};
\addlegendentry{Simulation};
\end{axis} 
\end{tikzpicture}
\end{minipage}}\hfill
\subfloat[${\norm{\hat{\bf w}_H}}$]{
\begin{minipage}[c]{0.3\textwidth}
\begin{tikzpicture}[scale=0.64] 
\begin{axis}[xlabel=$\delta$, ylabel=${\norm{\hat{\bf w}_H}}$]
\addplot[blue,no marks] coordinates{
(0.100000,0.317814)(0.200000,0.455894)(0.300000,0.570062)(0.400000,0.674977)(0.500000,0.776231)(0.600000,0.876761)(0.700000,0.978453)(0.800000,1.082711)(0.900000,1.190710)(1.000000,1.303535)(1.100000,1.422261)(1.200000,1.548003)(1.300000,1.681977)(1.400000,1.825536)(1.500000,1.980226)(1.600000,2.147846)(1.700000,2.330519)(1.800000,2.530784)(1.900000,2.751724)(2.000000,2.997120)(2.100000,3.271683)(2.200000,3.581366)(2.300000,3.933803)(2.400000,4.338971)(2.500000,4.810150)(2.600000,5.365445)(2.700000,6.030192)(2.800000,6.840974)(2.900000,7.852653)(3.000000,9.151430)
};
\addlegendentry{Theory};
\addplot+[only marks,error bars/.cd,y dir=both,y explicit] coordinates{
(0.100000,0.319913)+-(-0.015755,0.015755)(0.500000,0.779787)+-(-0.048421,0.048421)(0.900000,1.200604)+-(-0.100904,0.100904)(1.300000,1.675963)+-(-0.155774,0.155774)(1.700000,2.303410)+-(-0.252017,0.252017)(2.100000,3.368251)+-(-0.608712,0.608712)(2.500000,5.055271)+-(-1.811203,1.811203)(2.900000,8.641536)+-(-3.940439,3.940439)
};
\addlegendentry{Simulation};
\end{axis}
\end{tikzpicture}
\end{minipage}}\hfill
\subfloat[Total classification error rate]{
\begin{minipage}[c]{0.3\textwidth}
\begin{tikzpicture}[scale=0.64]
\begin{axis}[xlabel=$\delta$, ylabel=Total classification error rate, ymode=log]
\addplot[blue,no marks] coordinates{
(0.100000,0.387398)(0.200000,0.356505)(0.300000,0.339267)(0.400000,0.327979)(0.500000,0.319866)(0.600000,0.313662)(0.700000,0.308730)(0.800000,0.304673)(0.900000,0.301270)(1.000000,0.298343)(1.100000,0.295794)(1.200000,0.293547)(1.300000,0.291556)(1.400000,0.289757)(1.500000,0.288126)(1.600000,0.286640)(1.700000,0.285276)(1.800000,0.284020)(1.900000,0.282856)(2.000000,0.281774)(2.100000,0.280764)(2.200000,0.279817)(2.300000,0.278928)(2.400000,0.278090)(2.500000,0.277299)(2.600000,0.276550)(2.700000,0.275839)(2.800000,0.275163)(2.900000,0.274520)(3.000000,0.273906)
};
\addlegendentry{Theory};
\addplot+[only marks,error bars/.cd,y dir=both,y explicit] coordinates{
(0.100000,0.388054)+-(-0.027462,0.027462)(0.500000,0.321232)+-(-0.017746,0.017746)(0.900000,0.305498)+-(-0.017779,0.017779)(1.300000,0.290670)+-(-0.015193,0.015193)(1.700000,0.285058)+-(-0.014186,0.014186)(2.100000,0.279372)+-(-0.016972,0.016972)(2.500000,0.278726)+-(-0.014126,0.014126)(2.900000,0.273488)+-(-0.014389,0.014389)
};
\addlegendentry{Simulation};
\end{axis}
\end{tikzpicture}
\end{minipage}}
\caption{Effect of $\delta$ on hard-margin SVM performances, when $\mu=1$, $\pi_0=\pi_1=0.5$, $p=200$ and $\sigma=1$. The solid blue line corresponds to $\rho^\star$, $q_0^\star$ and $\varepsilon^\star$ as defined in Theorem \ref{th:convergence_hard} and Corollary \ref{cor:hard_margin}, while the squares and bars represent the mean and standard deviation of $\cos(\measuredangle(\boldsymbol{\mu},\hat{\bf w}_H)$, $\norm{\hat{\bf w}_H}_2$ and $\varepsilon$ based on $100$ simulated data sets. }
\label{fig:hard_margin_2}
\end{figure*}
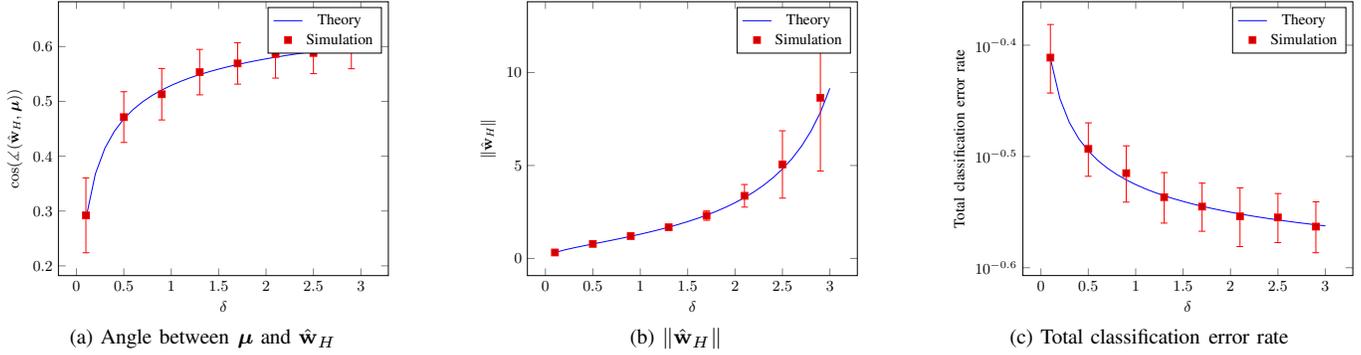

\subsection{Soft Margin SVM}
Figure \ref{fig:soft_margin} investigates the impact of $\mu$ on the angle between the optimal Bayes separating hyperplane aligned with $\boldsymbol{\mu}$ and $\hat{\bf w}_S$ the separating hyperplane of SVM, as well as on the inverse of the margin and the classification error rate. It shows that the alignment significantly improves as $\mu$ increases fast when $\mu<2$. The increase then becomes less important for high $\mu$. We also note that curiously the margin tends to decrease in the range of small $\mu$. This can be explained by the fact that in this region the alignment with the mean vector $\boldsymbol{\mu}$ is weak, causing the margin to decrease when $\mu$ is small. In the region of large $\mu$, the margin increases rapidly ($\norm{\hat{\bf w}}_2$ decreases). 

Figure \ref{fig:soft_margin_delta} investigates the impact of the number of samples on the classification performances. As expected, as more training data are used, a better alignment with the mean vector $\boldsymbol{\mu}$ is noted. However, this results also in a decrease in the margin which does not hopefully translates into a loss in classification performances, these latter being determined by only how good is the alignment with  $\boldsymbol{\mu}$.

Finally, we investigate in \ref{fig:soft_margin_tau} the impact of $\tau$ on the performances. As seen, the alignment with $\boldsymbol{\mu}$ and the margin decrease significantly when $\tau$ is greater than a certain threshold value, suggesting to use smallest values of $\tau$. Such an observation is in agreement with the simulations of \cite{Huang17}, where it was suggested to use the threshold value  since using too tiny values for $\tau$ is known to pose numerical difficulties in solving the optimization problem.

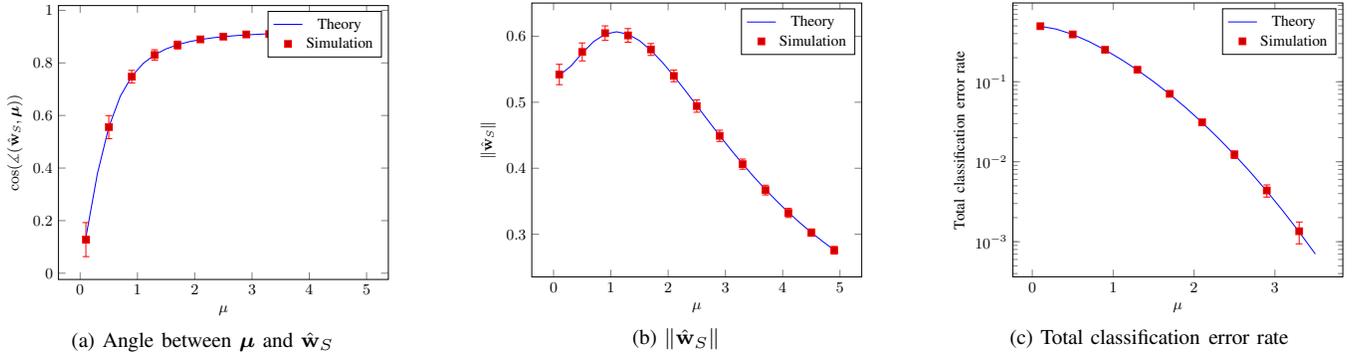
\begin{figure*}[ht]
\captionsetup[subfigure]{justification=centering}
\subfloat[Angle between $\boldsymbol{\mu}$ and $\hat{\bf w}_S$]{
\begin{minipage}[c]{0.3\textwidth}
\begin{tikzpicture}[scale=0.64] 
\begin{axis}[xlabel=$\mu$, ylabel=${\cos(\measuredangle(\hat{\bf w}_S,\boldsymbol{\mu}))}$]
\addplot[blue,no marks] coordinates{
(0.100000,0.135963)(0.300000,0.377844)(0.500000,0.555454)(0.700000,0.673082)(0.900000,0.748750)(1.100000,0.797941)(1.300000,0.830743)(1.500000,0.853235)(1.700000,0.869068)(1.900000,0.880469)(2.100000,0.888840)(2.300000,0.895096)(2.500000,0.899852)(2.700000,0.903539)(2.900000,0.906458)(3.100000,0.908830)(3.300000,0.910819)(3.500000,0.912548)(3.700000,0.914110)(3.900000,0.915572)(4.100000,0.916986)(4.300000,0.918392)(4.500000,0.919815)(4.700000,0.921273)(4.900000,0.922776)
};
\addlegendentry{Theory};
\addplot+[only marks,error bars/.cd,y dir=both,y explicit] coordinates{
(0.100000,0.127772)+-(-0.064730,0.064730)(0.500000,0.555762)+-(-0.043863,0.043863)(0.900000,0.747662)+-(-0.023757,0.023757)(1.300000,0.830099)+-(-0.020030,0.020030)(1.700000,0.868012)+-(-0.014900,0.014900)(2.100000,0.889516)+-(-0.009208,0.009208)(2.500000,0.899936)+-(-0.008778,0.008778)(2.900000,0.908229)+-(-0.009037,0.009037)(3.300000,0.909839)+-(-0.008482,0.008482)(3.700000,0.914640)+-(-0.008125,0.008125)(4.100000,0.917212)+-(-0.006883,0.006883)(4.500000,0.919855)+-(-0.007405,0.007405)(4.900000,0.922943)+-(-0.008281,0.008281)
};
\addlegendentry{Simulation};
\end{axis} 
\end{tikzpicture}
\end{minipage}}\hfill
\subfloat[${\norm{\hat{\bf w}_S}}$]{
\begin{minipage}[c]{0.3\textwidth}
\begin{tikzpicture}[scale=0.64] 
\begin{axis}[xlabel=$\mu$, ylabel=${\norm{\hat{\bf w}_S}}$]
\addplot[blue,no marks] coordinates{
(0.100000,0.540732)(0.300000,0.554618)(0.500000,0.574409)(0.700000,0.592032)(0.900000,0.603233)(1.100000,0.606809)(1.300000,0.603152)(1.500000,0.593344)(1.700000,0.578678)(1.900000,0.560437)(2.100000,0.539772)(2.300000,0.517652)(2.500000,0.494854)(2.700000,0.471972)(2.900000,0.449448)(3.100000,0.427591)(3.300000,0.406606)(3.500000,0.386620)(3.700000,0.367696)(3.900000,0.349858)(4.100000,0.333098)(4.300000,0.317387)(4.500000,0.302683)(4.700000,0.288936)(4.900000,0.276091)};
\addlegendentry{Theory};
\addplot+[only marks,error bars/.cd,y dir=both,y explicit] coordinates{
(0.100000,0.542057)+-(-0.015421,0.015421)(0.500000,0.576317)+-(-0.013595,0.013595)(0.900000,0.604751)+-(-0.011028,0.011028)(1.300000,0.601306)+-(-0.010323,0.010323)(1.700000,0.579985)+-(-0.008994,0.008994)(2.100000,0.539963)+-(-0.008864,0.008864)(2.500000,0.494191)+-(-0.009318,0.009318)(2.900000,0.449139)+-(-0.008544,0.008544)(3.300000,0.406217)+-(-0.007712,0.007712)(3.700000,0.366913)+-(-0.007391,0.007391)(4.100000,0.332364)+-(-0.007014,0.007014)(4.500000,0.302432)+-(-0.005633,0.005633)(4.900000,0.275671)+-(-0.005795,0.005795)
};
\addlegendentry{Simulation};
\end{axis}
\end{tikzpicture}
\end{minipage}}\hfill
\subfloat[Total classification error rate]{
\begin{minipage}[c]{0.3\textwidth}
\begin{tikzpicture}[scale=0.64]
\begin{axis}[xlabel=$\mu$, ylabel=Total classification error rate, ymode=log]
\addplot[blue,no marks] coordinates{
(0.100000,0.494576)(0.300000,0.454875)(0.500000,0.390611)(0.700000,0.318764)(0.900000,0.250195)(1.100000,0.190044)(1.300000,0.140079)(1.500000,0.100298)(1.700000,0.069782)(1.900000,0.047174)(2.100000,0.030981)(2.300000,0.019761)(2.500000,0.012236)(2.700000,0.007353)(2.900000,0.004285)(3.100000,0.002421)(3.300000,0.001325)(3.500000,0.000702)};
\addlegendentry{Theory};
\addplot+[only marks,error bars/.cd,y dir=both,y explicit] coordinates{
(0.100000,0.495752)+-(-0.005544,0.005544)(0.500000,0.391237)+-(-0.009954,0.009954)(0.900000,0.251228)+-(-0.007699,0.007699)(1.300000,0.141395)+-(-0.007083,0.007083)(1.700000,0.070668)+-(-0.004143,0.004143)(2.100000,0.031154)+-(-0.002209,0.002209)(2.500000,0.012306)+-(-0.001323,0.001323)(2.900000,0.004374)+-(-0.000760,0.000760)(3.300000,0.001353)+-(-0.000412,0.000412)};
\addlegendentry{Simulation};
\end{axis}
\end{tikzpicture}
\end{minipage}}
\caption{Effect of $\mu$ when when $\tau=2$, $\delta=2$, $p=200$, $\sigma=1$, $\pi_1=\pi_0=0.5$. The solid blue line corresponds to $\rho^\star$, $q_0^\star$ and $\varepsilon^\star$ as defined in Theorem \ref{th:soft_margin}, while the squares and bars represent the mean and standard deviation of $\cos(\measuredangle(\boldsymbol{\mu},\hat{\bf w}_S)$, $\norm{\hat{\bf w}_S}_2$ and $\varepsilon$ based on $100$ simulated data sets.}
\label{fig:soft_margin}
\end{figure*}

\begin{figure*}
\captionsetup[subfigure]{justification=centering}
\subfloat[Angle between $\boldsymbol{\mu}$ and $\hat{\bf w}_S$]{
\begin{minipage}[c]{0.3\textwidth}
\begin{tikzpicture}[scale=0.64] 
\begin{axis}[xlabel=$\delta$, ylabel=${\cos(\measuredangle(\hat{\bf w}_S,\boldsymbol{\mu}))}$]
\addplot[blue,no marks] coordinates{
(0.500000,0.536832)(1.000000,0.645639)(1.500000,0.705460)(2.000000,0.745151)(2.500000,0.774032)(3.000000,0.796261)(3.500000,0.814035)(4.000000,0.828648)(4.500000,0.840920)(5.000000,0.851391)};
\addlegendentry{Theory};
\addplot+[only marks,error bars/.cd,y dir=both,y explicit] coordinates{
(0.500000,0.535724)+-(-0.043998,0.043998)(1.000000,0.645911)+-(-0.034209,0.034209)(1.500000,0.706376)+-(-0.028138,0.028138)(2.000000,0.749615)+-(-0.024462,0.024462)(2.500000,0.774751)+-(-0.023879,0.023879)(3.000000,0.798500)+-(-0.019165,0.019165)(3.500000,0.812315)+-(-0.017253,0.017253)(4.000000,0.829047)+-(-0.014283,0.014283)(4.500000,0.840631)+-(-0.016247,0.016247)(5.000000,0.850802)+-(-0.015865,0.015865)
};
\addlegendentry{Simulation}
\end{axis}
\end{tikzpicture}
\end{minipage}}\hfill
\subfloat[${\norm{\hat{\bf w}_S}}$]{
\begin{minipage}[c]{0.3\textwidth}
\begin{tikzpicture}[scale=0.64] 
\begin{axis}[xlabel=$\delta$, ylabel=${\norm{\hat{\bf w}_S}}$]
\addplot[blue,no marks] coordinates{
(0.500000,0.519515)(1.000000,0.648842)(1.500000,0.728905)(2.000000,0.786883)(2.500000,0.832170)(3.000000,0.869184)(3.500000,0.900368)(4.000000,0.927221)(4.500000,0.950730)(5.000000,0.971587)
};
\addplot+[only marks,error bars/.cd,y dir=both,y explicit] coordinates{
(0.500000,0.520970)+-(-0.011723,0.011723)(1.000000,0.647284)+-(-0.013539,0.013539)(1.500000,0.728338)+-(-0.015010,0.015010)(2.000000,0.785275)+-(-0.017545,0.017545)(2.500000,0.833578)+-(-0.016254,0.016254)(3.000000,0.868152)+-(-0.016594,0.016594)(3.500000,0.899201)+-(-0.018513,0.018513)(4.000000,0.926154)+-(-0.019137,0.019137)(4.500000,0.946909)+-(-0.016327,0.016327)(5.000000,0.973289)+-(-0.018476,0.018476)
};
\end{axis}
\end{tikzpicture}
\end{minipage}}\hfill
\subfloat[Total classification error rate]{
\begin{minipage}[c]{0.3\textwidth}
\begin{tikzpicture}[scale=0.64]
\begin{axis}[xlabel=$\delta$, ylabel=Total classification error rate]
\addplot[blue,no marks] coordinates{
(0.500000,0.295692)(1.000000,0.259256)(1.500000,0.240262)(2.000000,0.228090)(2.500000,0.219456)(3.000000,0.212940)(3.500000,0.207813)(4.000000,0.203652)(4.500000,0.200196)(5.000000,0.197276)
};
\addplot+[only marks,error bars/.cd,y dir=both,y explicit] coordinates{
(0.500000,0.296155)+-(-0.015549,0.015549)(1.000000,0.259971)+-(-0.012631,0.012631)(1.500000,0.240475)+-(-0.009477,0.009477)(2.000000,0.227144)+-(-0.007828,0.007828)(2.500000,0.219108)+-(-0.008964,0.008964)(3.000000,0.212145)+-(-0.006607,0.006607)(3.500000,0.208775)+-(-0.006367,0.006367)(4.000000,0.204582)+-(-0.005460,0.005460)(4.500000,0.200653)+-(-0.006218,0.006218)(5.000000,0.197336)+-(-0.006211,0.006211)
};
\end{axis}
\end{tikzpicture}
\end{minipage}}
\caption{Effect of $\delta$  when $\tau=2$, $\mu=1$, $p=200$, $\sigma=1$, $\pi_1=\pi_0=0.5$. The solid blue line corresponds to $\rho^\star$, $q_0^\star$ and $\varepsilon^\star$ as defined in Theorem \ref{th:soft_margin}, while the squares and bars represent the mean and standard deviation of $\cos(\measuredangle(\boldsymbol{\mu},\hat{\bf w}_S)$, $\norm{\hat{\bf w}_S}_2$ and $\varepsilon$ based on $100$ simulated data sets.}
\label{fig:soft_margin_delta}
\end{figure*}

\begin{figure*}
\captionsetup[subfigure]{justification=centering}
\subfloat[Angle between $\boldsymbol{\mu}$ and $\hat{\bf w}_S$]{
\begin{minipage}[c]{0.3\textwidth}
\begin{tikzpicture}[scale=0.64] 
\begin{axis}[xlabel=$\tau$, ylabel=${\cos(\measuredangle(\hat{\bf w}_S,\boldsymbol{\mu}))}$,xmode=log]
\addplot[blue,no marks] coordinates{
(0.062500,0.707106)(0.066986,0.707106)(0.071794,0.707106)(0.076947,0.707106)(0.082469,0.707106)(0.088388,0.707106)(0.094732,0.707106)(0.101532,0.707106)(0.108819,0.707106)(0.116629,0.707106)(0.125000,0.707106)(0.133972,0.707106)(0.143587,0.707106)(0.153893,0.707106)(0.164938,0.707106)(0.176777,0.707106)(0.189465,0.707106)(0.203063,0.707106)(0.217638,0.707106)(0.233258,0.707106)(0.250000,0.707106)(0.267943,0.707105)(0.287175,0.707098)(0.307786,0.707078)(0.329877,0.707022)(0.353553,0.706893)(0.378929,0.706640)(0.406126,0.706211)(0.435275,0.705562)(0.466516,0.704670)(0.500000,0.703525)(0.535887,0.702132)(0.574349,0.700502)(0.615572,0.698650)(0.659754,0.696590)(0.707107,0.694340)(0.757858,0.691916)(0.812252,0.689330)(0.870551,0.686596)(0.933033,0.683725)(1.000000,0.680729)(1.071773,0.677617)(1.148698,0.674398)(1.231144,0.671080)(1.319508,0.667670)(1.414214,0.664175)(1.515717,0.660602)(1.624505,0.656957)(1.741101,0.653244)(1.866066,0.649470)(2.000000,0.645639)(2.143547,0.641757)(2.297397,0.637828)(2.462289,0.633856)(2.639016,0.629846)(2.828427,0.625803)(3.031433,0.621730)(3.249010,0.617633)(3.482202,0.613514)(3.732132,0.609381)(4.000000,0.605236)(4.287094,0.601085)(8.574188,0.560717)
};
\addplot+[only marks,error bars/.cd,y dir=both,y explicit] coordinates{
(0.062500,0.706375)+-(-0.030535,0.030535)(0.125000,0.707649)+-(-0.030839,0.030839)(0.250000,0.706600)+-(-0.030893,0.030893)(0.500000,0.704184)+-(-0.030511,0.030511)(1.000000,0.681354)+-(-0.031926,0.031926)(2.000000,0.645165)+-(-0.034643,0.034643)(4.000000,0.604909)+-(-0.037977,0.037977)(8.574188,0.561436)+-(-0.040773,0.040773)
};
\end{axis}
\end{tikzpicture}
\end{minipage}}\hfill
\subfloat[${\norm{\hat{\bf w}_S}}$]{
\begin{minipage}[c]{0.3\textwidth}
\begin{tikzpicture}[scale=0.64] 
\begin{axis}[xlabel=$\tau$, ylabel=${\norm{\hat{\bf w}_S}}$,xmode=log]
\addplot[blue,no marks] coordinates{
(0.062500,0.044194)(0.066986,0.047366)(0.071794,0.050766)(0.076947,0.054409)(0.082469,0.058315)(0.088388,0.062500)(0.094732,0.066986)(0.101532,0.071794)(0.108819,0.076947)(0.116629,0.082469)(0.125000,0.088388)(0.133972,0.094732)(0.143587,0.101532)(0.153893,0.108819)(0.164938,0.116629)(0.176777,0.125000)(0.189465,0.133972)(0.203063,0.143587)(0.217638,0.153893)(0.233258,0.164938)(0.250000,0.176774)(0.267943,0.189452)(0.287175,0.203007)(0.307786,0.217443)(0.329877,0.232689)(0.353553,0.248589)(0.378929,0.264909)(0.406126,0.281401)(0.435275,0.297858)(0.466516,0.314154)(0.500000,0.330229)(0.535887,0.346077)(0.574349,0.361720)(0.615572,0.377195)(0.659754,0.392545)(0.707107,0.407812)(0.757858,0.423039)(0.812252,0.438263)(0.870551,0.453520)(0.933033,0.468842)(1.000000,0.484259)(1.071773,0.499798)(1.148698,0.515485)(1.231144,0.531344)(1.319508,0.547395)(1.414214,0.563660)(1.515717,0.580159)(1.624505,0.596909)(1.741101,0.613929)(1.866066,0.631235)(2.000000,0.648842)(2.143547,0.666766)(2.297397,0.685021)(2.462289,0.703618)(2.639016,0.722571)(2.828427,0.741889)(3.031433,0.761581)(3.249010,0.781656)(3.482202,0.802117)(3.732132,0.822967)(4.000000,0.844206)(4.287094,0.865830)(8.574188,1.097813)
};
\addplot+[only marks,error bars/.cd,y dir=both,y explicit] coordinates{
(0.062500,0.044106)+-(-0.001887,0.001887)(0.125000,0.088228)+-(-0.003881,0.003881)(0.250000,0.176745)+-(-0.007533,0.007533)(0.500000,0.330313)+-(-0.009509,0.009509)(1.000000,0.484738)+-(-0.010136,0.010136)(2.000000,0.648612)+-(-0.013585,0.013585)(4.000000,0.844087)+-(-0.020165,0.020165)(8.574188,1.094519)+-(-0.035541,0.035541)
};
\end{axis}
\end{tikzpicture}
\end{minipage}}\hfill
\subfloat[Total classification error rate]{
\begin{minipage}[c]{0.3\textwidth}
\begin{tikzpicture}[scale=0.64] 
\begin{axis}[xlabel=$\tau$, ylabel=Total classification error rate,xmode=log,ymode=log]
\addplot[blue,no marks] coordinates{
(0.062500,0.239750)(0.066986,0.239750)(0.071794,0.239750)(0.076947,0.239750)(0.082469,0.239750)(0.088388,0.239750)(0.094732,0.239750)(0.101532,0.239750)(0.108819,0.239750)(0.116629,0.239750)(0.125000,0.239750)(0.133972,0.239750)(0.143587,0.239750)(0.153893,0.239750)(0.164938,0.239750)(0.176777,0.239750)(0.189465,0.239750)(0.203063,0.239750)(0.217638,0.239750)(0.233258,0.239750)(0.250000,0.239750)(0.267943,0.239751)(0.287175,0.239753)(0.307786,0.239759)(0.329877,0.239776)(0.353553,0.239816)(0.378929,0.239895)(0.406126,0.240028)(0.435275,0.240230)(0.466516,0.240508)(0.500000,0.240864)(0.535887,0.241299)(0.574349,0.241807)(0.615572,0.242386)(0.659754,0.243030)(0.707107,0.243734)(0.757858,0.244495)(0.812252,0.245308)(0.870551,0.246169)(0.933033,0.247074)(1.000000,0.248021)(1.071773,0.249007)(1.148698,0.250029)(1.231144,0.251085)(1.319508,0.252172)(1.414214,0.253289)(1.515717,0.254434)(1.624505,0.255604)(1.741101,0.256800)(1.866066,0.258017)(2.000000,0.259256)(2.143547,0.260515)(2.297397,0.261793)(2.462289,0.263087)(2.639016,0.264398)(2.828427,0.265722)(3.031433,0.267060)(3.249010,0.268409)(3.482202,0.269768)(3.732132,0.271136)(4.000000,0.272511)(4.287094,0.273892)(8.574188,0.287495)
};
\addplot+[only marks,error bars/.cd,y dir=both,y explicit] coordinates{
(0.062500,0.240654)+-(-0.009597,0.009597)(0.125000,0.240278)+-(-0.009738,0.009738)(0.250000,0.240804)+-(-0.009863,0.009863)(0.500000,0.243471)+-(-0.010400,0.010400)(1.000000,0.248911)+-(-0.010249,0.010249)(2.000000,0.260199)+-(-0.011373,0.011373)(4.000000,0.273318)+-(-0.012690,0.012690)(8.574188,0.287949)+-(-0.014022,0.014022)
};
\end{axis}
\end{tikzpicture}
\end{minipage}
}
\caption{Effect of $\tau$ when $\mu=1$, $p=200$, $\delta=1$, $\sigma=1$ and $\pi_0=\pi_1=0.5$. The solid blue line corresponds to $\rho^\star$, $q_0^\star$ and $\varepsilon^\star$ as defined in Theorem \ref{th:soft_margin}, while the squares and bars represent the mean and standard deviation of $\cos(\measuredangle(\boldsymbol{\mu},\hat{\bf w}_S)$, $\norm{\hat{\bf w}_S}_2$ and $\varepsilon$ based on $2000$ simulated data sets.}
\label{fig:soft_margin_tau}
\end{figure*}
\section{Technical Proofs}
\label{sec:technical}
\subsection{CGMT Framework}
Our technical proofs builds upong the CGMT framework, rooted in the works of Stojnic \cite{stojnic} and further mathematically formulated in the works of Thrampoulidis et al in \cite{thrampoulidis-IT}  and \cite{Thrampoulidis_conf}. The CGMT can be regarded as a generalization of a classical Gaussian comparison dating back to the early works of Gordon in 1988 \cite{Gor88}. This inequality allows to provide a high-probability lower bound of the optimal cost function of any optimization problem that can be written in the following form:
\begin{equation}
\Phi^{(n)}({\bf G}):=\min_{{\bf w}\in\mathcal{S}_{\bf w}}\max_{{\bf u}\in\mathcal{S}_{\bf u}} {\bf u}^{T}{\bf Gw}+\psi({\bf w},{\bf u}),
\label{eq:primary}
\end{equation}
where ${\bf G}\in\mathbb{R}^{n\times p}$ is a standard gaussian matrix, $\mathcal{S}_{\bf w}$ and $\mathcal{S}_{\bf u}$ are two compact sets in $\mathbb{R}^{p}$ and $\mathbb{R}^{n}$ and $\psi:\mathbb{R}^{p}\times \mathbb{R}^n\to\mathbb{R}$ is continuous on $\mathcal{S}_{\bf w}\times \mathcal{S}_{\bf u}$, possibly random but independent of ${\bf G}$. The optimization problem in \eqref{eq:primary} is identified as a primary optimization problem (PO), the asymptotic behavior of which cannot be directly studied in general, due to the coupling between vectors ${\bf w}$ and ${\bf u}$ in the bilinear term. To this end, based on Gaussian comparison inequalities \cite{gordon85}, we associate with it the following optimization problem 
\begin{equation}
\phi^{(n)}({\bf g},{\bf h}):=\min_{{\bf w}\in\mathcal{S}_{\bf w}}\max_{{\bf u}\in\mathcal{S}_{\bf u}} \norm{{\bf w}}_2{\bf g}^{T}{\bf u} +{\norm{\bf u}}_2{\bf h}^{T}{\bf w}+\psi({\bf w},{\bf u}),
\end{equation}
where ${\bf g}\in\mathbb{R}^n$ and ${\bf u}\in\mathbb{R}^{p}$ are standard gaussian vectors. According to Gordon's comparison inequality, for any $c\in\mathbb{R}$, it holds that:
\begin{equation}
\mathbb{P}\left[\Phi^{(n)}({\bf G})<c\right]\leq 2\mathbb{P}\left[\phi^{(n)}({\bf g},{\bf h})<c\right].
\label{ineq_gordon}
\end{equation}
Particularly, if  a high-probability lower bound of the (AO) can be found, then by \eqref{ineq_gordon}, this lower bound translates also into  a high-probability lower bound of the (PO).
This result is remarkable since so far it does not require any assumption on the convexity of function $\psi$ or the sets  $\mathcal{S}_{\bf w}$ and $\mathcal{S}_{\bf u}$, and most importantly it allows to relate the (PO) to a seemingly unrelated (AO) problem which presents the advantage of being in general much easier to analyze than the (PO) problem, as the bilinear term is now decoupled into two independent quantities involving respectively vectors ${\bf g}$ and ${\bf h}$. Combining the Gordon's original result with convexity, it was  shown that this result can be strenghened to a more precise characterization of the asymptotic behavior of the (PO), \cite{StoLASSO,COLT15}. Particularly, if the sets   $\mathcal{S}_{\bf w}$ and 
$\mathcal{S}_{\bf u}$ are additionally convex and $\psi$ is convex-concave on $\mathcal{S}_{\bf w}\times \mathcal{S}_{\bf u}$, then, for any $\kappa\in\mathbb{R}$, and $t>0$,
$$
\mathbb{P}\left[\left|\Phi^{(n)}({\bf G})-\kappa\right|>t\right] \leq 2\mathbb{P}\left[|\phi^{(n)}({\bf g},{\bf h})-\kappa|>t\right].
$$
A direct consequence of this inequality is that if the optimal cost of the (AO) problem converges to $\kappa$ then the optimal cost of the (PO) converges also to the same constant. However, in most cases, the ultimate goal is not to characterize the optimal cost of the PO but rather properties of the minimizer of $\Phi^{(n)}$ which we denote by ${\bf w}_{\boldsymbol{\Phi}}$. Although not directly obvious, this can be related to a question on the evaluation of the optimal cost as shown in the following Theorem.  
\begin{theorem}(CGMT,\cite{symbol_error})
Let $\mathcal{S}$ be  an arbitrary open subset of $\mathcal{S}_{\bf w}$ and $\mathcal{S}^{c}=\mathcal{S}_{\bf w}\backslash \mathcal{S}$. Denote $\phi_{\mathcal{S}^{c}}^{(n)}({\bf g},{\bf h})$ the optimal cost of \eqref{eq:primary} when the optimization is constrained over ${\bf w}\in\mathcal{S}^{c}$. Suppose there exists constants $\overline{\phi}$ and $\overline{\phi}_{\mathcal{S}_c}$ such that (i) $\phi^{(n)}({\bf g},{\bf h})\to \overline{\phi}$ in probability, (ii) $\phi_{\mathcal{S}^{c}}^{(n)}({\bf g},{\bf h})\to \overline{\phi}_{\mathcal{S}^{c}}$ in probability, (iii) $\phi< \overline{\phi}_{\mathcal{S}^{c}}$. Then, $\lim_{n\to\infty}\mathbb{P}\left[{\bf w}_{\Phi}\in\mathcal{S}\right]=1$, where ${\bf w}_{\Phi}$ is a minimizer of \eqref{eq:primary}.
\label{th:original_CGMT} 
\end{theorem}
Theorem~\ref{th:original_CGMT} allows to characterize a set in which lies the minimizer of \eqref{eq:primary} with probability approaching $1$.  The main ingredient is to compare the asymptotic limit of the AO optimal costs on the set of interest and its complementary. These asymptotic limits are not needed as per Theorem \ref{th:original_CGMT} to be in the almost sure sense, although it is often the case that it is this stronger convergence that effectively holds. The reason is that Gordon's result involve comparison  between probabilities associated with the AO and PO. The rate of convergence of  the probability corresponding to the (AO) cannot be easily characterized, which does not allow to transfer directly the almost sure convergence  of the (AO) into that of the (PO).  Using a converse version of the Borel Cantelli Lemma we show that it is possible to prove the almost sure convergence of the (PO) given that of the (AO), which strengthens the original CGMT. The result is presented in the following Theorem:
 \begin{theorem}
Let $\mathcal{S}$ be an arbitrary open subset of $\mathcal{S}_{\bf w}$ and $\mathcal{S}^{c}=\mathcal{S}_{\bf w}\backslash \mathcal{S}$. Denote $\phi_{\mathcal{S}^{c}}^{(n)}({\bf g},{\bf h})$ the optimal cost of \eqref{eq:primary} when the optimization is constrained over ${\bf w}\in\mathcal{S}^{c}$. Suppose there exists constants $\overline{\phi}$ and $\overline{\phi}_{\mathcal{S}_c}$ such that (i) $\phi^{(n)}({\bf g},{\bf h})\to \overline{\phi}$ almost surely, (ii) $\phi_{\mathcal{S}^{c}}^{(n)}({\bf g},{\bf h})\to \overline{\phi}_{\mathcal{S}^{c}}$ almost surely, (iii) $\phi< \overline{\phi}_{\mathcal{S}^{c}}$. Then, $\mathbb{P}\left[{\bf w}_{\Phi}\in\mathcal{S}, \textnormal{i.o.} \right]=1$, where ${\bf w}_{\Phi}$ is a minimizer of \eqref{eq:primary}.
\label{th:new_CGMT} 
\end{theorem}
\begin{proof}
Let $\eta=\frac{\overline{\phi}_{\mathcal{S}^{c}}-\overline{\phi}}{3}>0$. Then, $\overline{\phi}_{\mathcal{S}^{c}}-\eta=\overline{\phi}+2\eta$. The event $\mathcal{G}_n:=\left\{\phi^{(n)}({\bf g},{\bf h})\geq \overline{\phi}+\eta\right\}$ does not occur infinitely often, hence,
$$
\mathbb{P}\left[\phi^{(n)}({\bf g},{\bf h})\geq \overline{\phi}+\eta,i.o.\right]=0.
$$
Since $\mathcal{G}_n$ are independent, each event being generated by independent vectors ${\bf g}$ and ${\bf h}$, the converse of Borel-Cantelli Lemma implies that $\sum_{n=1}^\infty  \mathbb{P}(\mathcal{G}_n)<\infty$. Similarly, we can prove that $\mathcal{R}_n:=\left\{\phi_{\mathcal{S}^c}^{(n)}({\bf g},{\bf h})\leq \overline{\phi}_{\mathcal{S}^{c}}-\eta\right\}$ satisfy $\sum_{n=1}^\infty \mathbb{P}[\mathcal{R}_n]<\infty$.
Let $\Phi_{\mathcal{S}^{c}}^{(n)}({\bf G})$ be the optimal cost of the PO problem when the minimization is constrained over ${\bf w}\in\mathcal{S}^{c}$. Consider now the event:
$$
\mathcal{K}_n=\left\{\Phi_{\mathcal{S}^{c}}^{(n)}({\bf G})\geq \overline{\phi}_{\mathcal{S}^c}-\eta,\ \Phi({\bf G})^{(n)}\leq \overline{\phi}+\eta\right\}
$$
In this event, we have  $\Phi_{\mathcal{S}^{c}}^{(n)}({\bf G})\geq \overline{\phi}_{\mathcal{S}^c}-\eta=\overline{\phi}+2\eta$. Hence, $\Phi_{\mathcal{S}^{c}}^{(n)}({\bf G})> \overline{\Phi}({\bf G})$, which implies that ${\bf w}_{\Phi}\in\mathcal{S}$. As a consequence, 
$$
\mathbb{P}\left[{\bf w}_{\Phi}\notin\mathcal{S}\right]\leq \mathbb{P}(\mathcal{K}_n^{c})
$$
where $\mathcal{K}_n^{c}$ is the complementary event of $\mathcal{K}_n$. From the union bound, 
$$
\mathbb{P}\left[\mathcal{K}_n^c\right]\leq \mathbb{P}[\mathcal{R}_n]+\mathbb{P}[\mathcal{G}_n].
$$
Hence, 
$$
\sum_{n=1}^\infty \left[\mathcal{K}_n^c\right]<\infty,
$$ 
which proves that $\mathcal{K}_n^{c}$ and thus $\left\{{\bf w}_{\Phi}\in\mathcal{S}\right\}$ does not occur infinitely often. 
\end{proof}

 

\subsection{Hard Margin SVM}
\subsubsection{Identification of the  (PO) and (AO) Problems}
The max-margin solution is obtained by solving the following optimization problem
$$
\min_{{\bf w},b}\max_{\substack{\tilde{u}_i\geq 0\\ i=1,\cdots,n}}{\bf w}^{T}{\bf w}+\sum_{i=1}^n\tilde{u}_i\left(1-y_i{\bf w}^{T}\left(y_i\boldsymbol{\mu}+\sigma{\bf z}_i\right)-y_ib\right).
$$
Let $\tilde{\bf u}=\left[\tilde{u}_1,\cdots,\tilde{u}_n\right]^{T}$ and $\left\{j_i\right\}$ be the vector indexing the observations belonging to class $\mathcal{C}_i$. Let ${\bf Z}=\left[-y_1 {\bf z}_1,\cdots, -y_n {\bf z}_n\right]^{T}$. We need thus to solve the following optimization problem
$$
\min_{{\bf w},b}\max_{\tilde{\bf u}\geq 0} {\bf w}^{T}{\bf w}+{\bf 1}^{T}\tilde{\bf u}-{\bf w}^{T}\boldsymbol{\mu}{\bf 1}^{T}\tilde{\bf u}-({\bf j}_1^{T}-{\bf j}_0^{T})\tilde{\bf u}\  b +\sigma\tilde{\bf u}^{T}{\bf Z}{\bf w}.
$$ 
Performing the change of variable ${\bf u}=\sigma\sqrt{p}\tilde{\bf u}$ leads to the following primary optimization problem 
$$
\Phi^{(n)}\triangleq \min_{{\bf w},b}\max_{{\bf u}\geq 0} 
\frac{1}{\sqrt{p}}{\bf u}^{T}{\bf Z}{\bf w} +\psi({\bf w},{\bf u}).
$$ 
with $\psi({\bf w},{\bf u})\triangleq {\bf w}^{T}{\bf w}+\frac{1}{\sqrt{p}\sigma}{\bf 1}^{T}{\bf u}-\frac{1}{\sqrt{p}\sigma}{\bf w}^{T}\boldsymbol{\mu}{\bf 1}^{T}{\bf u}-\frac{1}{\sqrt{p}\sigma}({\bf j}_1^{T}-{\bf j}_0^{T}){\bf u}\  b$. 
The CGMT requires the feasibility sets of the optimization variables to be compact. Obviously, this is not satisfied since the feasibility set associated with ${\bf w}$ and ${\bf u}$ are not compact. To solve this issue, we write $\Phi^{(n)}$ as:
\begin{equation}
\Phi^{(n)}=\inf_{r,B\geq 0}\sup_{\theta \geq 0}\min_{\substack{{\bf w}\in \mathbb{R}^{p}\\ {\norm{\bf w}}_2\leq r\\ b\leq B}} 
	\frac{1}{\sqrt{p}}{\bf u}^{T}{\bf Z}{\bf w} +\psi({\bf w},{\bf u})
\label{eq:writing}
\end{equation}
and call $\Phi_{r,B}$ for $r,B, \theta\geq 0$ and $\Phi_{r,B,\theta}$  the following optimization problems:
$$
\Phi_{r,B}^{(n)}=\sup_{\theta\geq 0}\min_{\substack{{\bf w}\in \mathbb{R}^{p}\\ {\norm{\bf w}}_2\leq r\\ |b|\leq B}} \max_{\substack{{\bf u}\geq 0 \\ {\norm{\bf u}}_2\leq \theta}} \frac{1}{\sqrt{p}}{\bf u}^{T}{\bf Z}{\bf w} +\psi({\bf w},{\bf u})
$$
\begin{equation}
\Phi_{r,B,\theta}^{(n)}=\min_{\substack{{\bf w}\in \mathbb{R}^{p}\\ {\norm{\bf w}}_2\leq r\\ |b|\leq B}} \max_{\substack{{\bf u}\geq 0 \\ {\norm{\bf u}}_2\leq \theta}} \frac{1}{\sqrt{p}}{\bf u}^{T}{\bf Z}{\bf w} +\psi({\bf w},{\bf u})
\label{eq:cgmt_phi_rtheta}
\end{equation}

We identified thus a family of primiary problems indexed by $(r,B,\theta)$, each of which admits the desired format and satisfies the compactness conditions required by the CGMT theorem.
Particularly, we can easily distinguish the  the bilinear form $\frac{1}{\sqrt{p}}{\bf u}^{T}{\bf Z}{\bf w}$ and the  function $\psi({\bf w},{\bf u})$ which is convex in ${\bf w}$ and linear thus concave in ${\bf u}$.  
 We associate thus with each one of them the following auxiliary optimization (AO) problem  which can be written as:
$$
\phi_{r,B,\theta}^{(n)}=\min_{\substack{{\bf w}\in \mathbb{R}^{p}\\ {\norm{\bf w}}_2\leq r \\ |b|\leq B}} \max_{\substack{{\bf u}\geq 0 \\ {\norm{\bf u}}_2\leq \theta}} \frac{1}{\sqrt{p}}{\norm{\bf w}}_2{\bf g}^{T}{\bf u}-\frac{1}{\sqrt{p}}{\norm{\bf u}}_2{\bf h}^{T}{\bf w}+\psi({\bf w},{\bf u})
$$
Now that we have identified the (AO) problems, we wish to solve them and infer their asymptotic behavior. To this end, we proceed in two steps. First, we simplify the (AO) problems by reducing them to problems that involve only optimization over a few number of scalars. In doing so, the asymptotic behavior of the AO problems is much simplified and is carried out in the second step.   
\subsubsection{Simplification of the (AO)  Problems}
One major step towards the simplification of the (AO) problems is to reduce them to problems that involve only few scalar optimization parameters.  Obviously, the objective function of the AO lends itself to this kind of simplification, vector ${\bf w}$ appearing only through its norm or its scalar product ${\bf w}^{T}\boldsymbol{\mu}$ and ${\bf w}^{T}{\bf h}$. In light of this observation,  we decompose ${\bf w}$ as:
$$
{\bf w}=\alpha_1 \frac{\boldsymbol{\mu}}{\norm{\boldsymbol{\mu}}_2} +\alpha_2 {\bf w}_{\perp},
$$
where ${\bf w}_{\perp}$ is a unit norm vector orthogonal to $\boldsymbol{\mu}$. With these notations at hand, we write the (AO) as:
\begin{align*}
	&\phi_{r,B, \theta}^{(n)}\!=\!\min_{\substack{{\bf w}\in \mathbb{R}^{p}\\ {\norm{\bf w}}_2\leq r \\ |b|\leq B}} \!\!\max_{\substack{{\bf u}\geq 0 \\ {\norm{\bf u}}_2\leq \theta}} \frac{1}{\sqrt{p}} \sqrt{\alpha_1^2+\alpha_2^2}{\bf g}^{T}{\bf u}- \frac{1}{\sqrt{p}} {\norm{\bf u}}_2 (\alpha_1 \frac{{\bf h}^{T}\boldsymbol{\mu}}{{\norm{\boldsymbol{\mu}}}_2}\\
	&+{\alpha_2} {\bf h}^{T}{\bf w}_{\perp}) +\alpha_1^2+\alpha_2^2 +\frac{1}{\sqrt{p}\sigma}{\bf 1}^{T}{\bf u}-\alpha_1\norm{\boldsymbol{\mu}}_2 \frac{1}{\sqrt{p}\sigma}{\bf 1}^{T}{\bf u}\\
	&- \frac{1}{\sqrt{p}\sigma} ({\bf j}_1^{T}-{\bf j}_0^{T}){\bf u}  b \ \ .
\end{align*}
We will now prove that optimizing over ${\bf w}$ reduces to optimizing over the set of scalars $(\alpha_1,\alpha_2)$. Note here, that flipping the min-max is not permitted since the objective function is not convex in ${\bf w}$ and concave in ${\bf u}$. One however is tempted to replace ${\bf w}_{\perp}$ by ${\rm sign}(\alpha_2)\frac{{\bf h}_\perp }{{\norm{\bf h}_\perp}_2}$, where ${\bf h}_{\perp}$ is the orthogonal projection of ${\bf h}$ onto the subspace orthogonal to $\boldsymbol{\mu}$,  since this would minimize the objective function for any ${\bf u}$. This property, that the vector ${\bf w}_{\perp}= {\rm sign}(\alpha_2)\frac{{\bf h}_{\perp} }{{\norm{{\bf h}_\perp}}_2}$ minimizes the objective function for any ${\bf u}$ allows us,   using Lemma \ref{lem:prop} proven in the Appendix, to show that ${\phi}_{r,B,\theta}$ is also given by:
\begin{align*}
	&\phi_{r,B,\theta}^{(n)}\!\!=\!\!\min_{\substack{\alpha_1,\alpha_2\in\mathbb{R} \\ \alpha_1^2+\alpha_2^2\leq r^2\\ |b|\leq B}} \max_{\substack{{\bf u}\geq 0 \\ {\norm{\bf u}}_2\leq \theta}} \frac{1}{\sqrt{p}} \sqrt{\alpha_1^2+\alpha_2^2}{\bf g}^{T}{\bf u}\\
	&- \frac{1}{\sqrt{p}} {\norm{\bf u}}_2 (\alpha_1 \frac{{\bf h}^{T}\boldsymbol{\mu}}{{\norm{\boldsymbol{\mu}}}_2}+{|\alpha_2|} {\norm{{\bf h}_\perp}_2}) +\alpha_1^2+\alpha_2^2 +\frac{1}{\sqrt{p}\sigma}{\bf 1}^{T}{\bf u}\\
	&-\alpha_1\norm{\boldsymbol{\mu}}_2 \frac{1}{\sqrt{p}\sigma}{\bf 1}^{T}{\bf u}- \frac{1}{\sqrt{p}\sigma} ({\bf j}_1^{T}-{\bf j}_0^{T}){\bf u}  b\ \ .
\end{align*}
obtained by replacing ${\bf w}_{\perp}$ by ${\rm sign}(\alpha_2)\frac{{\bf h} }{{\norm{\bf h}}_2}$. We emphasize here on the fact that in this operation, we do not perform any permutation of the order of the min-max. In the sequel, it is convenient to perform the optimization over $q_0=\sqrt{\alpha_1^2+\alpha_2^2}$ and $\alpha_1$. With this notation at hand, $\phi_{r,B,\theta}$ is simplified in \eqref{eq:sim} where \eqref{eq:phi_1} follows from decomposing the maximization over ${\bf u}$ into the maximization over its direction and its magnitude, \eqref{eq:phi_2} is obtained by applying lemma \ref{lem:opt_norm} and \eqref{eq:AO} is derived  by performing the change of variable $\rho=\frac{\alpha_1}{q_0}$.  
\begin{figure*}
	\begin{subequations}\label{eq:sim}
\begin{align}
	&\phi_{r,B,\theta}^{(n)}\!\!=\!\!\min_{\substack{0\leq q_0\leq r\\|\alpha_1|\leq q_0\\ |b|\leq B}} \max_{\substack{{\bf u}\geq 0 \\ {\norm{\bf u}}_2\leq \theta}} \frac{1}{\sqrt{p}} q_0{\bf g}^{T}{\bf u}- \frac{1}{\sqrt{p}} {\norm{\bf u}}_2 (\alpha_1 \frac{{\bf h}^{T}\boldsymbol{\mu}}{{\norm{\boldsymbol{\mu}}}_2}+{\sqrt{q_0^2-\alpha_1^2}} {\norm{{\bf h}_\perp}_2}) +\alpha_1^2+\alpha_2^2 +\frac{1}{\sqrt{p}\sigma}{\bf 1}^{T}{\bf u}-\alpha_1\norm{\boldsymbol{\mu}}_2 \frac{1}{\sqrt{p}\sigma}{\bf 1}^{T}{\bf u}\nonumber\\
	&- \frac{1}{\sqrt{p}\sigma} ({\bf j}_1^{T}-{\bf j}_0^{T}){\bf u}  b, \label{eq:phi_1}\\
	&=\min_{\substack{q_0\leq r\\|\alpha_1|\leq q_0\\ |b|\leq B}} \max_{\theta \geq m\geq 0} \max_{{\norm{\bf u}}_2=m} {\bf u}^{T}\left(q_0 {\bf g} +\frac{1}{\sqrt{p}\sigma}{\bf 1}-\alpha_1{\norm{\bf u}}_2\frac{1}{\sqrt{p}\sigma}{\bf 1}-\frac{1}{\sqrt{p}\sigma}({\bf j}_1-{\bf j}_0)b\right) +q_0^2-m\left(\frac{1}{\sqrt{p}}\alpha_1\frac{{\bf h}^{T}\boldsymbol{\mu}}{\norm{\boldsymbol{\mu}}_2}+\sqrt{q_0^2-\alpha_1^2}\frac{1}{\sqrt{p}}\norm{{\bf h}_{\perp}}_2\right)\label{eq:phi_2}\\
	&=\min_{\substack{0\leq q_0\leq r\\|\alpha_1|\leq q_0\\ |b|\leq B}}\max_{\theta\geq m\geq 0}q_0^2+m\sqrt{\frac{1}{p}\sum_{i\in\mathcal{C}_1}\big(q_0 g_i+\frac{1-\alpha_1\norm{\boldsymbol{\mu}}_2-b}{\sigma}\big)_{+}^2+\frac{1}{p}\sum_{i\in\mathcal{C}_0}\big(q_0 g_i+\frac{1-\alpha_1{\norm{\boldsymbol{\mu}}}_2+b}{\sigma}\big)_{+}^2}\nonumber\\
	&-m\left(\frac{1}{\sqrt{p}}\alpha_1\frac{{\bf h}^{T}\boldsymbol{\mu}}{\norm{\boldsymbol{\mu}}_2}+\sqrt{q_0^2-\alpha_1^2}\frac{1}{\sqrt{p}}\norm{{\bf h}_{\perp}}_2\right),\\
	&=\min_{\substack{0\leq q_0\leq r\\-1\leq \rho\leq 1\\ |b|\leq B}}q_0^2+\theta q_0\Big(\sqrt{\frac{1}{p}\sum_{i\in\mathcal{C}_1}\big( g_i+\frac{1}{q_0\sigma}-\frac{\rho\norm{\boldsymbol{\mu}}_2}{\sigma}-\frac{b}{q_0\sigma}\big)_{+}^2+\frac{1}{p}\sum_{i\in\mathcal{C}_0}\left(g_i+\frac{1}{q_0\sigma}-\rho\frac{\norm{\boldsymbol{\mu}}_2}{\sigma}+\frac{b}{\sigma q_0}\right)_{+}^2}\nonumber\\
	&-\big(\frac{1}{\sqrt{p}}\rho\frac{{\bf h}^{T}\boldsymbol{\mu}}{\norm{\boldsymbol{\mu}}_2}+\sqrt{1-\rho^2}\frac{1}{\sqrt{p}}\norm{{\bf h}_{\perp}}_2\big)\Big)_{+}\ \ .\label{eq:AO}
\end{align}
	\end{subequations}
\hrule
\end{figure*}

The above simplification of the auxiliary problem follows through a deterministic analysis that does not involve any asymptotic approximations. Contrary to the original writing of the AO, this new simplification is more handy towards understanding its asymptotic behavior. This constitutes the objective of the next section.   
\subsubsection{Asymptotic Behavior of the (AO) problems (Proof of Theorem \ref{th:hardmargin_condition})}
\label{tech_hard_margin_cond}
A well-known fact is that the hard-margin SVM does not always lead to a finite solution, but it is not clear as to when this happens. In the following, we prove that a careful analysis of the AO problems allow us to provide a rigorous answer to this question. Particularly, we will prove that if the condition in Theorem~\ref{th:hardmargin_condition} holds true, with probability $1$, the hard-margin SVM leads to infinite solution for sufficiently large dimensions $n$ and $p$. The key idea of the proof relies on showing that under the condition  in Theorem~\ref{th:hardmargin_condition}, with probability $1$, $\phi_{r,B,\theta}^{(n)}\geq u_n \theta$, where $u_n$ is a certain sequence. We will show later that this property implies the failure of the hard-margin SVM to provide a finite solution. 
To begin with, we define for fixed $q_0\in\mathbb{R}_{+}$, $\rho\in[-1,1]$ and $\eta\in\mathbb{R}$ function $\hat{D}_H(q_0,\rho,\eta)$ in \eqref{eq:DH_hat}
\begin{figure*}
	\begin{align}
		\hat{D}_H:(q_0,\rho,\eta)\mapsto &\sqrt{\frac{1}{p}\sum_{i\in\mathcal{C}_1}\big( g_i+\frac{1}{q_0\sigma}-\frac{\rho\norm{\boldsymbol{\mu}}_2}{\sigma}-\eta\big)_{+}^2+\frac{1}{p}\sum_{i\in\mathcal{C}_0}\big(g_i+\frac{1}{q_0\sigma}-\rho\frac{\norm{\boldsymbol{\mu}}_2}{\sigma}+\eta\big)_{+}^2}-\left(\frac{1}{\sqrt{p}}\rho\frac{\left|{\bf h}^{T}\boldsymbol{\mu}\right|}{\norm{\boldsymbol{\mu}}_2}+\sqrt{1-\rho^2}\frac{1}{\sqrt{p}}\norm{{\bf h}_{\perp}}_2\right)
\label{eq:DH_hat}
	\end{align}
\hrule
\end{figure*}
We can thus lower-bound $\phi_{r,B,\theta}^{(n)}$ as:
 \begin{align}
	 \phi_{r,B,\theta}^{(n)}\!\!&\geq  \min_{\substack{0\leq q_0\leq r\\-1\leq \rho\leq 1\\ \eta\in\mathbb{R}}} q_0^2+\theta q_0\left(\hat{D}_H(q_0,\rho,\eta)\right)_{+}.
\label{eq:lower_bound}
\end{align}
where $(a)_{+}=\max(a,0)$. 
Function $h_n:q_0\mapsto \hat{D}_H(q_0,\rho,\eta)$ is decreasing in $q_0$.  
We may thus find a lower-bound for it by taking its limit as $q_0\to\infty$. However this would not be helpful, since after replacing this function by this lower-bound, and optimizing over $q_0$, we find that $\tilde{\phi}_{r,B,\theta}\geq 0$, a fact that does not carry a lot of information.  To solve this problem, we need to consider the cases when $q_0$ is in the vicinity of zero, and when $q_0$ is sufficiently  far away from zero. When $q_0$ is very close to zero, in a sense that will be defined, we may expect $h_n(q_0)\geq \frac{C}{q_0}$, and hence $q_0h_n(q_0)\geq C$. This will allow us to prove the sought-for scaling behaviour with respect to $\theta$ of $\phi_{r,B,\theta}$ when $q_0$ is in the vicinity of zero. One can easily see that if  $0\leq q_0\leq q_{U}\triangleq \frac{1}{2\sigma \max_{1\leq i\leq n}\left|g_i\right|+2\norm{\boldsymbol{\mu}}_2}$, then $g_i+\frac{1}{q_0\sigma}-\frac{\rho\norm{\boldsymbol{\mu}_2}}{\sigma}\geq \frac{1}{2q_0\sigma}$, thereby implying that:
\begin{align}
	&\min_{\substack{0\leq q_0\leq q_U\\ \eta\in\mathbb{R}\\ |\rho|\leq 1 }}\theta q_0\Big(\hat{D}_H(q_0,\rho,\eta)\Big)_{+}\!\!\\
	&\geq \min_{\substack{0\leq q_0\leq q_U\\ \eta\in\mathbb{R}\\ |\rho|\leq 1 }}\theta q_0 \Big(\sqrt{\frac{n_1}{p}(\frac{1}{2q_0\sigma}-\eta)_{+}+\frac{n_0}{p}(\frac{1}{2q_0\sigma}+\eta)_{+}^2}-\frac{1}{\sqrt{p}} \frac{|{\bf h}^{T}\boldsymbol{\mu}|}{\norm{\boldsymbol{\mu}}_2}\\
	&-\frac{1}{\sqrt{p}} \norm{{\bf h}_{\perp}}_2\Big)_{+}\nonumber\\
&\stackrel{(a)}{\geq} \theta \frac{1}{\sigma} \sqrt{\frac{n_1n_0}{np}}-\theta q_U \left(\frac{1}{\sqrt{p}} \frac{|{\bf h}^{T}\boldsymbol{\mu}|}{\norm{\boldsymbol{\mu}}_2}+\frac{1}{\sqrt{p}}\norm{{\bf h}_{\perp}}\right)\ \ .\label{eq:first_bound}
\end{align}
where $(a)$ follows from performing the optimization over $\eta\in\mathbb{R}$. Since  $q_U\leq \frac{1}{2\sigma |\max_{1\leq i\leq n}g_i|+2\norm{\boldsymbol{\mu}}_2}$, and  $|\frac{\max_{1\leq i\leq n}g_i}{\sqrt{2\log n}}|\asto 1$, $q_U$ converges to $0$ almost surely. Hence, with probability $1$ as $n$ and $p$ are sufficiently large, $q_U\leq \frac{1}{2\sigma}\sqrt{\frac{n_1n_0}{np}}$. We have thus proved that with probability $1$, for large $n$ and $p$ 
\begin{equation}
	\min_{\substack{0\leq q_0\leq q_U\\ \eta\in\mathbb{R}\\ |\rho|\leq 1 }}\theta q_0\Big(\hat{D}_H(q_0,\rho,\eta)\Big)\!\!\geq \theta\frac{1}{2\sigma}\sqrt{\frac{n_1n_0}{np}}\ \ \label{eq:cond1}. 
\end{equation}
We will now consider the optimization of $\phi_{r,B,\theta}^{(n)}$ when $q_U\leq q_0\leq r$. 
\footnote{Without loss of generality, we  assume that $r\geq q_U$.} Using the fact that function $h_n$ is decreasing in $q_0$, we obtain:
\begin{align}
	\min_{\substack{q_U\leq q_0\leq r\\ |\rho|\leq 1 \\ \eta\in\mathbb{R}}}  \theta q_0\Big(\hat{D}_H(q_0,\rho,\eta)\Big)_{+}&\geq\min_{\substack{q_U\leq q_0\leq r\\ |\rho|\leq 1 \\ \eta\in\mathbb{R}}}  \theta q_0 \Big(\ell_n(\rho,\eta)\Big)_{+}
\end{align}
where $\ell_n:[-1,1]\times\mathbb{R}$ with
\begin{align}
	\ell_n(\rho,\eta)&=\sqrt{\frac{1}{p}\sum_{i\in\mathcal{C}_1}\left(g_i-\frac{\rho\norm{\boldsymbol{\mu}}}{\sigma}-\eta\right)_{+}^2+\frac{1}{p}\sum_{i\in\mathcal{C}_0}\left(g_i-\frac{\rho\norm{\boldsymbol{\mu}}}{\sigma}+\eta\right)_{+}^2}\nonumber\\
	&-\frac{1}{\sqrt{p}}\frac{{\bf h}^{T}\boldsymbol{\mu}}{\norm{\boldsymbol{\mu}}}-\sqrt{\frac{1-\rho^2}{p}}\norm{{\bf h}_{\perp}}.
\end{align}
It is easy to see that $\ell_n$ is  jointly convex function in its arguments $(\rho,\eta)$ and converges almost surely to 
\begin{align*}
	&\overline{\ell}:(\rho,\eta)\mapsto \sqrt{\delta \pi_1\mathbb{E}\big(G-\frac{\rho\mu}{\sigma}-\eta\big)_{+}^2+\delta\pi_0\mathbb{E}\big(G-\frac{\rho\mu}{\sigma}+\eta\big)_{+}^2}\\
	&	-\sqrt{1-\rho^2}
\end{align*}
where $G\sim\mathcal{N}(0,1)$. 
 Since $\lim_{\eta\to\infty} \overline{\ell}(\rho,\eta)=\infty$, using Lemma 11 and Lemma 10 in \cite{thrampoulidis-IT}, we obtain:
$$
\min_{\eta\in\mathbb{R}}\ell_n(\rho,\eta)\asto \min_{\eta\in\mathbb{R}} \overline{\ell}(\rho,\eta)\ \ .
$$
Moreover, for $\rho\in\left[-1,1\right]$, expressing the first order conditions with respect to $\eta$, the optimum $\eta^\star(\rho)$ is a solution to the following equation:
\begin{equation}
\eta=\frac{\frac{\pi_1}{\sqrt{2\pi}}\int_{\frac{\rho\mu}{\sigma}+\eta}^\infty (x-\frac{\rho \mu}{\sigma})Dx+\frac{\pi_0}{\sqrt{2\pi}}\int_{\frac{\rho\mu}{\sigma}-\eta}^\infty (\frac{\rho\mu}{\sigma}-x)Dx}{\frac{\pi_1}{\sqrt{2\pi}}\int_{\frac{\rho\mu}{\sigma}+\eta}^\infty Dx+\frac{\pi_0}{\sqrt{2\pi}}\int_{\frac{\rho\mu}{\sigma}-\eta}^\infty Dx}\label{eq:equation}
\end{equation}
It is easy to see that the solution of  \eqref{eq:equation} is unique. This is because function \begin{align}\eta:\mapsto &\eta\left({\pi_1}\int_{\frac{\rho\mu}{\sigma}+\eta}^\infty Dx+{\pi_0}\int_{\frac{\rho\mu}{\sigma}-\eta}^\infty Dx\right)\nonumber\\
&-{\pi_1}{}\int_{\frac{\rho\mu}{\sigma}+\eta}^\infty (x-\frac{\rho \mu}{\sigma})Dx-{\pi_0}\int_{\frac{\rho\mu}{\sigma}-\eta}^\infty (\frac{\rho\mu}{\sigma}-x)Dx\end{align} is decreasing with limits $\infty$ and $-\infty$ when $\eta\to-\infty$ and $\eta\to\infty$ respectively. 
Using Lemma \ref{lem:convexity} in the Appendix, Function $\rho\mapsto \displaystyle\min_{\eta\in\mathbb{R}}\ell_{n}(\eta,\rho)$ is convex in $\rho$. Since the convergence of convex functions is uniform over compacts, from Theorem 2.1 in  \cite{newey_large_1994}, we have:
\begin{equation}
\min_{\substack{\eta\in\mathbb{R}\\-1\leq \rho\leq 1 }}\ell_n(\rho,\eta)\asto \min_{\substack{\eta\in\mathbb{R}\\ -1\leq \rho\leq 1}} \overline{\ell}(\rho,\eta)\ \ .
\label{eq:conv}
\end{equation}
If Condition \eqref{condition} is satisfied, $\overline{\ell}\triangleq\min_{\substack{\eta\in\mathbb{R}\\ -1\leq\rho\leq 1}}\overline{\ell}(\rho,\eta)>0$, which implies that for all $\epsilon>0$, and sufficiently large $n$  and $p$ we have with probability $1$,
\begin{equation}
	\min_{\substack{q_U\leq q_0\leq r\\ |\rho|\leq 1}} \theta q_0\hat{D}_H(q_0,\rho,\eta) \geq \theta q_U \left(\overline{\ell}-\epsilon\right)\ \ .
\label{eq:second_bound}
\end{equation}
Taking $\epsilon=\frac{\overline{\ell}}{2}$ and combining \eqref{eq:first_bound} and \eqref{eq:second_bound} leads to:
\begin{equation}
 \min_{\substack{q_U\leq q_0\leq r\\ |\rho|\leq 1}} \theta q_0\hat{D}_H(q_0,\rho,\eta)\geq\theta q_U \frac{\overline{\ell}}{2} \label{eq:rB}
\end{equation}
almost surely for enough large $n$ and $p$. Combining \eqref{eq:cond1} and \eqref{eq:rB}, yields
$$
\phi_{r,B,\theta}^{(n)}\geq \theta\min\left(\frac{1}{2\sigma}\sqrt{\frac{n_1n_0}{np}},q_U\frac{\overline{\ell}}{2}\right)
$$
As $q_U$ converges almost surely to zero, for sufficiently large $n$ and $p$, $\min\left(\frac{1}{2\sigma}\sqrt{\frac{n_1n_0}{np}},q_U\frac{\overline{\ell}}{2}\right)=q_U\frac{\overline{\ell}}{2}$. Hence for sufficiently large $n$ and $p$,
\begin{equation}
\phi_{r,B,\theta}^{(n)}\geq \theta q_U\frac{\overline{\ell}}{2}
	\label{eq:phi_rBtheta}
\end{equation}
With the above inequality \eqref{eq:phi_rBtheta} at hand, we are now ready to establish Theorem \ref{th:hardmargin_condition}. First, we shall bring to the reader's attention that the order of magnitude of $n$ and $p$ above which  \eqref{eq:rB} holds is independent of $r,B$ and $\theta$. It entails from this that the set:
\begin{align}
	&\mathcal{E}\triangleq  \left\{\cup_{k=1}^\infty \cap_{m=1}^\infty\Big\{ \phi_{k,k,m}^{(n)} \geq m q_U\frac{\overline{\ell}}{2}\right\} \nonumber\\
	&\ \textnormal{for} \ n  \ \textnormal{and} \ p \ \textnormal{sufficiently large}\Big\}
\label{eq:io}
\end{align}
verifies $\mathbb{P}\left[\mathcal{E}\right]=1$ as the countable intersection and union of events with probability $1$.
Let us now consider the optimal value $\Phi^{(n)}$ of the primary optimization problem and illustrate how the characterization of the auxiliary problem allows to ensure that under the setting of Theorem \ref{th:hardmargin_condition}, $\Phi=\infty$ for $n$ and $p$ sufficiently large. One way to prove this is to show that for all $x>0$, $ \mathbb{P}\left[\Phi^{(n)}\leq x, \ \ \textnormal{for} \ \ n,p \ \ \textnormal{sufficiently large}\right]=0$. From \eqref{eq:writing}, if $\Phi^{(n)}\neq \infty$, for $\epsilon>0$ sufficiently small, there exists $k\in\mathbb{N}$ such that $\Phi^{(n)} \geq \Phi_{k,k}^{(n)}-\epsilon$. Hence, 
\begin{align*}
\mathbb{P}\left[\left\{\Phi^{(n)}\leq x\right\}\cap\left\{\Phi^{(n)}\neq \infty\right\}\right] \leq \mathbb{P}\left[\cup_{k=1}^\infty \left\{\Phi_{k,k}^{(n)}\leq x+\epsilon\right\}\right]\\
\leq \mathbb{P}\left[\cup_{k=1}^\infty \left\{\cap_{m=1}^\infty\left\{\Phi_{k,k,m}^{(n)}\leq x+\epsilon\right\}\right\}\right],
\end{align*} 
For $m\in\mathbb{N}^\star$, the events $\mathcal{E}_k=\left\{\cap_{m=1}^\infty \left\{\Phi_{k,k,m}^{(n)}\leq (x+\epsilon)\right\}\right\}$ forms an increasing sequence of events, thus:
$$
\mathbb{P}\left[\cup_{k=1}^\infty \mathcal{E}_k\right]=\lim_{k\to\infty} \mathbb{P}\left[\mathcal{E}_k\right].
$$
Similarly, as $\Phi_{k,k,m}^{(n)}\geq \Phi_{k,k,m-1}^{(n)}$, for $k\in\mathbb{N}^\star$, the sequence of events, $\mathcal{E}_{k,m}=\left\{\Phi_{k,k,m}^{(n)}\leq x\right\}$ is decreasing, thus:
$$
\mathbb{P}\left[\cap_{m=1}^\infty \mathcal{E}_{k,m}\right]=\lim_{m\to\infty} \mathbb{P}\left[\mathcal{E}_{k,m}\right]
$$
We thus obtain:
$$
\mathbb{P}\left[\Phi^{(n)}\leq x\right] \leq \lim_{k\to\infty} \lim_{m\to\infty} \mathbb{P}\left[\left\{\Phi_{k,k,m}^{(n)}\leq x+\epsilon\right\}\right]
$$
From the CGMT theorem, we have:
\begin{align*}
\!\!\mathbb{P}\left[\Phi_{k,k,m}^{(n)}\leq x+\epsilon\right]\leq & \mathbb{P}\left[\phi_{k,k,m}^{(n)}\leq x+\epsilon\right].
\end{align*}
Hence, 
\begin{align}
	&\!\!	\lim_{k\to\infty}\lim_{m\to\infty} \mathbb{P}\left[\Phi_{k,k,m}^{(n)}\leq x+\epsilon\right]\leq \!\!\lim_{k\to\infty}\lim_{m\to\infty} \mathbb{P}\left[\phi_{k,k,m}^{(n)}\leq x+\epsilon\right]\\
	&=\mathbb{P}\left[\cup_{k=1}^\infty\left\{\cap_{m=1}^\infty \{\phi_{k,k,m}^{(n)}\leq x+\epsilon\}\right\}\right]
\end{align}
Using the fact that $\mathbb{P}[\mathcal{E}]=1$ with $\mathcal{E}$ given by \eqref{eq:io}, the event $A_n=\left\{\cup_{k=1}^\infty\left\{\cap_{m=1}^\infty \{\phi_{k,k,m}^{(n)}\leq x+\epsilon\}\right\}\right\}$ does not occur infinitely often, or in other words $\mathbb{P}\left[A_n, i.o\right] =0$. Since $(A_n)$ are independent, each event being generated by independent vectors ${\bf g}$ and ${\bf h}$ in $\mathbb{R}^{n\times 1}$ and $\mathbb{R}^{p\times 1}$, the converse of Borel-Cantelli lemma implies that $\sum_{n=1}^\infty \mathbb{P}(A_n)<\infty$. Therefore, 
$$
\sum_{n=1}^\infty\mathbb{P}\left[\Phi^{(n)}\leq x\right] <\infty \ \ . 
$$ 
Using Borel-Cantelli Lemma, we deduce that for any $x$, 
$$
\mathbb{P}\left[\Phi^{(n)}\leq x, \text{i.o}\right]=0\ \ .
$$
This implies that $\left\{\Phi^{(n)}=\infty\right\}$ occurs infinitely often. 
\subsubsection{Asymptotic Behavior of the (AO) problems (Proof of Theorem \ref{th:convergence_hard})}
\label{technical_proof_hard_margin_asym}
To begin with, we check first that function $\beta$ has a unique zero $q_0^\star$. Towards this end, note that the $\eta^\star(\rho,q_0)$ minimizing $D_H(q_0,\rho,\eta)$ for fixed $q_0$ and $\rho$ should be a solution of the following equation in $\eta$:
$$
\eta\!\!=\!\!\frac{\pi_1\int_{\frac{\rho\mu}{\sigma}-\frac{1}{q_0}+\eta}^\infty(x-\frac{\rho\mu}{\sigma}+\frac{1}{q_0})Dx-\pi_0 \int_{\frac{\rho\mu}{\sigma}-\frac{1}{q_0}-\eta}^\infty (x-\frac{\rho\mu}{\sigma}+\frac{1}{q_0})Dx}{\pi_1\int_{\frac{\rho\mu}{\sigma}-\frac{1}{q_0}+\eta}^\infty Dx+\pi_0\int_{\frac{\rho\mu}{\sigma}-\frac{1}{q_0}-\eta}Dx}
$$
Such an equation admits a unique solution because function 
\begin{align*}
&\eta\mapsto \eta\left[\pi_1\int_{\frac{\rho\mu}{\sigma}-\frac{1}{q_0}+\eta}^\infty Dx+\pi_0\int_{\frac{\rho\mu}{\sigma}-\frac{1}{q_0}-\eta}^\infty Dx\right]\\
&-\Big[\pi_1\int_{\frac{\rho\mu}{\sigma}-\frac{1}{q_0}+\eta}^\infty(x-\frac{\rho\mu}{\sigma}+\frac{1}{q_0})Dx\\
	&-\pi_0 \int_{\frac{\rho\mu}{\sigma}-\frac{1}{q_0\sigma}-\eta}^\infty (x-\frac{\rho\mu}{\sigma}+\frac{1}{q_0\sigma})Dx\Big]
\end{align*}
is an increasing function with limits $-\infty$ and $\infty$ when $\eta\to-\infty$ and $\eta\to\infty$. Moreover, $(\rho,q_0)\mapsto \eta^\star(\rho,q_0)$ is a continuous function. 
From the Maximum Theorem \cite[Theorem~9.17]{Sundaram}, function $q_0\mapsto \min_{-1\leq \rho\leq 1} D_H(q_0,\rho,\eta^\star(\rho,q_0))$ is continuous.   It tends to $\infty$ as $q_0\to 0^{+}$ and to 
\begin{align}
	&\displaystyle\min_{-1\leq \rho\leq 1}\Big(\frac{\delta\pi_1}{2\pi}\int_{\frac{\rho\mu}{\sigma}+\eta^\star(\rho)}^\infty (x-\frac{\rho\mu}{\sigma}-\eta^\star(\rho))^2Dx\nonumber\\
	&+ \frac{\delta\pi_0}{2\pi}\int_{\frac{\rho\mu}{\sigma}-\eta^\star(\rho)}(x-\frac{\rho\mu}{\sigma}+\eta^\star(\rho))^2 Dx\Big)^{\frac{1}{2}}-\sqrt{1-\rho^2}<0
\end{align}
when $q$ tends to $\infty$. There exists thus $q_0^\star$ such that $\displaystyle \min_{\substack{-1\leq \rho\leq 1\\ \eta\in\mathbb{R}}}D_H(q_0^\star,\rho,\eta)=0$. We will prove that necessarily  such a $q_0^\star$ is unique. Assume that there exists two solutions $q_{01}^\star$ and $q_{02}^\star$ such that $\displaystyle\min_{\substack{-1\leq \rho\leq 1 \\ \eta\in\mathbb{R}}}D_H(q_{01}^\star,\rho,\eta)=\displaystyle\min_{\substack{-1\leq \rho\leq 1 \\ \eta\in\mathbb{R}}}D_H(q_{02}^\star,\rho,\eta)=0$. Let ($\rho_1^\star$, $\rho_2^\star$) and ($\eta_1^\star(\rho_1^\star,q_{01}^\star)$, $\eta_2^\star(\rho_2,q_{02}^\star)$) such that $\displaystyle\min_{\substack{-1\leq \rho\leq 1 \\ \eta\in\mathbb{R}}}D_H(q_{01}^\star,\rho,\eta)=D_H(q_{01}^\star,\rho_1^\star,\eta_1^\star(\rho_1^\star,q_{01}^\star))$ and $\displaystyle\min_{\substack{-1\leq \rho\leq 1 \\ \eta\in\mathbb{R}}}D_H(q_{02}^\star,\rho,\eta)=D_H(q_{02}^\star,\rho_2^\star,\eta_2^\star(\rho_2^\star,q_{02}^\star))$. Hence,
\begin{align}
	0=D_H(q_{02}^\star,\rho_2^\star,\eta^\star(q_{02}^\star,\rho_2^\star))&=D_H(q_{01}^\star,\rho_1^\star,\eta^\star(q_{01}^\star,\rho_1^\star))\\
	&\leq D_H(q_{01}^\star,\rho_2^\star,\eta^\star(q_{02}^\star,\rho_2^\star))
\end{align}
Since for any $\eta\in \mathbb{R}$ and $\rho\in[-1,1]$, $q\mapsto D_H(q,\rho,\eta)$ is decreasing, $q_{01}^\star\geq q_{02}^\star$. The same reasoning leads also to  $q_{01}^\star\leq q_{02}^\star$. Hence $q_{01}^\star=q_{02}^\star$. 
We will prove now that there exists unique $\rho^\star$ and $\eta^\star(\rho^\star,q_0^\star)$ such that:
$$
\min_{\substack{-1\leq \rho\leq 1\\ \eta\in\mathbb{R}}} D_H(q_0^\star,\rho,\eta)=D_H(q_0^\star,\rho^\star,\eta^\star(\rho^\star,q_0^\star))=0.  
$$
Function \begin{align}\phi:(\rho,\eta)\mapsto &\Big(\frac{\delta \pi_1}{\sqrt{2\pi}}\int_{\frac{\rho\mu}{\sigma}-\frac{1}{\sigma q_0^\star}+\eta}^\infty\left(x-\frac{\rho\mu}{\sigma}+\frac{1}{\sigma q_0^\star}-\eta\right)^2Dx\nonumber\\
&+\frac{\delta \pi_0}{\sqrt{2\pi}}\int_{\frac{\rho\mu}{\sigma}-\frac{1}{\sigma q_0^\star}-\eta}^\infty\left(x-\frac{\rho\mu}{\sigma}+\frac{1}{\sigma q_0^\star}+\eta\right)^2Dx\Big)^{\frac{1}{2}}\end{align} is jointly convex in its arguments. Hence, $\rho\mapsto\min_{\eta\in\mathbb{R}} \phi(\rho,\eta)$ is  convex in $[-1,1]$.  As $\rho\mapsto -\sqrt{1-\rho^2}$ is strictly convex in $[-1,1]$, then $\rho\mapsto \min_{\eta\in\mathbb{R}}D_H(q_0^\star,\rho,\eta)$ is strictly  convex in $[-1,1]$. Assume that there exists $\rho^\star$ and $\tilde{\rho}^\star$ in $[-1,1]$ such that:
\begin{align}
	\min_{\eta\in\mathbb{R}}D_H(q_0^\star,\rho^\star,\eta)&=\min_{\eta\in\mathbb{R}}D_H(q_0^\star,\tilde{\rho}^\star,\eta)\\
	&=\min_{\substack{-1\leq \rho\leq 1\\ \eta\in\mathbb{R}}} D_H(q_0^\star,\rho,\eta)=0.
\end{align}
Let $\lambda\in(0,1)$. Assume $\rho^\star\neq \tilde{\rho}^\star$. Then 
\begin{align}
	&\min_{\eta\in\mathbb{R}}D_H(q_0^\star,\lambda\rho^\star+(1-\lambda)\tilde{\rho}^\star,\eta)\\
	&< \lambda\min_{\eta\in\mathbb{R}}D_H(q_0^\star,\rho^\star,\eta)+(1-\lambda) \min_{\eta\in\mathbb{R}}D_H(q_0^\star,\tilde{\rho}^\star,\eta)= 0
\end{align}
We obtain thus a contradiction, since $0=\min_{\substack{\eta\in\mathbb{R}}\\ \rho\in[-1,1]}D_H(q_0^\star,\rho,\eta)$. Hence the uniqueness of the minimizer $\rho^\star$.  
Combining all the above results shows the uniqueness of $q_0^\star, \eta^\star$ and $\rho^\star$. With the uniqueness of these parameters at hand, we will now proceed to the proof of the convergence result. Towards this goal, it suffices to prove  the following convergences for any $\epsilon>0$ and $k$ sufficiently large,
\begin{align}
& \mathbb{P}\left[\lim_{m\to\infty} \phi_{k,k,m}\leq(q_0^\star)^2-\epsilon,i.o\right]=0,  \label{eq:lower_bound_H} \\
&  \mathbb{P}\left[\lim_{m\to\infty} \phi_{k,k,m}\geq(q_0^\star)^2+\epsilon,i.o\right]=0. \label{eq:upper_bound_H}
\end{align}
that establish respectively a high-probability lower bound and high-probability upper bound on $\lim_{m\to\infty} \phi_{k,k,m}$. 
Assume that \eqref{eq:lower_bound_H} and \eqref{eq:upper_bound_H} hold true. 
We will prove that they translate into $\mathbb{P}\left[\Phi^{(n)}\leq (q_0^\star)^2-\epsilon,i.o\right]=0$ and $\mathbb{P}\left[\Phi^{(n)}\geq (q_0^\star)^2+\epsilon,i.o\right]=0$, the combination of both of which leads to $\lim_{n\to\infty}\Phi^{(n)}\to (q_0^\star)^2$ almost surely. 

Let us start by proving $\mathbb{P}\left[\Phi^{(n)}\leq (q_0^\star)^2-\epsilon,i.o\right]=0$.
For $\epsilon>0$ sufficiently small, there exists $k\in\mathbb{N}$ such that $\Phi^{(n)} \geq \Phi_{k,k}^{(n)}-\frac{\epsilon}{2}$. Hence,
\begin{align*}
	&\mathbb{P}\left[\Phi^{(n)}\leq (q_0^\star)^2-\epsilon\right] \leq \mathbb{P}\left[\cup_{k=1}^\infty \left\{\Phi_{k,k}^{(n)}\leq (q_0^\star)^2-\frac{\epsilon}{2}\right\}\right]\\
&\leq \mathbb{P}\left[\cup_{k=1}^\infty \cap_{m=1}^\infty \left\{\Phi_{k,k,m}^{(n)}\leq (q_0^\star)^2-\frac{\epsilon}{2}\right\}\right]\\
&\stackrel{(a)}{=}\lim_{k\to\infty}\lim_{m\to\infty} \mathbb{P}\left[\left\{\Phi_{k,k,m}^{(n)}\leq (q_0^\star)^2-\frac{\epsilon}{2}\right\}\right]
\end{align*}
where $(a)$ follows from the fact that the sequence of events $\left\{\cap_{m=1}^\infty \Phi_{k,k,m}^{(n)}\leq (q_0^\star)^2-\frac{\epsilon}{2}\right\}_{k\in\mathbb{N}^\star}$ forms an increasing sequence of events, while $\left\{\Phi_{k,k,m}^{(n)}\leq  (q_0^\star)^2-\frac{\epsilon}{2}\right\}_{m\in\mathbb{N}^\star}$ forms a decreasing sequence of events. 
Using the CGMT Theorem, we have:
$$
\mathbb{P}\left[\left\{\Phi_{k,k,m}^{(n)}\leq (q_0^\star)^2-\frac{\epsilon}{2}\right\}\right]\leq 2\mathbb{P}\left[\left\{\phi_{k,k,m}^{(n)}\leq (q_0^\star)^2-\frac{\epsilon}{2}\right\}\right]
$$
Hence
\begin{align}
	&\lim_{k\to\infty}\lim_{m\to\infty} \mathbb{P}\left[\left\{\Phi_{k,k,m}^{(n)}\leq (q_0^\star)^2-\frac{\epsilon}{2}\right\}\right]\nonumber\\
	&\leq\lim_{k\to\infty}\lim_{m\to\infty} \mathbb{P}\left[\left\{\phi_{k,k,m}^{(n)}\leq (q_0^\star)^2-\frac{\epsilon}{2}\right\}\right]\nonumber\\
	&= \mathbb{P}\left[\cup_{k=1}^\infty \cap_{m=1}^\infty \left\{\phi_{k,k,m}^{(n)}\leq (q_0^\star)^2-\frac{\epsilon}{2}\right\}\right]\nonumber
\end{align}
Let $k_0$ be an integer chosen such that $k_0> q_0^\star$. Hence, 
\begin{align*}
	&\mathbb{P}\left[\cup_{k=1}^\infty \cap_{m=1}^\infty \left\{\phi_{k,k,m}^{(n)}\leq (q_0^\star)^2-\frac{\epsilon}{2}\right\}\right]\nonumber\\
	&\leq \mathbb{P}\left[\cup_{k=k_0}^{\infty}\cap_{m=1}^\infty \left\{\phi_{k,k,m}^{(n)}\leq (q_0^\star)^2-\frac{\epsilon}{2}\right\}\right]\nonumber\\
	&=\mathbb{P}\left[\cup_{k=k_0}^\infty\forall m, \phi_{k,k,m}^{(n)}\leq(q_0^\star)^2-\frac{\epsilon}{2} \right]\nonumber\\
&\leq\mathbb{P}\left[\cup_{k=k_0}^\infty\lim_{m\to\infty} \phi_{k,k,m}^{(n)}\leq(q_0^\star)^2-\frac{\epsilon}{2} \right]
\end{align*}
It entails from \eqref{eq:lower_bound_H} that the set $B_n^{(AO)}\triangleq\cup_{k=1}^\infty\cap_{m=1}^\infty  \left\{\phi_{k_0,k_0,m}^{(n)}\leq(q_0^\star)^2-\frac{\epsilon}{2}\right\}$ does not occur infinitely often, or in other words $\mathbb{P}\left[B_n^{(AO)},i.o\right]=0.$ Since $(B_n^{(AO)})$ are independent, from the converse of Borel-Cantelli Lemma, we have:
$$
\sum_{n=1}^\infty\mathbb{P}(B_n^{(AO})<\infty.
$$
Hence, we have $B_n^{(P)}\triangleq\cup_{k=1}^\infty\cap_{m=1}^\infty  \left\{\Phi_{k,k,m}^{(n)}\leq(q_0^\star)^2-\frac{\epsilon}{2}\right\}$ satisfies $\sum_{n=1}^\infty\mathbb{P}\left[B_n^{(P)}\right]<\infty$, implying that the set $\left\{\Phi^{(n)}\leq(q_0^\star)^2-\frac{\epsilon}{2}\right\}$ does not occur infintely often. 

We will now consider proving $\mathbb{P}\left[\Phi^{(n)}\geq (q_0^\star)^2+\epsilon,i.o\right]=0$. From the fact that  for all $k\in\mathbb{N}$,  $\Phi^{(n)}\leq  \Phi_{k,k}^{(n)}$, we thus bound $\mathbb{P}\left[\Phi^{(n)}\geq (q_0^\star)^2+\epsilon\right]$ as:
\begin{align*}
	&\mathbb{P}\left[\Phi^{(n)}\geq (q_0^\star)^2+\epsilon\right]\\
	&\leq \mathbb{P}\left[\cap_{k=1}^\infty\left\{\Phi_{k,k}^{(n)}\geq  (q_0^\star)^2+\epsilon\right\}\right]
\end{align*}
For each $k\in\mathbb{N}^\star$, there exists $m$ such that: 
$$
\Phi_{k,k}^{(n)}\leq \Phi_{k,k,m}^{(n)}-\frac{\epsilon}{2}.
$$
Hence,
\begin{align*}
	\mathbb{P}\left[\Phi^{(n)}\geq (q_0^\star)^2+\epsilon\right]&\leq \mathbb{P}\left[\cap_{k=1}^\infty\cup_{m=1}^\infty\left\{\Phi_{k,k}^{(n)}\geq  (q_0^\star)^2+\epsilon\right\}\right]\\
	&=\lim_{k\to\infty}\lim_{m\to\infty} \mathbb{P}\left[\left\{\Phi_{k,k,m}^{(n)}\geq  (q_0^\star)^2+\frac{\epsilon}{2}\right\}\right]
\end{align*}
Using the CGMT theorem, we have:
\begin{align*}
	&\lim_{k\to\infty}\lim_{m\to\infty} \mathbb{P}\left[\left\{\Phi_{k,k,m}^{(n)}\geq  (q_0^\star)^2+\frac{\epsilon}{2}\right\}\right]\\
	&\leq  2\lim_{k\to\infty}\lim_{m\to\infty} \mathbb{P}\left[\left\{\phi_{k,k,m}^{(n)}\geq  (q_0^\star)^2+\frac{\epsilon}{2}\right\}\right]\\
	&=2\mathbb{P}\left[\cap_{k=1}^\infty\cup_{m=1}^\infty\left\{\phi_{k,k,m}^{(n)}\geq  (q_0^\star)^2+\frac{\epsilon}{2}\right\}\right]
\end{align*}
Let $k_0$ be an integer such that $k_0> q_0^\star$. Hence, 
\begin{align*}
	&\mathbb{P}\left[\cap_{k=1}^\infty\cup_{m=1}^\infty\left\{\phi_{k,k,m}^{(n)}\geq  (q_0^\star)^2+\epsilon\right\}\right]\\
	&\leq \mathbb{P}\left[\cup_{m=1}^\infty\left\{\phi_{k_0,k_0,m}^{(n)}\geq  (q_0^\star)^2+\epsilon\right\}\right]\\
	&\leq \mathbb{P}\left[\lim_{m\to\infty} \phi_{k_0,k_0,m}\geq (q_0^\star)^2+\frac{\epsilon}{2}\right]
\end{align*}
With this at hand, we can in a similar way as before invoke the Converse of Borel-Cantelli Lemma to prove that $\left\{\Phi^{(n)}\geq (q_0^\star)^2+\frac{\epsilon}{2}\right\}$ does not occur infinitely often.

So far, we have thus proven that establishing \eqref{eq:lower_bound_H} and \eqref{eq:upper_bound_H} leads to proving that $\Phi^{(n)}\to (q_0^\star)^2 $ almost surely. 
We will now proceed to the proof of \eqref{eq:lower_bound_H} and \eqref{eq:upper_bound_H}.
From \eqref{eq:lower_bound} and the discussion following it, we have:
\begin{align*}
	& \mathbb{P}\left[\lim_{m\to\infty} \phi_{k,k,m}\leq(q_0^\star)^2-\epsilon\right] \\
	&\leq \mathbb{P}[\min_{\substack{q_U\leq q_0\leq k\\ \hat{D}_H(q_0,\rho,\eta)\leq 0\\ -1\leq \rho \leq 1\\ \eta\in\mathbb{R}} } q_0^2\leq (q_0^\star)^2-\epsilon]\\
&=\mathbb{P}[\min_{{\substack{q_U\leq q_0\leq k\\ -1\leq \rho \leq 1\\{{\scriptscriptstyle\displaystyle\min_{\eta\in\mathbb{R}} \hat{D}_H(q_0,\rho,\eta)\leq 0 }}}}} q_0^2\leq (q_0^\star)^2-\epsilon]
\end{align*}
Function $\eta\mapsto \hat{D}_H(q_0,\rho,\eta)$ is  convex in $\eta$ and conveges pointwise to $D_H(q_0,\rho,\eta)$. Since $\lim_{\eta\to\infty} D_H(q_0,\rho,\eta)=\infty$ and $\lim_{\eta\to-\infty} D_H(q_0,\rho,\eta)=\infty$ , from Lemma 10 in \cite{thrampoulidis-IT}, we have:
\begin{equation}
\min_{\eta\in\mathbb{R}}  \hat{D}_H(q_0,\rho,\eta) \to  \min_{\eta\in\mathbb{R}}{D}_H(q_0,\rho,\eta)
\label{eq:conver}
\end{equation}
Now, Function $(q_0,\rho)\mapsto \min_{\eta\in\mathbb{R}}  \hat{D}_H(q_0,\rho,\eta)$ is convex in its arguments, and converges pointwise from \eqref{eq:conver} to $(q_0,\rho)\mapsto \min_{\eta\in\mathbb{R}}D_H(q_0,\rho,\eta)$. As the pointwise convergence of convex functions implies uniform convergence in compact sets, we have:
$$
\sup_{\substack{-1\leq \rho \leq 1\\ q_U\leq q_0\leq k}}\left|\min_{\eta\in\mathbb{R}}\hat{D}_H(q_0,\rho,\eta) -  \min_{\eta\in\mathbb{R}}{D}_H(q_0,\rho,\eta)\right|\asto 0
$$
For all $\delta$ sufficiently small, we can choose $n$, and $p$ sufficiently large such that for all $q_0\in[q_U,k]$ and $\rho\in[-1,1]$, 
$$
\min_{\eta\in\mathbb{R}}D_H(q_0,\rho,\eta)-\delta \leq \min_{\eta\in\mathbb{R}}\hat{D}_H(q_0,\rho,\eta)\leq \min_{\eta\in\mathbb{R}}D_H(q_0,\rho,\eta)+\delta 
$$
We thus have:
\begin{align*}
	&\mathbb{P}\left[\lim_{m\to\infty} \phi_{k,k,m}\leq(q_0^\star)^2-\epsilon\right] \\
	&\leq\mathbb{P}[\min_{{\substack{q_U\leq q_0\leq k\\ -1\leq \rho \leq 1\\{{\scriptscriptstyle\displaystyle\min_{\eta\in\mathbb{R}} {D}_H(q_0,\rho,\eta)\leq \delta }}}}} q_0^2\leq (q_0^\star)^2-\epsilon]\\
&\stackrel{(a)}{=}\mathbb{P}[\min_{{\substack{q_U\leq q_0\leq k\\{{\scriptscriptstyle\displaystyle\min_{\substack{\eta\in\mathbb{R}\\ -1\leq \rho \leq 1}} {D}_H(q_0,\rho,\eta)\leq \delta }}}}} q_0^2\leq (q_0^\star)^2-\epsilon]
\end{align*}
Before going further, it is noteworthy to mention that the right-hand side event is casted in the form of a determinstic statement that does not involve any random variables. It suffices that to check that for $\delta$ sufficiently small,  this statement is false. This can be easily checked using Lemma \ref{lem:equality} which enables to show that there exists $\delta_0$ such that for all $\delta\leq \delta_0$,
$$
\min_{{\substack{q_U\leq q_0\leq k\\{{\scriptscriptstyle\displaystyle\min_{\substack{-1\leq \rho\leq 1\\\eta\in\mathbb{R}}} {D}_H(q_0,\rho,\eta)\leq \delta }}}}} q_0^2\geq\min_{{\substack{q_U\leq q_0\leq k\\{{\scriptscriptstyle\displaystyle\min_{\substack{-1\leq \rho\leq 1\\\eta\in\mathbb{R}}} {D}_H(q_0,\rho,\eta)\leq 0}}}}} q_0^2 -\frac{\epsilon}{2}$$
This proves that :
$$
 \mathbb{P}\left[\lim_{m\to\infty} \phi_{k,k,m}\leq(q_0^\star)^2-\epsilon,i.o\right]=0.
$$
We will now proceed to the proof of \eqref{eq:upper_bound_H} for $k$ chosen such that $k^2\geq \eta^\star$ and $k\geq q_0^\star$.  To begin with, we recall that:
\begin{align}
	\lim_{m\to\infty}\phi_{k,k,m}^{(n)}&= \lim_{m\to\infty} \min_{\substack{0\leq q_0\leq k,  \  -1\leq \rho\leq 1\\ -kq_0\leq \eta\leq kq_0}} q_0^2+mq_0\hat{D}_H(q_0,\rho,\eta) \\
	&\leq \lim_{m\to\infty} \min_{\substack{0\leq q_0\leq k,  \  -1\leq \rho\leq 1\\ -kq_0\leq \eta\leq kq_0}} q_0^2+mk\hat{D}_H(q_0,\rho,\eta)
\end{align}
Due to the condition \eqref{eq:cond_hard_margin}, we know that the set  $$\left\{(q_0,\rho,\eta)\in([0,k]\times [-1,1]\times (-kq_0,kq_0)) \ |  \ \hat{D}_H(q_0,\rho,\eta)\leq 0\right\}$$ is almost surely non-empty. Hence, with probability $1$, for sufficiently large $n$ and $p$,
$$
\lim_{m\to\infty}\phi_{k,k,m}^{(n)} =\min_{\substack{0\leq q_0\leq k\\ -1\leq \rho\leq 1 \\  \hat{D}_H(q_0,\rho,\eta)\leq 0\\ -kq_0\leq \eta\leq kq_0}} q_0^2
$$
Function $(q_0,\rho,\eta)\mapsto \hat{D}_H(q_0,\rho,\eta)$ is jointly convex in its arguments and converges pointwise to $(q_0,\rho,\eta)\mapsto {D}_H(q_0,\rho,\eta)$. Hence it converges uniformly to  $(q_0,\rho,\eta)\mapsto {D}_H(q_0,\rho,\eta)$ over compact sets. Thus for all $\delta$ sufficiently small, we can select $n$ and $p$ sufficiently large such that for all $0\leq q_0 \leq k$, $-1\leq \rho\leq 1$ and $-kq_0\leq \eta\leq kq_0$, 
$$
D_H(q_0,\rho,\eta)-\delta \leq \hat{D}_H(q_0,\rho,\eta)\leq D_H(q_0,\rho,\eta)+\delta
$$
Using the above inequality, we thus obtain with probability $1$ for sufficiently large $n$ and $p$:
$$
\lim_{m\to\infty}\phi_{k,k,m}^{(n)} \leq\min_{\substack{0\leq q_0\leq k\\ -1\leq \rho\leq 1 \\  {D}_H(q_0,\rho,\eta)\leq -\delta\\ -kq_0\leq \eta\leq kq_0}} q_0^2
$$
Similar to the proof of \eqref{eq:lower_bound_H}, invoking Lemma \ref{lem:equality}, we show that there exists $\delta_0$ such that for all $\delta\leq \delta_0$, 
$$
\lim_{m\to\infty}\phi_{k,k,m}^{(n)} \leq \min_{\substack{0\leq q_0\leq k\\ -1\leq \rho\leq 1 \\  {D}_H(q_0,\rho,\eta)\leq 0\\ -kq_0\leq \eta\leq kq_0}} q_0^2 +\frac{\epsilon}{2}
$$
We can easily check that $(q_0^\star,\rho^\star,\eta^\star)$ are in the constraint set of the above optimization problem. Hence, with probability $1$ for sufficiently large $n$ and $p$:
$$
\lim_{m\to\infty}\phi_{k,k,m}^{(n)}\leq (q_0^\star)^2 +\frac{\epsilon}{2}
$$
Hence, 
$$
\mathbb{P}\left[\lim_{m\to\infty}\phi_{k,k,m}^{(n)}\geq (q_0^\star)^2+\epsilon,i.o.\right]=0.
$$
\subsection{Soft-margin SVM}
\subsubsection{Identification of the primary optimization (PO) and auxiliary optimization (AO) problems }
To begin with, we introduce the Lagrangian associated with the soft-margin problem:
\begin{align}
	&L_{S}({\bf w},\left\{\xi_i\right\}_{i=1}^n,\left\{\lambda_i\right\}_{i=1}^n)={\bf w}^{T}{\bf w}+\frac{\tilde{\tau}}{p} \sum_{i=1}^n\xi_i\nonumber\\
&+\sum_{i=1}^n\lambda_i(1-\xi_i-y_i((y_i\boldsymbol{\mu}+\sigma{\bf z}_i)^{T}{\bf w})+b))
\end{align}
Define $\boldsymbol{\lambda}=\left\{\lambda_i\right\}_{i=1}^n$, $\boldsymbol{\xi}=\left\{\xi_i\right\}_{i=1}^n$ and ${\bf Z}=\left[-y_1{\bf z}_1,\dots,-y_n{\bf z}_n\right]^{T}$. 
We need thus to solve the following problem:
\begin{align*}
	&\min_{\substack{{\bf w}},b}\min_{\boldsymbol{\xi}\geq 0} \max_{\boldsymbol{\lambda}\geq 0}{\bf w}^{T}{\bf w}+\sum_{i=1}^n\xi_i(\frac{\tilde{\tau}}{p}-\lambda_i)+\boldsymbol{\lambda}^{T}\boldsymbol{1}-\boldsymbol{\lambda}^{T}{\bf 1}\boldsymbol{\mu}^{T}{\bf w} \\
	&+\sigma\boldsymbol{\lambda}^{T}{\bf Z}{\bf w}-\boldsymbol{\lambda}^{T}({\bf j}_1-{\bf j}_0)b\\
&=\min_{{\bf w},b}\min_{\boldsymbol{\xi}\geq 0} \max_{\boldsymbol{\lambda}\geq 0} {\bf w}^{T}{\bf w}+\sum_{i=1}^n\xi_i(\frac{\tilde{\tau}}{p}-\lambda_i)+\boldsymbol{\lambda}^{T}\boldsymbol{1}-\boldsymbol{\lambda}^{T}{\bf 1}\boldsymbol{\mu}^{T}{\bf w}\\
	&+\sigma\boldsymbol{\lambda}^{T}{\bf Z}{\bf w}-\boldsymbol{\lambda}^{T}({\bf j}_1-{\bf j}_0)b
\end{align*}
Let $\xi^\star$ be the optimum. Then, from the first order conditions, we have, for all $\boldsymbol{\xi}\geq 0$
$$
\frac{\partial L_S({\bf w},\left\{\xi_i\right\}_{i=1}^n,\left\{\lambda_i\right\}_{i=1}^n)}{\partial \xi_j}(\xi_j-\xi_j^\star)=(\frac{\tilde{\tau}}{p}-\lambda_j)(\xi_j-\xi_j^\star)\geq 0
$$
For this condition to always hold for all $\boldsymbol{\xi}\geq 0$, we need that $(\frac{\tilde{\tau}}{p}-\lambda_j)\geq 0$ and  $(\frac{\tilde{\tau}}{p}-\lambda_j)\xi_j^\star=0$.  
Hence, the problem becomes:
\begin{align*}
	\min_{{\bf w},b}\min_{\boldsymbol{\xi}\geq 0}\max_{0\leq \boldsymbol{\lambda}\leq \frac{\tilde{\tau}}{p}}& {\bf w}^{T}{\bf w} +\boldsymbol{\lambda}^{T}\boldsymbol{1}-\boldsymbol{\lambda}^{T}{\bf 1}\boldsymbol{\mu}^{T}{\bf w}\\
	&+\sigma\boldsymbol{\lambda}^{T}{\bf Z}{\bf w}-\boldsymbol{\lambda}^{T}({\bf j}_1-{\bf j}_0)b
\end{align*}
Let us consider the change of variable $\tilde{\lambda}_i=\sigma\sqrt{p}\lambda_i$. Then, the above problem can be written as:
\begin{align}
	\min_{{\bf w},b} \max_{0\leq \tilde{\boldsymbol{\lambda}}\leq \frac{\tilde{\tau}}{\sqrt{p}}}& {\bf w}^{T}{\bf w} + \frac{1}{\sqrt{p}\sigma}\tilde{\boldsymbol{\lambda}}^{T}\boldsymbol{1}-\frac{1}{\sqrt{p}\sigma}\tilde{\boldsymbol{\lambda}}^{T}{\bf 1}\boldsymbol{\mu}^{T}{\bf w}\nonumber\\
	&+\frac{1}{\sqrt{p}}\tilde{\boldsymbol{\lambda}}^{T}{\bf Z}{\bf w}-\frac{1}{\sqrt{p}\sigma}\tilde{\boldsymbol{\lambda}}^{T}({\bf j}_1-{\bf j}_0)b
\end{align}
We need to prove that we can assume that there exists a constant $C_w$ and $C_b$ such that $\|{\bf w}\|_2\leq C$ and $b<C_b$. From the first order optimality conditions, we have:
$$
{\bf w}^\star=\frac{1}{2}\frac{1}{\sqrt{p}\sigma} \tilde{\boldsymbol{\lambda}}^{T}{\bf 1}\boldsymbol{\mu} -\frac{1}{2}\frac{1}{\sqrt{p}}{\bf Z}\tilde{\boldsymbol{\lambda}}
$$ 
Hence,
$$
\|{\bf w}^\star\| \leq \frac{1}{2\sigma}\|\boldsymbol{\mu}\|\|\tilde{\boldsymbol{\lambda}}^{T}\frac{\bf 1}{\sqrt{p}}\| +  \frac{1}{2}\|\frac{1}{\sqrt{p}}{\bf Z}\tilde{\boldsymbol{\lambda}}\|
$$
As $\|\frac{1}{\sqrt{p}}{\bf Z}\tilde{\boldsymbol{\lambda}}\|\leq \|\frac{1}{\sqrt{p}}{\bf Z}\| \| \tilde{\boldsymbol{\lambda}}\|_2 $ and  $\|\frac{1}{\sqrt{p}}{\bf Z}\|$ is almost surely bounded from results of random matrix theory, we conclude that we can assume without changing the analysis that there exists a constant $C_w$ such that $\|{\bf w}\|\leq C_w$. 
Similarly, from the first order optimality conditions, 
$$
\frac{1}{\sqrt{p}\sigma}{\bf 1}-\frac{1}{\sqrt{p}\sigma}\boldsymbol{\mu}^{T}{\bf w}+\frac{1}{\sqrt{p}}{\bf Zw}-\frac{1}{\sqrt{p}\sigma}({\bf j}_1-{\bf j}_0)b=0.
$$
which gives:
\begin{align}
	|b|&=\frac{\norm{\frac{1}{\sqrt{p}}{\bf Zw}-\frac{1}{\sqrt{p}\sigma}{\bf 1}\boldsymbol{\mu}^{T}{\bf w}+\frac{1}{\sqrt{p}\sigma}{\bf 1}}_2}{\norm{\frac{1}{\sqrt{p}\sigma}({\bf j}_1-{\bf j}_0)}_2}\\
	&=\frac{\sigma\sqrt{p}}{\sqrt{n}} \norm{\frac{1}{\sqrt{p}}{\bf Zw}-\frac{1}{\sqrt{p}\sigma}{\bf 1}\boldsymbol{\mu}^{T}{\bf w}+\frac{1}{\sqrt{p}\sigma}{\bf 1}}_2
\end{align}
Hence, there exists $C_b$ such that $|b|\leq C_b$.
We will thus consider from now on solving the following primary optimization problem:
\begin{align*}
	\Phi_S^{(n)}&\triangleq\min_{\substack{{\bf w}\\ \|{\bf w}\|\leq C_w}}\min_{|b|\leq C_b} \max_{0\leq \tilde{\boldsymbol{\lambda}}\leq \frac{\tilde{\tau}}{\sqrt{p}}} {\bf w}^{T}{\bf w} + \frac{1}{\sqrt{p}\sigma}\tilde{\boldsymbol{\lambda}}^{T}\boldsymbol{1}-\frac{1}{\sqrt{p}\sigma}\tilde{\boldsymbol{\lambda}}^{T}{\bf 1}\boldsymbol{\mu}^{T}{\bf w} \\
	&+\frac{1}{\sqrt{p}}\tilde{\boldsymbol{\lambda}}^{T}{\bf Z}{\bf w}-\frac{1}{\sqrt{p}\sigma}\tilde{\boldsymbol{\lambda}}^{T}({\bf j}_1-{\bf j}_0)b
\end{align*}
It is important to mention that constants $C_b$ and $C_w$ can be set to any finite constants, the values of which can be as large as desired. More details on how these constants need to be set will be provided in the next section. 

Instead of analyzing the optimization problem above, we will analyze a simpler auxiliary optimization problem that is tightly related to the (PO) via the CGMT.
Having identified the (PO), it is easy to write the corresponding (AO) problem as:
\begin{align}
	\phi_{S}^{(n)}&\triangleq\min_{\substack{{\bf w}\\ \|{\bf w}\|\leq C}} \min_{|b|\leq C_b}\max_{0\leq {\bf u}\leq \frac{\tilde{\tau}}{\sqrt{p}}} \frac{1}{\sqrt{p}}\|{\bf w}\|_2{\bf g}^{T}{\bf u} - \frac{1}{\sqrt{p}}\|{\bf u}\|_2 {\bf h}^{T}{\bf w} \\
	&+ {\bf w}^{T}{\bf w} +\frac{1}{\sqrt{p}\sigma}{\bf u}^{T}{\bf 1}-\frac{1}{\sqrt{p}\sigma}{\bf u}^{T}{\bf 1}\boldsymbol{\mu}^{T}{\bf w}-\frac{1}{\sqrt{p}\sigma}{\bf u}^{T}({\bf j}_1-{\bf j}_0)b
\end{align}
\subsubsection{Simplification of the (AO) problem}
In the same way as in the hard-margin SVM, we start by decomposing ${\bf w}$ as
$$
{\bf w}=\alpha_1\frac{\boldsymbol{\mu}}{\|\boldsymbol{\mu}\|} +\alpha_2 {\bf w}_{\perp}
$$
where ${\bf w}_{\perp}$ is orthogonal to $\boldsymbol{\mu}$ and $\|{\bf w}_{\perp}\|=1$.
With this decomposition at hand, the (AO) problem becomes:
\begin{align}
	\phi_{S}^{(n)}&=\min_{\substack{{\bf w}=\alpha_1\frac{\boldsymbol{\mu}}{\norm{\boldsymbol{\mu}}_2}+\alpha_2 {\bf w}_{\perp}\\ \alpha_1^2+\alpha_2^2\leq C_w^2}}\min_{|b|\leq C_b}\max_{0\leq {\bf u}\leq \frac{\tilde{\tau}}{\sqrt{p}}} \frac{1}{\sqrt{p}}\sqrt{\alpha_1^2+\alpha_2^2}{\bf g}^{T}{\bf u}\nonumber\\
	&-\frac{1}{\sqrt{p}}\norm{{\bf u}}_2\left(\alpha_1 \frac{{\bf h}^{T}\boldsymbol{\mu}}{\norm{\boldsymbol{\mu}}}+\alpha_2 {\bf h}^{T}{\bf w}_{\perp}\right)+{\bf w}^{T}{\bf w}\nonumber \\
	&+\frac{1}{\sqrt{p}\sigma}{\bf u}^{T}{\bf 1}-\frac{1}{\sqrt{p}\sigma}{\bf u}^{T}{\bf 1}\boldsymbol{\mu}^{T}{\bf w}-\frac{1}{\sqrt{p}\sigma}{\bf u}^{T}({\bf j}_1-{\bf j}_0)b\nonumber
\end{align}
Using lemma~\ref{lem:prop}, we thus have:
\begin{align}
	\phi_{S}^{(n)}&=\min_{ \alpha_1^2+\alpha_2^2\leq C_w^2}\min_{|b|\leq C_b}\max_{0\leq {\bf u}\leq \frac{\tilde{\tau}}{\sqrt{p}}} \frac{1}{\sqrt{p}}\sqrt{\alpha_1^2+\alpha_2^2}{\bf g}^{T}{\bf u}\nonumber\\
	&-\frac{1}{\sqrt{p}}\norm{{\bf u}}_2\left(\alpha_1 \frac{{\bf h}^{T}\boldsymbol{\mu}}{\norm{\boldsymbol{\mu}}}+|\alpha_2| \norm{{\bf h}_{\perp}}_2\right)+q_0^2\nonumber \\
	&+\frac{1}{\sqrt{p}\sigma}{\bf u}^{T}{\bf 1}-\frac{1}{\sqrt{p}\sigma}{\bf u}^{T}{\bf 1}\alpha_1\norm{\boldsymbol{\mu}}_2-\frac{1}{\sqrt{p}\sigma}{\bf u}^{T}({\bf j}_1-{\bf j}_0)b\nonumber
\end{align}
Performing the change of variable $q_0=\sqrt{\alpha_1^2+\alpha_2^2}$, we thus obtain:
\begin{align*}
	\tilde{\phi}_{S}^{(n)}&=\min_{\substack{0\leq q_0\leq C_w\\ |\alpha_1|\leq q_0}}\min_{|b|\leq C_b}\max_{0\leq {\bf u}\leq \frac{\tilde{\tau}}{\sqrt{p}}} \frac{1}{\sqrt{p}}q_0{\bf g}^{T}{\bf u}\\
	&-\frac{1}{\sqrt{p}}\norm{{\bf u}}_2\left(\alpha_1 \frac{{\bf h}^{T}\boldsymbol{\mu}}{\norm{\boldsymbol{\mu}}}+\sqrt{q_0^2-\alpha_1^2} \norm{{\bf h}_{\perp}}_2\right)+q_0^2\nonumber \\
	&+\frac{1}{\sqrt{p}\sigma}{\bf u}^{T}{\bf 1}-\frac{1}{\sqrt{p}\sigma}{\bf u}^{T}{\bf 1}\alpha_1\norm{\boldsymbol{\mu}}_2-\frac{1}{\sqrt{p}\sigma}{\bf u}^{T}({\bf j}_1-{\bf j}_0)b\nonumber
\end{align*}
\subsubsection{Asymptotic behavior of the (AO) problem}
Consider now the function $\tilde{\phi}_{S}^{(n)}$ given by:
\begin{align*}
	\tilde{\phi}_{S}^{(n)}&=\min_{\substack{0\leq q_0\leq C_w\\ |\alpha_1|\leq q_0}}\min_{|b|\leq C_b}\max_{0\leq {\bf u}\leq \frac{\tilde{\tau}}{\sqrt{p}}} \frac{1}{\sqrt{p}}q_0{\bf g}^{T}{\bf u}-\norm{{\bf u}}_2\sqrt{q_0^2-\alpha_1^2} +q_0^2\\
	&+\frac{1}{\sqrt{p}\sigma}{\bf u}^{T}{\bf 1}-\frac{1}{\sqrt{p}\sigma}{\bf u}^{T}{\bf 1}\alpha_1\norm{\boldsymbol{\mu}}_2-\frac{1}{\sqrt{p}\sigma}{\bf u}^{T}({\bf j}_1-{\bf j}_0)b\nonumber\\
&\stackrel{(a)}{=}\inf_{0< q_0\leq C_w}\min_{ 0\leq\alpha_1\leq q_0}\min_{|b|\leq C_b}\max_{0\leq {\bf u}\leq \frac{\tilde{\tau}}{\sqrt{p}}} \frac{1}{\sqrt{p}}q_0{\bf g}^{T}{\bf u}-\norm{{\bf u}}_2\sqrt{q_0^2-\alpha_1^2} \\
	&+q_0^2+\frac{1}{\sqrt{p}\sigma}{\bf u}^{T}{\bf 1}-\frac{1}{\sqrt{p}\sigma}{\bf u}^{T}{\bf 1}\alpha_1\norm{\boldsymbol{\mu}}_2-\frac{1}{\sqrt{p}\sigma}{\bf u}^{T}({\bf j}_1-{\bf j}_0)b\nonumber
\end{align*}
where (a) follows from Lemma \ref{lem:prop}. We note also that in $(a)$, we replaced $\min$ by $\inf$, since in (a), the minimization is now on the open set $(0,C_w]$. This is allowed due to continuity with respect to $q_0$ which can be proven using the Maximum Theorem {\cite[Theorem~9.17]{Sundaram}}. 
Obviously $\phi_{S}^{(n)}-\tilde{\phi}_{S}^{(n)}\to 0$ almost surely since $\alpha_1 \frac{{\bf h}^{T}\boldsymbol{\mu}}{\norm{\boldsymbol{\mu}}}$ and $\sqrt{q_0^2-\alpha_1^2}(\norm{{\bf h}_{\perp}}-1)$ converges uniformly to zero on the set over which the optimization is performed. With this approximation performed, we removed the randomnes induced by vector ${\bf h}$. The randomness of $\tilde{\phi}_{S}^{(n)}$ is now only due to vector ${\bf g}$. We are now ready to study the asymptotic behavior of $\tilde{\phi}_{S}^{(n)}$. 
For the moment, we assume given that the uniqueness of the solution to \eqref{min_max_soft_margin}, the proof of which is given in section \ref{sec:uniqueness}. We assume that $C_w$ and $C_b$ are taken such that $C_w> q_0^\star$ and $C_b> \frac{\eta^\star}{\sigma q_0^\star}$ and proceed with the following steps. First, we prove that for any $\epsilon>0$ and sufficiently large $n$ and $p$
\begin{equation}
\tilde{\phi}_S^{(n)}\geq \sup_{\xi> 0}\mathcal{D}_{S,\tilde{\tau}}(q_0^\star,\rho^\star,\eta^\star,\xi)-\epsilon=\inf_{\substack{q_0> 0 \\ \eta\in\mathbb{R}\\ 0}}\min_{0\leq \rho\leq 1}\sup_{\xi> 0}\mathcal{D}_{S,\tilde{\tau}}(q_0,\rho,\eta,\xi)-\epsilon
\label{eq:lower_bound_soft}
\end{equation}
where $q_0^\star$, $\eta^\star$ and $\rho^\star$ the unique solutions to \eqref{min_max_soft_margin}. In a similar way, we prove that:
\begin{equation}
\tilde{\phi}_S^{(n)}\leq \inf_{\substack{q_0> 0 \\ \eta\in\mathbb{R}\\ 0\leq \rho\leq 1}}\sup_{\xi>0}\mathcal{D}_{S,\tilde{\tau}}(q_0,\rho,\eta,\xi)+\epsilon
\label{eq:upper_bound_soft}
\end{equation}
Combining \eqref{eq:lower_bound_soft} and \eqref{eq:upper_bound_soft} yields:
$$
\tilde{\phi}_S^{(n)}-\inf_{\substack{q_0> 0 \\ \eta\in\mathbb{R}\\ 0\leq \rho\leq 1}}\sup_{\xi> 0}\mathcal{D}_{S,\tilde{\tau}}(q_0,\rho,\eta,\xi)\asto 0. 
$$
We start by lower-bounding $\tilde{\phi}_{S}^{(n)}$ as:
\begin{align*}
\tilde{\phi}_{S}^{(n)}&\geq \inf_{\substack{q_0> 0\\ 0\leq \alpha_1\leq q_0}}\inf_{b\in\mathbb{R}}\max_{0\leq {\bf u}\leq \frac{\tilde{\tau}}{\sqrt{p}}} \frac{1}{\sqrt{p}}q_0{\bf g}^{T}{\bf u}-\norm{{\bf u}}_2\sqrt{q_0^2-\alpha_1^2} +q_0^2\\
	&+\frac{1}{\sqrt{p}\sigma}{\bf u}^{T}{\bf 1}-\frac{1}{\sqrt{p}\sigma}{\bf u}^{T}{\bf 1}\alpha_1\norm{\boldsymbol{\mu}}_2-\frac{1}{\sqrt{p}\sigma}{\bf u}^{T}({\bf j}_1-{\bf j}_0)b
\end{align*}
Performing the change of variable $\tilde{\bf u}=q_0{\bf u}$, $\rho=\frac{\alpha_1}{q_0}$ and $\eta=\frac{b}{\sigma q_0}$, we thus obtain:
\begin{align*}
\tilde{\phi}_{S}^{(n)}&\geq\inf_{\substack{q_0> 0}}\min_{0\leq \rho\leq 1}\inf_{\eta\in\mathbb{R}}\max_{0\leq \tilde{\bf u}\leq q_0\frac{\tilde{\tau}}{\sqrt{p}}}\frac{1}{\sqrt{p}} {\bf g}^{T}\tilde{\bf u}- \norm{\tilde{\bf u}}_2 \sqrt{1-\rho^2} \\
	&+\frac{1}{\sqrt{p}\sigma q_0}\tilde{\bf u}^{T}{\bf 1}-\frac{1}{\sqrt{p}\sigma}\tilde{\bf u}^{T}{\bf 1}\rho \norm{\boldsymbol{\mu}}_2 -\frac{1}{\sqrt{p}}\tilde{\bf u}^{T}({\bf j}_1-{\bf j}_0)\eta
\end{align*}
Using Lemma \ref{lem:opt}, we perform  optimization over $\tilde{\bf u}$ thus yielding:
 \begin{align}
&\tilde{\phi}_{S}^{(n)}\geq \inf_{q_0> 0}\min_{0\leq \rho\leq 1}\min_{\eta\in\mathbb{R}}\sup_{\xi> 0} \hat{D}_{S,\tilde{\tau}}(q_0,\rho,\eta,\xi)
\label{eq:phi_lower}
\end{align}
where
\begin{equation}
\hat{D}_{S,\tilde{\tau}}(q_0,\rho,\eta,\xi)\!\!=\begin{cases*} q_0^2+q_0\hat{R}_{\tilde{\tau}}(\frac{1}{q_0},\rho,\eta,\xi), &\ \ \ \text{If} \ \ $0\leq\rho< 1$\\
q_0^2+q_0\tilde{R}_{\tilde{\tau}}(\frac{1}{q_0},\eta),&\  \  \  \text{If} \  $\rho\!=\!1$
\end{cases*}
\end{equation}
with  $\hat{R}_{\tilde{\tau}}(x,\rho,\eta,\xi)$ and $\tilde{R}_{\tilde{\tau}}(x,\eta)$ given by:
\begin{align}
	&\hat{R}_{\tilde{\tau}}(x,\rho,\eta,\xi)=\left\{\displaystyle\frac{1}{p}\sum_{i\in\mathcal{C}_1} \left(g_i\tilde{\tau}+\frac{\tilde{\tau}x}{\sigma}-\tilde{\tau}\rho\norm{\boldsymbol{\mu}}_2-\tilde{\tau}\eta-\frac{\tilde{\tau}^2}{2\xi}\right)_{+}\right. \nonumber\\
	&\left.+\frac{1}{p}\sum_{i\in\mathcal{C}_0} \left(g_i\tilde{\tau}+\frac{\tilde{\tau}x}{\sigma}-\frac{\tilde{\tau}}{\sigma}\rho\norm{\boldsymbol{\mu}}+\tilde{\tau}\eta-\frac{\tilde{\tau}^2}{2\xi}\right)_{+}\right.\nonumber\\
&+\frac{\xi}{2}\frac{1}{p}\sum_{i\in\mathcal{C}_1}\left(g_i+\frac{x}{\sigma }-\frac{1}{\sigma}\rho\norm{\boldsymbol{\mu}}_2-\eta\right)^2{\bf 1}_{\left\{0\leq (g_i+\frac{x}{\sigma }-\frac{1}{\sigma}\rho\norm{\boldsymbol{\mu}}_2-\eta)\leq \frac{\tilde{\tau}}{\xi}\right\}}\nonumber\\
&\left.+\frac{\xi}{2}\frac{1}{p}\sum_{i\in\mathcal{C}_0}\left(g_i+\frac{x}{\sigma }-\frac{1}{\sigma}\rho\norm{\boldsymbol{\mu}}_2+\eta\right)^2{\bf 1}_{\left\{0\leq (g_i+\frac{x}{\sigma }-\frac{1}{\sigma}\rho\norm{\boldsymbol{\mu}}_2+\eta)\leq \frac{\tilde{\tau}}{\xi}\right\}}\right. \nonumber\\
	&-\frac{\xi}{2}(1-\rho^2)\Big\}
\end{align}
and
\begin{align*}
	\tilde{R}_{\tilde{\tau}}(x,\eta)&=\left\{\displaystyle\frac{1}{p}\sum_{i\in\mathcal{C}_1} \left(g_i\tilde{\tau}+\frac{\tilde{\tau}x}{\sigma }-\frac{\tilde{\tau}}{\sigma}\norm{\boldsymbol{\mu}}_2-\tilde{\tau}\eta\right)_{+}\right.\\
	&\left.+\frac{1}{p}\sum_{i\in\mathcal{C}_0} \left(g_i\tilde{\tau}+\frac{\tilde{\tau}x}{\sigma }-\frac{\tilde{\tau}}{\sigma}\norm{\boldsymbol{\mu}}_2+\tilde{\tau}\eta\right)_{+}\right\},  
\end{align*}
The point-wise convergence of $\hat{D}_{S,\tilde{\tau}}(q_0,\rho,\eta,\xi)$ to $D_{S,\tilde{\tau}} (q_0,\rho,\eta,\xi)$ when $0\leq \rho<1$ can be shown using the strong law of large numbers. To transfer this convergence to the minimax of $\hat{D}_{S,\tilde{\tau}}(q_0,\rho,\eta,\xi)$, we will repeatedly invoke Lemma 10 in \cite{thrampoulidis-IT}. 
\begin{enumerate}
\item Proof of $\sup_{\xi> 0} \hat{D}_{S,\tilde{\tau}}(q_0,\rho,\eta,\xi)\to \sup_{\xi> 0}{D}_{S,\tilde{\tau}}(q_0,\rho,\eta,\xi)$.

Consider now function $\kappa:(x,\rho,\eta)\mapsto \max_{0\leq \tilde{\bf u}\leq \frac{\tilde{\tau}}{\sqrt{p}}} {\bf g}^{T}\tilde{\bf u}-\norm{\tilde{\bf u}}_2\sqrt{1-\rho^2} +\frac{x}{\sqrt{p}\sigma} \tilde{\bf u}^{T}{\bf 1} -\frac{1}{\sqrt{p}}\tilde{\bf u}^{T}{\bf 1}\rho\norm{\boldsymbol{\mu}}_2$ defined on $\mathbb{R}_{+}\times [0,1]\times \mathbb{R}$. 
It coincides thus with the supremum over an infinitely set of jointly convex functions in $(x,\rho,\eta)$ indexed by ${\tilde{\bf u}}$ and as such is jointly convex with respect to its arguments. It is easy to see using Lemma \ref{lem:opt} that:
\begin{equation}
\kappa(x,\rho,\eta)=\begin{cases*}
\sup_{\xi> 0}\hat{R}_{\tilde{\tau}}(x,\rho,\eta,\xi), & \text{If} $\rho\neq 1$\\
\tilde{R}_{\tilde{\tau}}(x,\eta), & \text{If} $\rho=1$.
\end{cases*}
\end{equation}
From Lemma \ref{lem:opt}, for a given $(x,\rho,\eta)$ in $\mathbb{R}_{+}\times [0,1)\times \mathbb{R}$,  function $\xi\mapsto \hat{R}_{\tilde{\tau}}(x,\rho,\eta,\xi)$ is thus concave.
For $0\leq \rho<1$, $(x,\rho,\eta,\xi)\mapsto \hat{R}_{\tilde{\tau}}(x,\rho,\eta,\xi)$ converges pointwise to $R_{\tilde{\tau}}(x,\rho,\eta,\xi)$ given by: 
\begin{align*}
&R_{\tilde{\tau}}(x,\rho,\eta,\xi)=\frac{\pi_1\delta}{\sqrt{2\pi}}\int_{\frac{\tilde{\tau}}{\xi}+\eta+\frac{\rho\mu}{\sigma}-\frac{x}{\sigma}}^{\infty} \tilde{\tau}\left(t+\frac{x}{\sigma}-\frac{\rho\mu}{\sigma}-\eta-\frac{\tilde{\tau}}{2\xi}\right)Dt\\
&+ \frac{\pi_0\delta}{\sqrt{2\pi}}\int_{\frac{\tilde{\tau}}{\xi}-\eta+\frac{\rho\mu}{\sigma}-\frac{x}{\sigma}}^{\infty} \tilde{\tau}\left(t+\frac{x}{\sigma}-\frac{\rho\mu}{\sigma}+\eta-\frac{\tilde{\tau}}{2\xi}\right)Dt-\frac{\xi}{2}(1-\rho^2)\\
&+\frac{\xi\pi_1\delta}{2\sqrt{2\pi}}\int_{-\frac{x}{\sigma}+\frac{\rho\mu}{\sigma}+\eta}^{-\frac{x}{\sigma}+\frac{\rho\mu}{\sigma}+\eta+\frac{\tilde{\tau}}{\xi}} \left(t+\frac{x}{\sigma}-\frac{\rho\mu}{\sigma}-\eta\right)^2Dt\\
&+\frac{\xi\pi_0\delta}{2\sqrt{2\pi}}\int_{-\frac{x}{\sigma}+\frac{\rho\mu}{\sigma}-\eta}^{-\frac{x}{\sigma}+\frac{\rho\mu}{\sigma}-\eta+\frac{\tilde{\tau}}{\xi}} \left(t+\frac{x}{\sigma}-\frac{\rho\mu}{\sigma}+\eta\right)^2Dt
\end{align*}
When $0\leq \rho< 1$,   $\lim_{\xi\to\infty}  R_{\tilde{\tau}}(x,\rho,\eta,\xi)=-\infty$. It entails thus  from Lemma 10 in \cite{thrampoulidis-IT}\footnote{Lemma 10 in \cite{thrampoulidis-IT} holds for convergence in probability but the generalization to almost sure convergence is straightforward.} that for  $0\leq \rho<1$,
$$
\sup_{\xi> 0} \hat{R}_{\tilde{\tau}}(x,\rho,\eta,\xi)\asto \sup_{\xi> 0} R(x,\rho,\eta,\xi), 
 $$
 From the maximum theorem \cite[Theorem~9.17]{Sundaram},  function $\rho\mapsto\kappa(x,\rho,\eta)$ is continuous. Therefore, $$\lim_{\rho\uparrow 1} \sup_{\xi> 0}\hat{R}_{\tilde{\tau}}(x,\rho,\eta,\xi) = \tilde{R}_{\tilde{\tau}}(x,\eta).$$ Moreover, $\tilde{R}_{\tilde{\tau}}(x,\eta)$ converges pointwise to $\overline{R}_{\tilde{\tau}}(x,\eta)$ given by:
		\begin{align}
			\overline{R}_{\tilde{\tau}}(x,\eta)&=\frac{\pi_1\delta}{\sqrt{2\pi}}\int_{\eta+\frac{\mu}{\sigma}-\frac{x}{\sigma}}^\infty \tilde{\tau}\left(t+\frac{x}{\sigma}-\frac{\mu}{\sigma}-\eta\right)Dt\\
			&+\frac{\pi_0\delta}{\sqrt{2\pi}}\int_{-\eta+\frac{\mu}{\sigma}-\frac{x}{\sigma}}^\infty \tilde{\tau}\left(t+\frac{x}{\sigma}-\frac{\mu}{\sigma}+\eta\right)\exp\left(-\frac{t^2}{2}\right)dt
		\end{align}
Now, for a given pair $(x,\eta)$ in $\mathbb{R}_{+}\times \mathbb{R}$, function $\rho \mapsto \sup_{\xi> 0}\hat{R}_{\tilde{\tau}}(x,\rho,\eta,\xi)$ is convex in $\rho\in[0,1)$ and converges to $\sup_{\xi > 0} R_{\tilde{\tau}}(x,\rho,\eta,\xi)$. The convergence is thus uniform on the interval $\rho\in\left[0,1\right)$, and as such:
$$
\lim_{\rho\uparrow 1} \sup_{\xi> 0} \hat{R}_{\tilde{\tau}}(x,\rho,\eta,\xi)\asto \lim_{\rho\uparrow 1} \sup_{\xi> 0} R_{\tilde{\tau}}(x,\rho,\eta,\xi)=\overline{R}_{\tilde{\tau}}(x,\eta)
$$
For $\rho=1$, function $R_{\tilde{\tau}}(x,1,\eta,\xi)$ is well defined. It is easy to see that it is increasing with respect to $\xi$, and as such $\sup_{\xi> 0} R_{\tilde{\tau}}(x,1,\eta,\xi)=\lim_{\xi\to\infty} R_{\tilde{\tau}}(x,1,\eta,\xi)$, which coincides as expected with $\overline{R}_{\tilde{\tau}}(x,\eta)$.
As a consequence, we have for all $\rho\in[0,1]$, $q_0>0$ and $\eta\in\mathbb{R}$
$$
\sup_{\xi> 0}\hat{D}_{S,\tilde{\tau}}(q_0,\rho,\eta,\xi)\asto \sup_{\xi> 0} D_{S,\tilde{\tau}}(q_0,\rho,\eta,\xi)$$
\item Proof of $\inf_{\eta\in\mathbb{R}}\sup_{\xi> 0} \hat{D}_{S,\tilde{\tau}}(x,\rho,\eta,\xi)\asto \inf_{\eta\in\mathbb{R}}\sup_{\xi> 0} {D}_{S,\tilde{\tau}}(x,\rho,\eta,\xi)$

Consider for a given $(x,\rho)\in \mathbb{R}_{>0}\times [0,1]$, function $\eta\mapsto \hat{R}_{\tilde{\tau}}^{x,\rho}(\eta)$ defined as:\begin{equation}\hat{R}_{\tilde{\tau}}^{x,\rho}(\eta)\mapsto \begin{cases}\sup_{\xi> 0} \hat{R}_{\tilde{\tau}}(x,\rho,\eta,\xi),& \text{If} \ \ \rho\neq 1 \\ 
\tilde{R}_{\tilde{\tau}}(x,\eta), &\text{If} \ \ \rho=1
\end{cases}
\end{equation}
and function $\eta\mapsto  {R}_{\tilde{\tau}}^{x,\rho}(\eta)$ given by:
\begin{equation}{R}_{\tilde{\tau}}^{x,\rho}(\eta)=\sup_{\xi> 0}R_{\tilde{\tau}}(x,\rho,\eta,\xi)
\end{equation}
Function $\eta\mapsto \hat{R}_{\tilde{\tau}}^{x,\rho}(\eta)$ converges pointwise to $\eta\mapsto{R}_{\tilde{\tau}}^{x,\rho}(\eta)$. In order to transfer this convergence into that of the minimum over $\eta\in\mathbb{R}$, we need to check that $\lim_{|\eta|\to\infty}  {R}_{\tilde{\tau}}^{x,\rho}(\eta)=\infty$. To this end, it suffices to observe that for any $\xi>0$, $\lim_{|\eta|\to\infty} R_{\tilde{\tau}}(x,\rho,\eta,\xi)=\infty$$. $
Hence,
\begin{align*}
\lim_{|\eta|\to\infty}	\sup_{\xi> 0}R_{\tilde{\tau}}(x,\rho,\eta,\xi)\to\infty. 
\end{align*}
With this result at hand, we can apply Lemma 10 in \cite{thrampoulidis-IT} to ensure that:
$$
\inf_{\eta\in\mathbb{R}} \hat{R}^{x,\rho}(\eta)\asto \inf_{\eta\in\mathbb{R}} R^{x,\rho}(\eta)
$$
As a consequence for any $q_0>0$ and $\rho\in[0,1]$, 
$$
\inf_{\eta\in\mathbb{R}}\sup_{\xi> 0} \hat{D}_{S,\tilde{\tau}}(q_0,\rho,\eta,\xi)\asto \inf_{\eta\in\mathbb{R}}\sup_{\xi> 0} {D}_{S,\tilde{\tau}}(q_0,\rho,\eta,\xi)
$$
\item Proof of $\min_{0\leq \rho\leq 1}\inf_{\eta\in\mathbb{R}}\sup_{\xi>0}\hat{D}_{S,\tilde{\tau}}(q_0,\rho,\eta,\xi)\asto \min_{0\leq \rho\leq 1}\inf_{\eta\in\mathbb{R}}\sup_{\xi>0} D_{S,\tilde{\tau}}(q_0,\rho,\eta,\xi)$.
Based on Lemma~\ref{lem:convexity}, function $\rho\mapsto \inf_{\eta\in\mathbb{R}} \hat{R}^{x,\rho}(\eta)$ is convex in $\rho$. It converges pointwise to $\rho \mapsto  \inf_{\eta\in\mathbb{R}} {R}^{x,\rho}(\eta)$. The convergence is uniform on $\rho\in[0,1]$, and as such:
$$
\min_{0\leq \rho\leq 1} \inf_{\eta\in\mathbb{R}} \hat{R}^{x,\rho}(\eta) \asto \min_{0\leq \rho\leq 1} \inf_{\eta\in\mathbb{R}}{R}^{x,\rho}(\eta).
$$
\item  Proof of $\inf_{q_0>0}\min_{0\leq \rho\leq 1}\inf_{\eta\in\mathbb{R}}\sup_{\xi>0}\hat{D}_{S,\tilde{\tau}}(q_0,\rho,\eta,\xi)\asto \inf_{q_0>0}\min_{0\leq \rho\leq 1}\inf_{\eta\in\mathbb{R}}\sup_{\xi>0}D_{S,\tilde{\tau}}(q_0,\rho,\eta,\xi)$.

Similarly, based on Lemma~\ref{lem:convexity}, $x\mapsto \min_{0\leq \rho\leq 1} \inf_{\eta\in\mathbb{R}} \hat{R}^{x,\rho}(\eta)$ is convex. Hence, its perspective function $q_0\mapsto q_0   \min_{0\leq \rho\leq 1} \inf_{\eta\in\mathbb{R}} \hat{R}^{\frac{x}{q_0},\rho}(\eta)$ is convex for any $x>0$. Particularly, $q_0\mapsto q_0   \min_{0\leq \rho\leq 1} \inf_{\eta\in\mathbb{R}} \hat{R}^{\frac{1}{q_0},\rho}(\eta)$ is convex. Noting that $  \min_{0\leq \rho\leq 1} \inf_{\eta\in\mathbb{R}}\sup_{\xi>0} \hat{D}(q_0,\rho,\eta,\xi)= q_0^2+q_0\min_{0\leq \rho\leq 1} \inf_{\eta\in\mathbb{R}}\hat{R}^{\frac{1}{q_0},\rho}(\eta)$, we conclude that $q_0\mapsto  \min_{0\leq \rho\leq 1} \inf_{\eta\in\mathbb{R}}\sup_{\xi>0} \hat{D}(q_0,\rho,\eta,\xi)$ is convex. As it tends to infinity as $q_0\to\infty$, we thus have from lemma 10 in \cite{thrampoulidis-IT} that 
\begin{align}
	&\inf_{q_0>0}\min_{0\leq \rho\leq 1}\inf_{\eta\in\mathbb{R}}\sup_{\xi>0}\hat{D}_{S,\tilde{\tau}}(q_0,\rho,\eta,\xi)\nonumber\\
	&\asto \inf_{q_0>0}\min_{0\leq \rho\leq 1}\inf_{\eta\in\mathbb{R}}\sup_{\xi>0}D_{S,\tilde{\tau}}(q_0,\rho,\eta,\xi)
\label{eq:usef}
\end{align}
\end{enumerate}
Recall now that:
$$
\tilde{\phi}_{S}^{(n)}\geq \inf_{q_0>0}\min_{0\leq \rho\leq 1}\inf_{\eta\in\mathbb{R}}\sup_{\xi>0}\hat{D}_{S,\tilde{\tau}}(q_0,\rho,\eta,\xi)
$$
The proof of \eqref{eq:lower_bound_soft} follows directly from \eqref{eq:usef}.
To prove \eqref{eq:upper_bound_soft}, it suffices to observe that for $C_w>q_0^\star$ and $C_b >\frac{\eta^\star}{\sigma q_0^\star}$, 
$$
\tilde{\phi}_{S}^{(n)} \leq \sup_{\xi>0}\hat{D}_{S,\tilde{\tau}}(q_0^\star,\rho^\star,\eta^\star,\xi)
$$
Noting that \begin{align}&\sup_{\xi>0}\hat{D}_{S,\tilde{\tau}}(q_0^\star,\rho^\star,\eta^\star,\xi)\\
&\asto \sup_{\xi>0}{D}_{S,\tilde{\tau}}(q_0^\star,\rho^\star,\eta^\star,\xi)= \inf_{q_0>0}\min_{0\leq \rho\leq 1}\inf_{\eta\in\mathbb{R}}\sup_{\xi>0}{D}_{S,\tilde{\tau}}(q_0,\rho,\eta,\xi),\end{align} yields \eqref{eq:upper_bound_soft}.
\subsubsection{Proof of the uniqueness of the solution to \eqref{min_max_soft_margin}}
\label{sec:uniqueness}
The solution to \eqref{min_max_soft_margin} should solves the following optimization problem. 
\begin{align}
	&\inf_{\substack {q_0> 0\\ \eta\in\mathbb{R}}} \min_{0\leq\rho\leq 1}\sup_{\xi > 0} \mathcal{D}_{S,\tilde{\tau}} (q_0,\rho,\eta,\xi)\\
	&=\inf_{q_0>0} q_0^2+q_0\min_{0\leq\rho\leq 1}\inf_{\eta\in\mathbb{R}} \sup_{\xi > 0}R_{\tilde{\tau}}(\frac{1}{q_0},\rho,\eta,\xi)
\end{align}
Obviously, the objective function is strictly convex in $q_0$, hence the uniqueness of the solution $q_0^\star$. Assume that at optimum, the maximum over $\xi$ is not  at the limit $\xi\to0$.  Under such a setting, it is easy to check by investigating the derivative of $\mathcal{D}$ with respect to $q_0$  that $\frac{\partial}{\partial \alpha}\left|\right._{q_0\downarrow 0}=-\infty$. Hence $q_0^\star\neq 0$. Assume now that the supremum over $\xi$ is attained at $\xi^\star>0$. 
Considering that $\xi^\star >0$, we can prove through simple calculations that   the Jacobian of the function  $(\rho,\eta)\mapsto \sup_{\xi>0} R_{\tilde{\tau}}(\frac{1}{q_0^\star},\rho,\eta,\xi)$ is positive definite. Hence function  $(\rho,\eta)\mapsto \sup_{\xi>0} R_{\tilde{\tau}}(\frac{1}{q_0^\star},\rho,\eta,\xi)$ is strictly convex, thus showing the uniqueness of $\rho^\star$ and $\eta^\star$ provided that $\xi^\star>0$. 
Assume now that the optimum is at the limit $\xi\to\infty$. From the previous section, this happens only if $\rho^\star=1$, because the function $\xi\mapsto R_{\tilde{\tau}}(\frac{1}{q_0^\star},\rho,\eta,\xi)$ is level-bounded for $0\leq \rho<1$. For all other values of $\rho$, the optimum is attained at $\xi^\star>0$. To prove that for each $\eta\in\mathbb{R}$,  $\rho\mapsto \sup_{\xi>0} R_{\tilde{\tau}}(\frac{1}{q_0^\star},\rho,\eta,\xi)$ is strictly convex, let $\eta\in\mathbb{R}$ and $\rho_1$ and $\rho_2$ be in $[0,1]$ such that $\rho_1\neq \rho_2$. Let $\theta\in(0,1)$. Then, $\rho_{\theta}:=\theta \rho_1+(1-\theta)\rho_2 \in[0,1)$. For any $\xi\in\mathbb{R}$, by strict convexity of function $\rho\mapsto R_{\tilde{\tau}}(\frac{1}{q_0^\star},\rho,\eta,\xi)$, we have:
\begin{align*}
	&R_{\tilde{\tau}}(\frac{1}{q_0^\star},\rho_\theta,\eta,\xi) < \theta R_{\tilde{\tau}}(\frac{1}{q_0^\star},\rho_1,\eta,\xi)+(1-\theta)R_{\tilde{\tau}}(\frac{1}{q_0^\star},\rho_2,\eta,\xi)\\
&\leq \theta \sup_{\xi>0} R_{\tilde{\tau}}(\frac{1}{q_0^\star},\rho_1,\eta,\xi)+(1-\theta)\sup_{\xi>0}R_{\tilde{\tau}}(\frac{1}{q_0^\star},\rho_2,\eta,\xi)
\end{align*}
The above equation applies for any $\xi>0$. It thus applies for $\xi^\star$ such that maximizes $R_{\tilde{\tau}}(\frac{1}{q_0^\star},\rho_\theta,\eta,\xi)$. Such $\xi^\star$ is strictly positive since $\rho_{\theta}\neq 1$.
This proves that   $\rho\mapsto \sup_{\xi>0} R_{\tilde{\tau}}(\frac{1}{q_0^\star},\rho,\eta,\xi)$ is strictly convex for each $\eta\in\mathbb{R}$, hence the uniqueness of $\rho^\star$. In both cases $\rho^\star=1$ or $\rho^\star\neq 1$, it is easy to check that $\eta\mapsto \sup_{\xi>0}R_{\tilde{\tau}}(\frac{1}{q_0^\star},\rho^\star,\eta,\xi)$ is strictly convex in $\eta$. Hence, the uniqueness of $\eta^\star$.   

It remains thus to show that the supremum over $\xi>0$ could not be attained in the limit $\xi\to 0$. Consider extending function ${R}_{\tilde{\tau}}$ to $\xi=0$ by setting $R_{\tilde{\tau}}(\frac{1}{q_0},\rho,\eta,0)=\lim_{\xi\downarrow 0}R_{\tilde{\tau}}(\frac{1}{q_0},\rho,\eta,\xi)=0$. Then,
\begin{equation}
\inf_{q_0>0} \inf_{\eta\in\mathbb{R}}\min_{0\leq\rho\leq 1}\mathcal{D}_{S,\tilde{\tau}}(q_0,\rho,\eta,0)=\inf_{q_0>0} q_0^2=0. 
\label{eq:contradict}
\end{equation}
which shows that in this case that the optimum with respect to $q_0$ is at the limit $q_0\to 0$. On the other hand, it is easy to show that for any $\rho\in[0,1]$, $\eta\in\mathbb{R}$ and $\xi>0$
\begin{align}
\lim_{q_0\downarrow 0}\mathcal{D}_{S,\tilde{\tau}}(q_0,\rho,\eta,\xi)=\frac{\tilde{\tau}\pi_1\delta}{\sqrt{2\pi}}\int_{-\infty}^\infty \frac{1}{\sigma} Dt+\frac{\tilde{\tau}\pi_0\delta}{\sqrt{2\pi}}\int_{-\infty}^\infty \frac{1}{\sigma}\label{eq:contradiction}
\end{align}
 thus raising a contradiction with \eqref{eq:contradict}.  Therefore at optimum the supremum over $\xi$ could not be attained in the limit $\xi\to 0$. 
\subsubsection{Concluding}
The asymptotic analysis of the (AO) problem led to proving that the optimal cost associated with the (AO) converges to
$$
{\phi}_S^{(n)}-\inf_{q_0>0}\inf_{\eta\in\mathbb{R}}\min_{0\leq \rho\leq 1}\sup_{\xi>0} \mathcal{D}_{S,\tilde{\tau}}(q_0,\rho,\eta,\xi)\asto 0. 
$$
Let $\epsilon>0$. Consider the set $\mathcal{S}_{\epsilon}:=\left\{|{\norm{\bf w}}_2-q_0^\star|>\epsilon\right\}$. Let $\phi_{\mathcal{S}_\epsilon}^{(n)}$ be the optimal cost of the (AO) problem when the minimization is constrained on $\mathcal{S}_\epsilon$. The same arguments of our analysis lead to:
$$
\phi_{\mathcal{S}_\epsilon}^{(n)}\asto \inf_{\substack{q_0>0\\ |q_0-q_0^\star|>\epsilon}}\inf_{\eta\in\mathbb{R}}\min_{0\leq \rho\leq 1}\sup_{\xi>0} \mathcal{D}_{S,\tilde{\tau}}(q_0,\rho,\eta,\xi)
$$
Since $q_0^\star$ is the unique minimizer, we thus have:
\begin{align}
	&\inf_{\substack{q_0>0\\ |q_0-q_0^\star|>\epsilon}}\inf_{\eta\in\mathbb{R}}\min_{0\leq \rho\leq 1}\sup_{\xi>0} \mathcal{D}_{S,\tilde{\tau}}(q_0,\rho,\eta,\xi)\\
	&> \inf_{q_0>0}\inf_{\eta\in\mathbb{R}}\min_{0\leq \rho\leq 1}\sup_{\xi>0} \mathcal{D}_{S,\tilde{\tau}}(q_0,\rho,\eta,\xi)
\end{align}
Using Theorem \ref{th:new_CGMT}, we have ${\norm{\bf w}}_2-q_0^\star\asto 0$. Similarly we can leverage the uniqueness of the minimizers $\rho^\star$ and $\eta^\star$ to establish the remaining convergences in Theorem~\ref{th:soft_margin}.  
\subsection{ Main differences with the previous CGMT based analysis}
\label{sec:diff}
When it comes to use the CGMT to analyze the asymptotic behavior of a certain min-max optimization problem, many difficulties are often encountered.
These primarily concern the following situations: 1) The sets over which optimization of the primiary problem is performed are non-compact, while compactness is a key assumption in the proof of the CGMT, 2) The objective in the auxiliary problem is not convex-concave which makes flipping the order of the min-max no longer a permissible operation, 3) The CGMT has thus far been employed to prove convergence in probability which is weaker than almost sure convergence, 4) The objective of the primary problem may not have a finite solution, the use of the CGMT  to characterize such scenarios becomes delicate. All these situations have been more or less encountered in previous works, but either they have been handled in a very complicated fashion or have been partially studied. 
In this work, we develop new tools to handle all these situations. For the reader convenience, we pinpoint for each situation the corresponding new key lemmas or sections  that were used and show that they do not only allow for a much simpler treatment than the approaches pursued in previous works but also they allow to prove stronger results never established before by the CGMT framework.
\begin{itemize}
\item Lack of the compactness assumption: In many situations, the primary problem involves non-compact sets. Non-compactness on the set over which we perform minimization is handled by approaching  the original problem by the limit of a sequence of (PO) problems with compact sets. With each of these (PO) problems, we associate an equivalent (AO) problem and apply the CGMT. If the problem of feasible and the solution is finite, the asymptotic optimal cost of the (AO) should remain the same  for sufficiently large compact sets.   We can thus prove that this stable limit correspond also to the asymptotic optimal cost of the (PO). This approach used in section \ref{sec:hard_margin} to analyze the hard-margin SVM is general and may be applied to any other problem that cannot be easily written in the form of an optimization problems  over compact sets.  
\item The objective in (AO) problems are  not directly amenable to convex-concave problems: If the objective function of the auxiliary problem is not convex-concave in the arguments over which minimization and maximization is performed, we are not allowed in general to flip the order of the min-max, posing a major difficulty towards simplifying the (AO) problem.  
To overcome such an issue, the work in \cite{thrampoulidis-IT} shows that under some conditions, the order of the min-max can be flipped. The provided arguments invoke basically an asymptotic equivalence between the min-max problem and its dual problem (in which the order of the min max is flipped) that they prove using the CGMT. In this work, we handled this issue in a very different way. Particularly, we noted that in many cases, we do not need to flip the order of the min-max, although the final result may appear as if such an operation were actually performed. More specifically, we noted that the optimal variables can be shown to lie in a smaller set in which the objective function of the auxiliary problem is convex-concave. The key Lemma that allows to perform such a treatment is  lemma \ref{lem:prop}. 
\item Proof of almost sure convergence results : The CGMT framework provides lower and upper bounds of the (PO) problem with high probability.  If this lower bound and upper bounds can be made as small as desired, the convergence in probability of the optimal cost of the (PO) problem holds. At first sight, it may appear that the CGMT allows only to prove asymptotic results that hold in probability and not in the almost sure sense. We prove in this work that convergence in the almost sure sense also holds. The main  ingredient is that the probability bounds hold for fixed dimensions and involve a sequence of independent events. The  converse Borel Cantelli Lemma is thus invoked to transfer the comparison of probability bounds into almost sure convergence results. This approach is used in the proof of Theorem \ref{th:new_CGMT}, extending on the original CGMT.  
\item Unboundedness of the objective function: If the constraints of the (PO) are not feasible, the optimization is performed over an empty set and as such the optimal cost is infinite. Characterizing the feasibility condition   of the (PO)  involves identifying the necessary and sufficient condition upon which the  optimal cost of the (PO) is finite. The CGMT allows at first sight to only characterize a sufficient condition of feasibility that guarantees  the boundedness of the optimal cost of the (PO). To the authors' knowledge, the only work that dealt with a similar question using the CGMT concerns solving the phase retrieval problem and has only focused on deriving sufficient conditions for the compactness of the optimal cost of the (PO) while failing to show their necessity \cite{oussama}. In this work, this issue is encountered with the hard-margin SVM, which is not always feasible. It was handled by approaching the original (PO) by the limit of a sequence of (PO), and proving that their optimal cost increase unboundedly. The corresponding techniques can be found in the proof Theorem \ref{th:hardmargin_condition}
 We believe that they may be of individual interest and can target other applications that go beyond the topic of the present work.            
\end{itemize}  

\section{Conclusion}
This paper presents an asymptotically sharp characterization of the performance of the hard-margin and soft-margin SVM. Our analysis builds upon the recently developed CGMT framework, which was mainly used before in the study of high-dimensional regression problems. Considering its use for the analysis of SVM poses  technical challenges, which have been handled through a new promising technical approach. This approach not only allowed for an easier use of the CGMT but also enabled to obtain stronger almost sure convergence results.
 We believe that the developed tools lay the groundwork to facilitate and pave the way towards the use of the CGMT to general optimization based-classifiers such as logistic regression, Adaboost, for which an explicit formulation is not available. 

\appendices
\section{Technical Lemmas}
This appendix gathers some important lemmas that are extensively used when optimizing the auxiliary problem. The following Lemma, whose proof is not complicated, is fundamental to simplify the optimization of the auxiliary problem. As shown above, it allowed in some cases to avoid the necessity of flipping  the order of the min-max when solving  min-max optimization problems.
\begin{lemma}
Let $d_1$ and $d_2$ be two strictly positive integers. Let $X\times Y$ be two non-empty sets in $\mathbb{R}_{d_1}\times \mathbb{R}_{d_2}$. Let $F:X\times Y\to \mathbb{R}$ be a given real-valued  function. Assume there exists $\tilde{X}\subset X$ such that for all ${\bf x}\in X$ there exists $\tilde{\bf x}\in \tilde{X}$ such that:
\begin{equation}
\forall {\bf y}\in Y, \ \ F({\bf x},{\bf y})\geq F(\tilde{\bf x},{\bf y}).
\label{eq:prop}
\end{equation}
Then
$$
\min_{{\bf x}\in X} \max_{{\bf y}\in Y} F({\bf x},{\bf y}) = \min_{\tilde{\bf x}\in \tilde{X}} \max_{{\bf y}\in Y} F(\tilde{\bf x},{\bf y})
$$
\label{lem:prop}
In particular, if $\tilde{\bf X}=\{\tilde{\bf x}\}$, then:
$$
\min_{{\bf x}\in X} \max_{{\bf y}\in Y} F({\bf x},{\bf y}) = \max_{{\bf y}\in Y} F(\tilde{\bf x},{\bf y})
$$
\end{lemma}
\begin{proof}
It is easy to see that $
\displaystyle\min_{{\bf x}\in X} \max_{{\bf y}\in Y} F({\bf x},{\bf y})$ is upper-bounded by:
$$
\min_{{\bf x}\in X} \max_{{\bf y}\in Y} F({\bf x},{\bf y}) \leq \min_{\tilde{\bf x}\in \tilde{X}} \max_{{\bf y}\in Y} F(\tilde{\bf x},{\bf y})
$$
To prove the lower-bound, we will exploit the property described in \eqref{eq:prop}. Let ${\bf x}\in X$ and let $\tilde{\bf x}({\bf x})$ is such that  for all $y\in Y$, 
$$
F({\bf x},{\bf y})\geq F(\tilde{\bf x}({\bf x}),{\bf y})
$$
Hence
$$
\max_{{\bf y}\in Y}F({\bf x},{\bf y})\geq \max_{{\bf y}\in Y} F(\tilde{\bf x}({\bf x}),{\bf y}) \geq \min_{\tilde{\bf x}\in \tilde{X}} \max_{{\bf y}\in Y} F(\tilde{\bf x},{\bf y})
$$
which proves that:
$$
\min_{{\bf x}\in X}\max_{{\bf y}\in Y}F({\bf x},{\bf y})\geq \min_{\tilde{\bf x}\in \tilde{X}} \max_{{\bf y}\in Y} F(\tilde{\bf x},{\bf y})
$$
\end{proof}
\begin{lemma}
\label{lem:equality}
Let ${d}\in\mathbb{N}^\star$. Let ${S}_x$ be a compact non-empty set in $\mathbb{R}^{d}$. Let $f$ and $c$ be two continuous functions over $S_x$ such that the set $\left\{c(x)\leq 0\right\}$ is non-empty. Then:
$$
\min_{\substack{{\bf x}\in{S}_x \\ c({\bf x})\leq 0}}  f({\bf x})\\
=\sup_{\delta \geq 0} \min_{\substack{{\bf x}\in{S}_x \\ c({\bf x})\leq \delta}}  f({\bf x})= \inf_{\delta>0} \min_{\substack{{\bf x}\in{S}_x \\ c({\bf x})\leq -\delta}}  f({\bf x}) 
$$ 
\end{lemma}
\begin{proof}
We will prove only the first equality, the second one following along the same lines. 
Obviously, the following inequality holds true,
$$
\min_{\substack{{\bf x}\in{S}_x \\ c({\bf x})\leq 0}}  f({\bf x})\leq \overline{f}\triangleq \sup_{\delta \geq 0}\min_{\substack{{\bf x}\in{S}_x \\ c({\bf x})\leq \delta}}  f({\bf x})
$$
 for all $\delta \geq 0$, 
$$
\min_{\substack{{\bf x}\in{S}_x \\ c({\bf x})\leq 0}}  f({\bf x})\geq  \min_{\substack{{\bf x}\in{S}_x \\ c({\bf x})\leq \delta}}  f({\bf x})
$$
Hence, 
$$
\min_{\substack{{\bf x}\in{S}_x \\ c({\bf x})\leq 0}}  f({\bf x})\geq \sup_{\delta\geq 0} \min_{\substack{{\bf x}\in{S}_x \\ c({\bf x})\leq \delta}}  f({\bf x})
$$
Let $\overline{f}=\displaystyle\sup_{\delta\geq 0} \min_{\substack{{\bf x}\in{S}_x \\ c({\bf x})\leq \delta}}  f({\bf x})$. Let $\epsilon>0$. There exists $\delta_\epsilon$  and ${\bf x}_\epsilon^\star$ such that $c({\bf x}_\epsilon)\leq \delta_\epsilon$,   $f({\bf x}_\epsilon^\star)=\displaystyle\min_{\substack{{\bf x}\in{S}_x \\ c({\bf x})\leq \delta^\epsilon}}  f({\bf x})$, and
$$
\overline{f}\leq f({\bf x}_\epsilon^\star)+\epsilon
$$
Now for all $m\in\mathbb{N}^\star$ such that $m\geq m_0\triangleq \ceil[\big]{\frac{1}{\delta^\epsilon}} $,  define ${\bf x}_m^\star$ such that $f({\bf x}_m^\star)= \displaystyle\min_{\substack{{\bf x}\in{S}_x \\ c({\bf x})\leq \frac{1}{m}}}  f({\bf x})$. Clearly, $f({\bf x}_m^\star)\geq f({\bf x}_\epsilon^\star)$. Hence, 
\begin{equation}
\overline{f}\leq f({\bf x}_m^\star)+\epsilon\leq \min_{\substack{{\bf x}\in{S}_x \\ c({\bf x})\leq 0}}  f({\bf x})+\epsilon
\label{eq:later}
\end{equation}
Assume that $\overline{f}>  \min_{\substack{{\bf x}\in{S}_x \\ c({\bf x})\leq 0}} f({\bf x})$. Then, taking $\epsilon=\frac{1}{2}\left(\overline{f}-\displaystyle\min_{\substack{{\bf x}\in{S}_x \\ c({\bf x})\leq 0}}f({\bf x})\right)$, \eqref{eq:later} leads to a contradiction.  
\end{proof}
\begin{lemma}
Let $X$ and $Y$ be two convex sets. Let $f:X\times Y\to\mathbb{R}$ be a jointly convex function in $X\times Y$. Assume that $\forall y\in Y$, $\inf_{x\in X}f(x,y)>-\infty$. Then:
$g:y\mapsto \inf_{x\in X} f(x,y)$   is convex in $Y$. 
\label{lem:convexity}
\end{lemma}
\begin{proof}
See \cite{boyd}
\end{proof}
\begin{lemma}
Let ${\bf a}=\left[a_1,\dots,a_n\right]^{T}$ be a vector in $\mathbb{R}^{n\times 1}$ and $\theta$ be a positive scalar.  Then:
$$
	\max_{\substack{{\bf u}\geq 0\\ {\norm{\bf u}}_2=\theta}}  {\bf a}^{T}{\bf u}= \theta\sqrt{\sum_{i=1}^n (a_i)_{+}^2}
$$
\label{lem:opt_norm}
\end{lemma} 
\begin{lemma}
\label{lem:opt}
Let ${\bf a}\in\mathbb{R}^{n\times 1}$. Let $\tilde{\beta}$ and $\tau$ be  positive scalars. Then,
if $\tilde{\beta}=0$,
$$
	\max_{0\leq {\bf u}\leq \tau} {\bf u}^{T}{\bf a}-\tilde{\beta}\norm{{\bf u}}_2=\max_{0\leq {\bf u}\leq \tau} {\bf u}^{T}{\bf a}=\sum_{i=1}^n \tau (a_i)_{+} 
$$
If $\tilde{\beta}\neq 0$, then:
\begin{align}
	&\max_{0\leq{\bf u}\leq \tau} {\bf u}^{T}{\bf a}-\tilde{\beta}\|{\bf u}\|_2=\sup_{\xi\geq 0} \sum_{i=1}^n(a_i\tau-\tau^2\frac{1}{2\xi}) {\bf 1}_{\left\{{a_i\xi}\geq \tau\right\}}\nonumber\\
	&+\sum_{i=1}^n \frac{a_i^2\xi}{2}{\bf 1}_{\left\{0\leq {a_i\xi}< \tau\right\}}-\frac{\tilde{\beta}^2\xi}{2}
\label{eq:formula}
\end{align}
Moreover,   function $\xi\mapsto \sum_{i=1}^n(a_i\tau -\tau^{2}\frac{1}{2\xi}){\bf 1}_{\left\{a_i\xi\geq \tau\right\}}+\sum_{i=1}^n\frac{{a}_i^2\xi}{2}{\bf 1}_{\left\{0\leq a_i\xi \leq \tau\right\}} -\frac{\tilde{\beta}^2\xi}{2}$ is concave in $\xi$ when $\xi\in(0,\infty)$. 
\end{lemma}
\begin{proof}
Using the fact that:
$$
\|{\bf u}\|=\inf_{\chi>0 } \frac{\chi}{2}+\frac{\|{\bf u}\|_2^2}{2\chi}
$$
we obtain:
\begin{align}
&\max_{0\leq{\bf u}\leq \tau} {\bf u}^{T}{\bf a}-\tilde{\beta}\|{\bf u}\|_2=\max_{0\leq{\bf u}\leq \tau}\sup_{\chi\geq 0}  {\bf u}^{T}{\bf a}-\tilde{\beta}\left[\frac{\chi}{2}+\frac{\|{\bf u}\|_2^2}{2\chi}\right]\\
&=\sup_{\chi>0 } \max_{0\leq{\bf u}\leq \tau} \sum_{i=1}^n a_iu_i -\frac{\tilde{\beta}}{2\chi}u_i^2-\tilde{\beta}\frac{\chi}{2}
\label{eq:maxu}
\end{align}
Function ${x}\mapsto a_i x -\frac{\tilde{\beta}}{2\chi}x^2$ is increasing on $(-\infty,\frac{a_i\chi}{\tilde{\beta}})$ and decreasing on  $(\frac{a_i\chi}{\tilde{\beta}},\infty)$ taking its maximum at ${x}^\star=\frac{a_i\chi}{\tilde{\beta}}$. Hence, 
$$
\max_{0\leq x\leq \tau} a_i x -\frac{\tilde{\beta}}{2\chi}x^2 =\begin{cases}
0 & {\textnormal{if}}\ \ \frac{a_i\chi}{\tilde{\beta}}<0\\
a_i\tau - \tau^2\frac{\tilde{\beta}}{2\chi} & \textnormal{if} \ \ 0\leq \tau< \frac{a_i\chi}{\tilde{\beta}}\\
\frac{a_i^2\chi}{2\tilde{\beta}}  & \textnormal{if} \ \  0\leq \frac{a_i\chi}{\tilde{\beta}}\leq \tau\end{cases}
$$
Hence,
	\begin{align*}
		&\max_{0\leq{\bf u}\leq \tau} {\bf u}^{T}{\bf a}-\tilde{\beta}\|{\bf u}\|_2=\sup_{\chi> 0} \sum_{i=1}^n(a_i\tau-\tau^2\frac{\tilde{\beta}}{2\chi}) {\bf 1}_{\left\{\frac{a_i\chi}{\tilde{\beta}}\geq \tau\right\}}\\
		&+\sum_{i=1}^n \frac{a_i^2\chi}{2\tilde{\beta}}{\bf 1}_{\left\{0\leq \frac{a_i\chi}{\tilde{\beta}}< \tau\right\}}-\frac{\tilde{\beta}\chi}{2}
	\end{align*}
Performing the change of variable $\xi:\triangleq \frac{\chi}{\tilde{\beta}}$ yields \eqref{eq:formula}. 

We will now proceed to proving the concavity of function $\varphi_i:\xi\mapsto (a_i\tau -\tau^{2}\frac{1}{2\xi}){\bf 1}_{\left\{a_i\xi\geq \tau\right\}}+\frac{{a}_i^2\xi}{2}{\bf 1}_{\left\{0\leq a_i\xi < \tau\right\}} -\frac{\tilde{\beta}^2\xi}{2}$. To this end, note that: 
$$
\varphi_i(\xi)=\max_{0\leq u_i\leq \tau} a_iu_i-\frac{u_i^2}{2\xi} -\tilde{\beta}^2\frac{\xi}{2}.
$$
Function $(\xi,u_i)\mapsto \frac{u_i^2}{2\xi}$ is jointly convex in $\mathbb{R}_{>0}\times [0,\tau]$ since it is the perspective function of $x\mapsto x^2$. Hence, $(\xi,u_i)$ is jointly concave in  $\mathbb{R}_{>0}\times [0,\tau]$. Using Lemma \ref{lem:convexity}, we thus get that $\xi\mapsto \varphi_i(\xi)$ is concave in $\mathbb{R}_{>0}$

\end{proof}

\bibliographystyle{IEEEtran}
\bibliography{IEEEabrv,myref}

\end{document}